\documentclass[preprint,3p,number]{elsarticle}

\usepackage{JLAMP}

\begin{document}

\begin{frontmatter}

\title{Automated Verification of Reactive and \\ Concurrent Programs by Calculation}

\author{Simon Foster}
\ead{simon.foster@york.ac.uk}

\cortext[cor1]{Corresponding author}

\author{Kangfeng Ye}
\ead{kangfeng.ye@york.ac.uk}

\author{Ana Cavalcanti}
\ead{ana.cavalcanti@york.ac.uk}

\author{Jim Woodcock}
\ead{jim.woodcock@york.ac.uk}

\begin{abstract}
  Reactive programs combine traditional sequential programming constructs with primitives to allow communication with
  other concurrent agents. They are ubiquitous in modern applications, ranging from components systems and web services,
  to cyber-physical systems and autonomous robots. In this paper, we present an algebraic verification strategy for
  concurrent reactive programs, with a large or infinite state space. We define novel operators to characterise
  interactions and state updates, and an associated equational theory. With this we can calculate a reactive program's
  denotational semantics, and thereby facilitate automated proof. Of note is our reasoning support for iterative
  programs with reactive invariants, based on Kleene algebra, and for parallel composition. We illustrate our strategy
  by verifying a reactive buffer. Our laws and strategy are mechanised in Isabelle/UTP, our implementation of Hoare and
  He's Unifying Theories of Programming (UTP) framework, to provide soundness guarantees and practical verification
  support.
\end{abstract}

\end{frontmatter}

\section{Introduction}
\label{sec:intro}
Reactive programming~\cite{Harel1985,Bainomugisha2013} is a paradigm that enables effective description of software
systems that exhibit both internal sequential behaviour and event-driven interaction with a concurrent party. Reactive
programs are ubiquitous in safety-critical systems, and typically have a very large or infinite state space. Though
model checking is an invaluable verification technique, it exhibits inherent limitations with state explosion and
infinite-state systems that can be overcome by supplementing it with theorem proving.

Previously~\cite{Foster17c}, we have shown how \emph{reactive contracts} support an automated verification technique for
reactive programs. Reactive contracts follow the design-by-contract paradigm~\cite{Meyer92}, where programs are
accompanied by pre- and postconditions. Reactive programs are often non-terminating and so we also capture intermediate
behaviours, where the program has not terminated, but is quiescent and offers opportunities to interact. Our contracts
are triples, $\rc{\!P_1\!}{\!P_2\!}{\!P_3\!}$, where $P_1$ is the precondition, $P_3$ the postcondition, and $P_2$ the
pericondition. $P_2$ characterises the quiescent observations in terms of the interaction history, and the events
enabled at that point. Broadly speaking, our contract theory has its roots in the CSP process algebra~\cite{Hoare85},
and its failures-divergences semantic model~\cite{Roscoe2005,Cavalcanti&06}.

Reactive contracts describe communication and state updates, so $P_1$, $P_2$, and $P_3$ can refer to both a trace
history of events and internal program variables. They are, therefore, called ``reactive relations'': like relations
that model sequential programs, they can refer to variables before ($x$) and later ($x'$) in execution, but also the
interaction trace ($\trace$), in both intermediate and final observations.

Verification using contracts employs refinement ($\refinedby$), which allows an implementation to weaken the
precondition, and strengthen both the peri- and postcondition when the precondition holds. We employ the
``programs-as-predicates'' approach~\cite{Hehner93}, where the implementation ($Q$) is itself denoted as a composition
of contracts. Thus, a verification problem, $\rc{\!P_1\!}{\!P_2\!}{\!P_3\!} \refinedby Q$, can be solved by calculating
a program $\rc{\!Q_1\!}{\!Q_2\!}{\!Q_3\!} = Q$, and then discharging three proof obligations: (1) $Q_1 \refinedby P_1$;
(2) $P_2 \refinedby (Q_2 \land P_1)$; and (3) $P_3 \refinedby (Q_3 \land P_1)$. These can be further decomposed, using
relational calculus, to produce verification conditions. In \cite{Foster17c} we employ this strategy in an Isabelle/HOL
tactic.

In summary, in our approach verification of reactive programs reduces to reasoning about reactive relations. For
programs of a significant size, these relations are complex, and so the resulting proof obligations are difficult to
discharge using relational calculus. We need, first, abstract patterns so that the relations can be simplified. This
necessitates bespoke constructs that allow us to concisely formulate the three parts of a contract: assumptions,
quiescent observations, and terminated observations. Second, we need calculational laws to handle iterative programs,
which are only partly handled in our previous work~\cite{Foster17c}.

In this paper we present a novel calculus for description, composition, and simplification of reactive relations in the
stateful failures-divergences model~\cite{Roscoe2005,Hoare&98,Oliveira&09}. We characterise conditions, external
interactions, and state updates. An equational theory allows us to reduce pre-, peri-, and postconditions to
compositions of the new constructs using operators of Kleene algebra~\cite{Kozen90} (KA) and utilise KA proof
techniques. Our theory is characterised in the Unifying Theories of Programming~\cite{Hoare&98,Cavalcanti&06} (UTP)
framework. For that, we identify a class of UTP theories that induce KAs, and utilise it in the derivation of
calculational laws for iteration. We use our UTP mechanisation, called
Isabelle/UTP~\cite{Foster2020-IsabelleUTP,Foster16a}, to implement an automated verification approach for infinite-state
reactive programs with rich data structures based on our calculus.

Our framework can be applied to a wide spectrum of reactive programming languages with trace-based semantics, including
real-time and hybrid dynamical systems~\cite{He94,Zhan2008,Sherif2010}. A particular focus is languages descended from
CSP~\cite{Hoare85,Roscoe2005}. In this paper, our approach is applied to the \Circus modelling
language~\cite{Woodcock2001-Circus,Oliveira&09} which combines state modelling using Z~\cite{Spivey89} and reactive
primitives from CSP~\cite{Hoare85,Roscoe2005}. An example application is verification of Simulink block diagrams, to
which both \Circus and hybrid CSP~\cite{He94} have been successfully applied~\cite{Cavalcanti2011Simulink,Zhan2017}. More
recently, \Circus and CSP have been used for verification of a formal state-machine based language for robotic controllers
called RoboChart~\cite{Miyazawa2019-RoboChart,Foster18b}.

The paper is structured as follows. \S\ref{sec:prelim} outlines preliminary material, including UTP, its mechanisation
in Isabelle/UTP, and reactive programs. \S\ref{sec:utp-kleene} identifies a class of UTP theories that induce KAs, and
applies this class for calculation of iterative contracts. \S\ref{sec:circus-rc} specialises reactive relations with new
operators to capture stateful failures-divergences, and derives their equational theory. This allows us to automatically
calculate semantics for sequential reactive programs. \S\ref{sec:ext-choice} extends our equational theory with support
for calculating external choices, for programs where the environment has control over a decision. We also develop
healthiness conditions characterising productivity -- a requirement for both algebraic laws of external choice and
iteration.  \S\ref{sec:iter} extends the strategy with while loops and reactive invariants. \S\ref{sec:par} encodes
parallel composition as a reactive design, and further extends the strategy with calculational laws for concurrent
behaviours. With this, we can then calculate semantics for concurrency and communication between reactive
processes. \S\ref{sec:verify} demonstrates the resulting proof strategy in a small verification. \S\ref{sec:concl}
outlines related work and concludes.

All our theorems, definitions, and proofs have been mechanically verified in Isabelle/UTP, and are documented in a
series of technical reports\footnote{For historical reasons, we use the syntax $\ckey{R}_s(P \vdash Q \diamond R)$ in
our mechanisation for a contract $\rc{\!P\!}{\!Q\!}{\!R\!}$. The former builds on Hoare and He's original syntax for the
theory of designs~\cite{Hoare&98}.}~\cite{Foster16a,Foster-KA-UTP,Foster-RDES-UTP,Foster-SFRD-UTP}. Additionally, most
theorems and definitions in the paper are accompanied by a small Isabelle icon (\isalogo). In the electronic version,
each icon is hyperlinked to the corresponding mechanised artefact in our Isabelle/UTP GitHub
repository\footnote{Isabelle/UTP repository: \url{https://github.com/isabelle-utp/utp-main}. An archive containing all
the files for this paper, and instructions on how to load them into Isabelle/HOL, can also be found at
\url{http://doi.org/10.5281/zenodo.3541080}.}.

This paper is an extension of~\cite{Foster18a}. It adds a body of additional theorems in \S\ref{sec:circus-rc} on more
specialised healthiness conditions for stateful-failure reactive relations (Theorem~\ref{thm:ncsp-intro}), calculation
of iterative reactive relations (Theorem~\ref{thm:crel-comp}-\eqref{thm:crc7}), preconditions of reactive contracts
(Theorem~\ref{thm:evwp}), and also extended supporting commentary. Moreover, a substantial new \S\ref{sec:par} extends
the strategy for parallel composition. A number of additional supporting theorems and definitions are also included in
the other sections.

\section{Preliminaries}
\label{sec:prelim}
This section describes background material relevant for the definition of our new calculus.

\subsection{Unifying Theories of Programming}
\label{sec:utp}

UTP~\cite{Hoare&98,Cavalcanti&06} uses the ``programs-as-predicates'' approach to encode denotational semantics and
facilitate reasoning about programs. It uses the alphabetised relational calculus, which combines predicate calculus
operators, such as disjunction ($\lor$), complement ($\neg$), and quantification ($\exists x \!@ P(x)$), with relation
algebra~\cite{Tarski41}, to denote programs as binary relations between initial variables ($x$) and their subsequent
values ($x'$). Here, ``alphabetised'' means that every such relational predicate is accompanied by a set of declarations
of variables to which the predicate can refer. For example, a program fragment, $x := 1 \relsemi x := x + 1$, with two
distinct variables $x : int$ and $y : bool$, can be modelled by the relational predicate $x' = 2 \land y' = y$, with the
alphabet $\alpha = \{x : int, x' : int, y : bool, y' : bool\}$.

In this presentation of the UTP, we first define the set of alphabetised expressions, $[\view, \src]\uexpr$, which is
parametric over $\view$ and $\src$, types that represent the value type and observation space, respectively. The latter
is induced by an alphabet, with a set of typed variable declarations. Expressions are isomorphic to functions
$\src \to \view$, which return a value in $\view$ for a given observation space. Alphabetised predicates are represented
by Boolean expressions, $[\src]\upred \defs [bool, \src]\uexpr$. We denote the set of alphabetised relations by
$[\src_1,\src_2]\urel \defs [\src_1 \times \src_2]\upred$, a predicate over a product space, where $\src_1$ and $\src_2$
are the initial and final observation space, and correspond to the sets of undashed and dashed
variables\footnote{Textbook presentations of UTP~\cite{Hoare&98,Cavalcanti&06} typically use $in\alpha P$ and
  $out\alpha P$ to denote the input and output alphabet. Here, we find it more convenient to invoke parametric sets,
  which is also consistent with our mechanisation.}, called the input and output alphabets. The set of homogeneous
relations $[\src]\hrel \defs [\src, \src]\urel$ has identical input and output alphabets. We often notationally
distinguish predicates over a unitary type and relations over a product type by use of boldface characters; for example,
$\ptrue$ is a predicate and $\true$ is a relation.

For any given $\src_1$ and $\src_2$, $[\src_1,\src_2]\urel$ is partially ordered by refinement $\refinedby$
(refined-by), denoting universally closed reverse implication, where $\false$ refines every relation. In this context,
$S \refinedby P$ means that $P$ is more deterministic that $S$. For example, we have it that
$(x' > 2) \refinedby (x := 3)$, since the specification that $x$ should finally have a value greater than $2$ is
satisfied by assigning $3$ to $x$.

Every operator of a sequential programming language can be denoted using relations in UTP. Relational composition
($P \relsemi Q$) denotes sequential composition, and has the type
$[\src_1, \src_2]\urel \to [\src_2, \src_3]\urel \to [\src_1, \src_3]\urel$, since the output alphabet of the first
relation must match the input alphabet of the second. Sequential composition has identity
$\II \defs (\skey{s}' = \skey{s})$, of type $[\src]\hrel$, where $\skey{s} : \src$ denotes the entire state. We also
define the conditional operator $(\conditional{p}{q}{r}) \defs ((q \land p) \lor (\neg q \land r))$, with
$p, q, r : \upred[\src]$, which selects $p$ or $r$ based on the truth valuation of $q$.

We summarise the algebraic properties of a homogeneous UTP theory of relations in terms of Boolean
quantales~\cite{Moller2006}, a useful algebraic structure for characterising homogeneous relations.
\begin{definition}[Boolean Quantales]
  A Boolean quantale~\cite{Moller2006} is a structure $(S, \le, 0, \cdot, 1)$, where $(S, \le)$ is a complete Boolean
  lattice with least element $0$; $(S, \cdot, 1)$ is a monoid with $0$ as left and right annihilator; and the function
  $\cdot$ distributes over the lattice join from the left and right.
\end{definition}
\begin{theorem} For any $\src$,  $([\src]\hrel, \mathrel{\sqsupseteq}, \false, \relsemi, \II)$ is a Boolean quantale~\cite{Moller2006}, so that:
\begin{enumerate}
  \item $([\src]\hrel, \refinedby)$ is a complete lattice, with infimum $\bigvee$, supremum $\bigwedge$, greatest element
     $\false$, least element $\true$, and weakest (least) fixed-point operator $\mu F$;
  \item $([\src]\hrel, \lor, \false, \land, \true, \neg)$ is a Boolean algebra;
  \item $([\src]\hrel, \relsemi, \II)$ is a monoid with $\false$ as left and right annihilator and $\II$ as identity;
  \item $\relsemi$ distributes over $\bigvee$ from the left and right.
\end{enumerate}
\end{theorem}
\noindent We emphasise that our complete lattice is inverted compared to several conventions~\cite{Kozen97,Moller2006},
which is normal for UTP~\cite{Hoare&98,Cavalcanti&06}. In particular, we often use $\Intchoice_{i \in I}\, P(i)$ to
denote an indexed disjunction over $I$, which intuitively refers to a nondeterministic choice, and likewise
$P \intchoice Q$ to denote $P \lor Q$. As we have mentioned, refinement reduces nondeterminism, which is illustrated by
the following law.

$$\bigsqcap_{i \in I}\, P(i) \, \refinedby \, \bigsqcap_{j \in J}\, P(j) \textnormal{ when } J \subseteq I$$
\noindent In other words, refinement reduces the possible choices that a program is permitted to make. We note that the
partial order $\le$ of the Boolean quantale is $\mathrel{\sqsupseteq}$, and so our lattice operators are inverted: for
example, $\bigvee$ is the infimum with respect to $\refinedby$, and $\mu F$ is the least fixed-point. More general
refinement laws can be found in the work of Back and von Wright~\cite[Chapter~7]{Back1998}.

Relations can be used to denote sequential programming constructs like assignment, and finite and infinite
iteration~\cite{Hoare&98,Armstrong2015}. From these denotations the algebraic laws of programming can be
derived~\cite{Hoare87}, along with operational and axiomatic presentations of the semantics~\cite{Hoare&98}. Moreover,
relations can be enriched to characterise more advanced computational paradigms --- such as object
orientation~\cite{SCS06}, real-time~\cite{Sherif2010}, hybrid computation~\cite{Foster17b}, and
concurrency~\cite{Hoare&98} --- using UTP theories that encode semantic domains.

UTP theories use distinguished observational variables to record observable quantities of the program or operating
environment. By their very nature, such variables are not under the control of the programmer, and instead are governed
by logical invariants called healthiness conditions. For example, we may introduce variables
$time, time' : \real_{\ge 0}$ into $\alpha$ to record the time before and after a real-time program fragment
executed. We can then define a delay construct, $delay(n) \defs time' = time + n \land s' = s$, where $s$ is shorthand
for any variable other than $time$, that advances time whilst leaving all other variables unchanged.

Normally time can only advance, and so a desirable healthiness condition is $time \le time'$, a predicate that any
relation modelling a healthy real-time program should respect. The delay construct $delay(n)$ is an example of a healthy
relation, and $time = 1 \land time' = 0$ is an unhealthy one. We can also prove more general theorems for the other
relational operators: for example, if $P$ and $Q$ are both healthy, then also clearly $P \relsemi Q$ is healthy, by
transitivity of $\le$. Similar closure laws can be proved for other operators, which allows us to characterise the
signature, or syntax, of our UTP theory: the set of function symbols guaranteed to construct healthy programs when the
arguments are healthy.

UTP thus inverts the typical denotational semantic approach of defining an inductive syntax tree, for example using an
algebraic datatype, and then giving it a semantics by a recursive function. It has the significant advantages that we
can (1) further constrain our semantic domain by adding extra healthiness conditions, in a compositional manner
supported by the predicative semantics, and (2) extend the signature with additional syntax when necessary, whilst at
the same time retaining all theorems proved with respect to the existing healthiness conditions and operators. Moreover,
we avoid the need to perform induction over the syntax tree in our proofs.

A UTP theory can be formally characterised as the set of fixed-points of a function
$\healthy{H} : [\src]\hrel \to [\src]\hrel$, that models the healthiness conditions. For example,
$\healthy{HT}(P) \defs (P \land time \le time')$ is an idempotent healthiness function whose fixed-points are those
relations that satisfy $time \le time'$. Any predicate on the observational variables can be encoded as a healthiness
function in this way, and therefore we treat the terms healthiness condition and healthiness function as synonyms. If
$P$ is a fixed-point of $\healthy{H}$, it is said to be $\healthy{H}$-healthy, and the set of healthy relations is
$\theoryset{\healthy{H}} \defs \{ P | \healthy{H}(P) = P \}$.

In UTP, it is desirable that $\healthy{H}$ is idempotent ($\healthy{H}\!\circ\!\healthy{H} = \healthy{H}$) and also
monotonic ($X \refinedby Y \implies \healthy{H}(X) \refinedby \healthy{H}(Y)$). Idempotence ensures that, for any $P$,
$\healthy{H}(P)$ is indeed $\healthy{H}$-healthy, and also means that $\theoryset{\healthy{H}}$ is actually the image of
$\healthy{H}$. Monotonicity additionally ensures, by the Knaster-Tarski theorem, that $\theoryset{\healthy{H}}$ forms a
complete lattice under $\refinedby$. Consequently, there exist strongest and weakest fixed-points operators, which allow
us to reason about both nondeterministic and recursive elements of the UTP theory.

Often, we construct a UTP theory by composition of several healthiness functions,
$\healthy{H}_1 \circ \healthy{H}_2 \cdots \circ \healthy{H}_n$. In this case, we can demonstrate idempotence and
monotonicity of $\healthy{H}$ using the following important theorem:

\begin{theorem} \label{thm:healthy-comp} We assume that $\healthy{H} \defs \healthy{H}_1 \circ \healthy{H}_2 \cdots \circ \healthy{H}_n$. Then,
  $\healthy{H}$ is idempotent provided that (1) each $\healthy{H}_i$, for $i \in 1..n$ is idempotent, and (2) each pair commutes:
  for any $i, j \in 1..n$, such that $i \neq j$,
  $\healthy{H}_i \circ \healthy{H}_j = \healthy{H}_j \circ \healthy{H}_i$. Moreover, $\healthy{H}$ is monotonic,
  provided each $\healthy{H}_i$ is monotonic. \isalink{https://github.com/isabelle-utp/utp-main/blob/1287e8d68ca23fea81bc31064febe3d956db6ee2/utp/utp_healthy.thy\#L125}
\end{theorem}

\noindent Consequently, we can reason about a composite healthiness condition in terms of its components. In this paper,
we use such a UTP theory to characterise concurrent and reactive programs.

\subsection{Isabelle/UTP}
\label{sec:isabelleutp}

Theory engineering and verification using UTP is supported by Isabelle/UTP~\cite{Foster2020-IsabelleUTP,Foster16a},
which provides a shallow embedding of the relational calculus on top of Isabelle/HOL, and various approaches to
automated proof. The foundation of Isabelle/UTP is its model of observations, which utilises
lenses~\cite{Foster09,Foster2020-IsabelleUTP,Foster16a} to model variables as algebraic structures. A lens is a pair of
functions $\lget : \src \to \view$ and $\lput : \src \to \view \to \src$, which are used to query and update a view
($\view$) of a larger observation space ($\src$). We write $X : \view \lto \src$ for a lens $X$ viewing $\view$ in the
source $\src$, and $\lget_X$ and $\lput_X$ for its functions. Typically, $\src$ is characterised by an alphabet of
variables ($\alpha$), and consequently we can safely conflate the observation space and alphabet. We characterise the
behaviour of each lens using three axioms~\cite{Foster09}, which link together the functions.
\begin{definition}[Lens Axioms] A lens $X : \view \lto \src$ satisfies, for any $s : \src$ and $v, v' : \view$, the equations

 $$\lget~(\lput~s~v) = v \qquad \lput~(\lput~s~v')~v = \lput~s~v \qquad \lput~s~(\lget~s) = s$$
\end{definition}

\noindent In this paper, we require that all lenses satisfy these three axioms. We note in passing that these axioms have close
analogues in Back and von Wright's variable calculus~\cite{Back1998}, which predate lenses by several years. There,
$\lget$ is called \textsf{val} and $\lput$ is called \textsf{set}, but they are governed by the same axioms. These
axioms have several models including record types, total functions, and products~\cite{Foster2020-IsabelleUTP,Foster16a}. From them, we can
characterise the laws of assignment and substitution without dependence on a particular state model. Moreover, we
describe semantically when two lenses correspond to different variables, using lens independence~\cite{Foster2020-IsabelleUTP}.

\begin{definition}[Lens Independence] We fix $X : \view_1 \lto \src$ and $Y : \view_2 \lto \src$, and then define: \isalink{https://github.com/isabelle-utp/utp-main/blob/bee1f650ce9a5f2f6faf7e5d76da99d4265cd5f0/optics/Lens_Laws.thy\#L309}
$$X \lindep Y \defs (\forall \,s : \src, u : \view_1, v : \view_2 @ \lput_X(\lput_Y~s~v)~u ~=~ \lput_Y(\lput_X~s~u)~v)$$ 
\end{definition}
\noindent $X$ and $Y$ are independent, written $X \lindep Y$, provided that their $\lput$ functions commute, meaning
that they do not interfere with one another. Lenses can model, not just individual variables, but also sets
thereof. Intuitively, a lens $X : \view \lto \src$ abstractly characterises a $\view$-shaped subregion of a $\src$. The
lens summation operator~\cite{Foster2020-IsabelleUTP}, $X \lplus Y$, allows us to compose two such independent regions.
With it, we can model a set of variables $\{x, y, z\}$ through the summation, $x \lplus y \lplus z$. We also introduce two
special lenses~\cite{Foster2020-IsabelleUTP}:
\begin{itemize}
\item $\lzero : \{\emptyset\} \lto \src$, which for any given $\src$, characterises an empty (point) region; and
  \item $\lone : \src \lto \src$, which characterises the entirety of $\src$.
\end{itemize}

\noindent We can also use lenses to construct a state by combining the view of one state $s_2 : \src$ with the
complement from another state $s_1 : \src$. This is useful for merging of parallel threads that act on disjoint parts of
the state. We define a novel lens override operator to perform this state merge.

\begin{definition} We fix $X : \view \lto \src$ and $s_1, s_2 : \src$, and define
  $\lovrd{s_1}{s_2}{X} \defs \lput_X~s_1~(\lget_X~s_2)$.
  \isalink{https://github.com/isabelle-utp/utp-main/blob/f22f260d20c95337e30ee7bb4d6d16e6fda27af0/optics/Lens_Laws.thy\#L73} \label{def:lovrd}
\end{definition}
\noindent Lens override ($\lovrd{s_1}{s_2}{X}$) extracts the region described by $X$ from $s_2$ and overwrites the corresponding
region in $s_1$, leaving the complement unchanged. This operator obeys a number of useful algebraic laws.
\begin{theorem}[Override Laws] \label{thm:ovrd} \isalink{https://github.com/isabelle-utp/utp-main/blob/07cb0c256a90bc347289b5f5d202781b536fc640/optics/Lens_Order.thy\#L462}
  \begin{align}
    \lovrd{s_1}{s_2}{\lzero} &= s_1 \label{law:ov1} \\
    \lovrd{s_1}{s_2}{\lone} &= s_2 \label{law:ov2} \\
    \lovrd{s}{s}{X} &= s \label{law:ov3} \\ 
    \lovrd{(\lovrd{s_1}{s_2}{X})}{s_3}{Y} &= \lovrd{(\lovrd{s_1}{s_3}{Y})}{s_3}{X} & \text{provided } X \lindep Y \label{law:ov4}
  \end{align}
\end{theorem}
Law \eqref{law:ov1} shows that overriding $s_1$ with $s_2$ using $\lzero$, the empty lens, effectively means that we use
none of $s_2$, and \eqref{law:ov2} is the dual case with the $\lone$ lens. Law \eqref{law:ov3} shows that overriding a
source element is idempotent. Law \eqref{law:ov4} is a kind of commutativity law. In the term
$\lovrd{\lovrd{s_1}{s_2}{X}}{s_3}{Y}$ we are constructing a composite source from the $X$ region of $s_2$, the $Y$
region of $s_3$, and the remainder from $s_1$, with the assumption that $X$ and $Y$ are independent. The law shows that
we can, in this case, commute the order in which we apply $s_2$ and $s_3$.

We can also relate lenses using the sublens preorder~\cite{Foster2020-IsabelleUTP}, $X \lsubseteq Y$, which requires that the view of
$X$ is contained within the view of $Y$. For example, $X \lsubseteq X \lplus Y$ -- the order is analogous to a subset
relation for variable sets: $\{x, y\} \lsubseteq \{x, y, z\}$. Moreover, $\lzero \lsubseteq X$ and $X \lsubseteq \lone$,
as these are the smallest and largest lenses.

With lenses, we can also construct substitutions, which are modelled as functions $\sigma : \src \to \src$. They are
used in Isabelle/UTP to unify variable substitutions, state updates, assignments, and evaluation contexts, also
following the pattern given by Back and von Wright~\cite{Back1998}. We can construct substitutions
$\substmap{x_1 \mapsto e_1, x_2 \mapsto e_2, \cdots, x_n \mapsto e_n}$, which assign an expression
$e_i : [\view_i, \src]\uexpr$ to each lens $x_i : \view_i \lto \src$ with a matching view type. Each expression can
refer to the previous values of the variables, and variables not mentioned retain their current value. A substitution
$\sigma : \src \to \src$ can be applied to an expression $e : [\view, \src]\uexpr$ using the operator
$\substapp{\sigma}{e} \defs e \circ \sigma$, which precomposes the characteristic function of $e$ with the substitution
function. We can then define $e[k/x] \defs \substapp{\substmap{x \mapsto k}}{e}$ to obtain the classical substitution
operator. It obeys similar laws to syntactic substitution, though it is a semantic operator~\cite{Foster2020-IsabelleUTP}.

This substitution constructor is syntactic sugar for a more general update operator
$$\sigma(x \mapsto e) \defs (\lambda s @ \lput_x ~ (\sigma(s)) ~ (e(s)))$$ which updates the value of lens $x$ to
expression $e$. We can perform several updates using the shorthand
$$\sigma(x_1 \mapsto e_1, x_2 \mapsto e_2, \cdots, x_n \mapsto e_n) = \sigma(x_1 \mapsto e_1)(x_2 \mapsto e_2)\cdots(x_n \mapsto e_n)$$ and moreover
$\substmap{x_1 \mapsto e_1, x_2 \mapsto e_2, \cdots} = id(x_1 \mapsto e_1)(x_2 \mapsto e_2) \cdots$, where
$id : \src \to \src$ is the identity substitution. Substitution update obeys several useful laws.

\begin{theorem}[Substitutions] We fix $x : \view \lto \src$, $y : \mathcal{W} \lto \src$, $e : [\view, \src]\uexpr$,
  $f : [\mathcal{W}, \src]\uexpr$, $e_i : [\view_i, \src]$ for $1 \le i \le n$, and
  $op : [\view_1, \src]\uexpr \to [\view_2, \src]\uexpr \cdots \to [\view, \src]\uexpr$ and then prove the following
  laws:
  \isalink{https://github.com/isabelle-utp/utp-main/blob/bee1f650ce9a5f2f6faf7e5d76da99d4265cd5f0/utp/utp_subst.thy}
  \begin{align}
    \sigma(x \mapsto x) &= \sigma \label{law:SB0} \\
    \sigma(x \mapsto e, y \mapsto f) &= \sigma(y \mapsto f, x \mapsto e) & \textnormal{if } x \lindep y \label{law:SB1} \\
    \sigma(x \mapsto e, y \mapsto f) &= \sigma(y \mapsto f) & \textnormal{if } x \lsubseteq y \label{law:SB2} \\
    \substapp{\sigma}{(op~e_1 \cdots e_n)} &= op~(\substapp{\sigma}{e_1}) \cdots (\substapp{\sigma}{e_n}) \label{law:SB4} \\
    \substapp{\sigma(x \mapsto e)}{x} &= e \label{law:SB5}
  \end{align}
\end{theorem}
\noindent An update of a variable to itself has no effect \eqref{law:SB0}. We can commute two updates provided the
variables are independent \eqref{law:SB1}. An update to $y$ overrides one to $x$ when $x$ is a narrower lens than $y$,
or is equivalent \eqref{law:SB2}. Substitution application distributes through applied operator symbols $op$ \eqref{law:SB4},
and replaces variables with their assigned value \eqref{law:SB5}. These laws provide the foundation for modelling state
in a variety of works. In this paper, lenses are valuable in characterising concurrent state updates, as
demonstrated in \S\ref{sec:par}.

\subsection{Reactive Programs}
\label{sec:rea-prog}

Whilst sequential programs determine the relationship between an initial and final state, reactive programs also pause
during execution to interact with the environment. For example, the CSP~\cite{Hoare85,Cavalcanti&06} and
\Circus~\cite{Woodcock2001-Circus,Oliveira&09} languages can model networks of concurrent processes that communicate
using shared channels. Reactive behaviour is described using primitives such as event prefix $a\!\then\!P$, which awaits
event $a$ and then enables $P$; conditional guard, $b \guard P$, which enables $P$ when $b$ is true; external choice
$P\!\extchoice\!Q$, where the environment resolves the choice by communicating an initial event of $P$ or $Q$; and
iteration $\ckey{while}~b~\ckey{do}~P$.  Channels can carry data, and so events can take the form of an input ($c?x$) or
output ($c!v$). \Circus processes also have local state variables that can be assigned ($x := e$).

We exemplify the \Circus notation with the program for an unbounded buffer.
\begin{example} \label{ex:buffer} In the $Buffer$ process below, variable $bf : \seq \nat$ is a finite
    sequence of natural numbers\footnote{In Isabelle/UTP, we model sequences using the HOL parametric type
    $[A]list$, which represents inductive lists.} that records the elements, and channels $inp(n : \nat)$ and
  $outp(n : \nat)$ represent inputs and outputs.
  $$Buffer ~\defs~ bf := \langle\rangle \relsemi \left(
  \begin{array}{l}
    \ckey{while}~~true~~\ckey{do} \\
    ~~
    \begin{array}{l}
      inp?v \then bf := bf \cat \langle v \rangle \\
      \extchoice (\#bf > 0) \guard out!(head(bf)) \then bf := tail(bf)
    \end{array} 
  \end{array} \right)$$

\noindent Here, $xs \cat ys$ denotes sequence concatenation~\cite{Spivey89}, and $\langle x, y, z, \cdots \rangle$
denotes an enumerated sequence. Variable $bf$ is set to the empty sequence $\langle\rangle$, and then a non-terminating
loop describes the main behaviour. Its body repeatedly allows the environment to either provide a value $v$ over $inp$,
followed by which $bf$ is extended, or else, if the buffer is non-empty, receive the value at the head, and then $bf$ is
contracted. \qed
\end{example}
\Circus has previously been given both a denotational~\cite{Oliveira&09} and an operational
semantics~\cite{Woodcock2005-OpSemCircus}, which are linked in the UTP framework. Here, we build on these previous
results and capture the axiomatic semantics for reactive programs using reactive contracts~\cite{Foster17c}. Reactive
contracts can be used both to specify requirements for reactive programs, under certain assumptions, and also to assign
denotational semantics to each operator of a reactive programming language. The denotational semantics symbolically
encodes the possible transitions a reactive program can exhibit. We can therefore use a theorem prover to reason about a
reactive program with a very large or infinite state space. As an example application, we have used them for
  verifying state-machine diagrams in the RoboChart language~\cite{Foster18b}.

\begin{table}
 \arraycolsep=1.4pt
 \def\arraystretch{1.3}
 $$\begin{array}{rlc|@{\hspace{1ex}}l}
     \multicolumn{2}{c}{\textbf{Observational Variable}} && \multicolumn{1}{c}{\textbf{Description}} \\ \hline
     ok, ok' &: \mathbb{B} && \text{Flags whether predecessor or current action has diverged.} \\
     wait, wait' &: \mathbb{B} && \text{Flags whether predecessor or current action is quiescent.} \\
     tr, tr', \trace &: \seq \textit{Event} && \text{The trace, before, after, and during execution.} \\
     \state, \state' &: \Sigma && \text{The state, before and after execution.} \\
     \refu'          &: \power(\textit{Event}) && \text{The set of events before refused at a quiescent point.}
   \end{array}$$

  \vspace{-1ex}
  \caption{Overview of Reactive Design Observational Variables}  
  \label{tab:reaobs}
\end{table}
 
Reactive contracts are built with the following constructor, which is part of our UTP theory's signature:
$$\rc{P_1(\trace, \state, r)}{P_2(\trace, \state, r, r')}{P_3(\trace, \state, \state', r, r')}$$
$P_1$ is called the precondition, $P_2$ is the pericondition, and $P_3$ is the postcondition. The notation $P_i(x,y,z)$
indicates that relation $P_i$ may refer only to $x$, $y$, and $z$ explicitly; any number of variables may be
indicated. The variables are modelled as lenses, but for brevity we omit this technicality. Variable $\trace$ refers to
the trace, which is modelled using a trace algebra~\cite{Foster17b}, and $\state, \state' : \Sigma$ to the state, for
state space $\Sigma$. Different to the basic relational program model, we follow the pattern of encapsulating
  all state variables under $\state$ to explicitly distinguish them from observational
  variables~\cite{Sherif2010,BGW09}. Traces are equipped with operators for the empty trace $\snil$, concatenation
$tt_1 \cat tt_2$, prefix $tt_1 \le tt_2$, and difference $tt_1 - tt_2$, which removes a prefix $tt_2$ from
$tt_1$. Reactive contracts have an extensible alphabet, and can encode additional semantic data, such as refusals, using
extension variables $r, r'$, which are placeholders for additional observational variables.

$P_{1-3}$ are reactive relations~\cite{Foster17c}: a specialised form of homogeneous UTP relation with
  information about the trace history and state. The different combinations of variables permitted by these relations
  are constrained using healthiness conditions. These three relations respectively encode, (1) the precondition in
terms of the initial state and permissible traces; (2) the pericondition with possible intermediate interactions with
respect to an initial state; and (3) the postcondition characterising possible final states should the program
  terminate. Pericondition $P_2$ and postcondition $P_3$ are both within the ``guarantee'' part of the underlying
design contract, and so can be strengthened by refinement; see \cite{Foster17c} for details. $P_2$ does not refer to
intermediate state variables since they are concealed when a program is quiescent. We sometimes abbreviate
$\rc{\truer}{P_2}{P_3}$, a contract with a true precondition, with the notation $\rcs{P_2}{P_3}$. Our precondition
corresponds to the ``assume'' part of a contract. Reactive contracts lie with the greater field of assume-guarantee
conditions~\cite{Benveniste2007,Benvenuti2008,Vincentelli2012}; a detailed comparison can be found in~\cite{Foster17c}.

In this paper, traces are modelled as finite sequences, $\trace : \seq \textit{Event}$, for some set of events given by
$Event$, though other models are also admitted~\cite{Foster17b}. Events can be parametric, written $a.x$, where $a$ is a
channel and $x$ is the data. Our theory provides an extensible denotational semantic model for reactive and concurrent
languages. To exemplify, we consider the semantics of the skip, event, and assignment actions from \Circus, which
require that we add variable $\refu' : \power(\textit{Event})$ to record refusals, which instantiates the extension
variable $r'$.
\begin{definition}[Skip Action, Terminated Event Prefix, and Assignment] \label{def:skipprefas}
\begin{align*}
  \Skip ~~\defs~~~ & \rc{\truer}{\false}{\trace = \snil \land \state' = \state} \\
  \doact{a} ~~\defs~~~& \rc{\truer}{\trace = \langle\rangle \land a \notin \refu'}{\trace = \langle a \rangle \land \state' = \state} & \text{for } a : [\textit{Event}, \Sigma]\uexpr \\
  x := e ~~\defs~~~& \rc{\truer}{\false}{\state' = \state(x \mapsto e) \land \trace =\langle\rangle} & \text{for } x : \view \lto \Sigma \text { and } e : [\view, \Sigma]\uexpr
\end{align*}
\end{definition}
\noindent
Each of these contracts specifies the possible behaviours that can be observed in the reactive program. $\Skip$ is an
action that cannot diverge, and immediately terminates leaving the state unchanged. Its precondition is $\truer$, the
universal reactive relation (defined below), since it is always satisfied. The pericondition is $\false$ because there
are no quiescent behaviours. In the postcondition, we define that no events occur ($\trace = \snil$), and the state is
left unchanged ($\state' = \state$). The event action ($\doact{a}$) also has a true precondition. In the pericondition,
we specify that in an intermediate state no events have occurred, but $a$ is not being refused -- intuitively this means
that the program is waiting to engage in the $a$ event. In the postcondition, we specify that the trace is extended by
$a$, since it has now happened, and the state is unchanged. With this we can define the \Circus event prefix:
$a \then P \defs \doact{a} \relsemi P$.  Assignment also has a true precondition, and a false pericondition since it terminates without
interaction. The postcondition specifies the updates to the state, and leaves the trace unchanged. This
  definition of assignment is naturally more expressive than the relational assignment (\S\ref{sec:utp}) since it also
  handles observational variables like $\trace$.

As mentioned, reactive contracts can also be used as a specification mechanism. For example, we can define the following
contract for deadlock-freedom.

\begin{example}[Deadlock-freedom Contract] $\ckey{CDF} \defs \textstyle\rc{\truer}{\exists e @ e \notin \refu'}{\truer}$ \isalink{https://github.com/isabelle-utp/utp-main/blob/90ec1d65d63e91a69fbfeeafe69bd7d67f753a47/theories/sf_rdes/utp_sfrd_fdsem.thy\#L425}
\end{example}

\noindent $\ckey{CDF}$ requires that every intermediate observation must exhibit at least one enabled event $e$, that
is, one event $e$ is not being refused -- that is what deadlock-freedom means. The pre- and postcondition do not specify
any particular behaviours, since we are only concerned with quiescent observations. Any reactive program that refines
$\ckey{CDF}$ must always have an enabled transition. For example, it is the case that $\ckey{CDF} \refinedby \doact{a}$.
This can be formally demonstrated using the contract refinement theorem below (Theorem~\ref{thm:rdesrefine}). First
though, we give an overview of the encoding of reactive contracts in UTP.

Following the UTP approach, the constructor $\rc{P_1}{P_2}{P_3}$ is really syntactic sugar for a complex
relation~\cite{Foster17c} that is defined using constructs from the UTP theories of reactive processes and
designs. Consequently, contracts can be composed using the UTP relational operators. Reactive relations and contracts
are characterised by healthiness conditions $\healthy{RR}$ and $\healthy{NSRD}$, respectively, which we have previously
described~\cite{Foster17c}, and reproduce in Table~\ref{tab:reahcond}. They are all both idempotent and
continuous~\cite{Foster17c}. The observational variables include $ok$ and $wait$, which are used to distinguish normal
from divergent behaviour, and quiescent from terminating behaviour, respectively. A summary of all the observational
variables is shown in Table~\ref{tab:reaobs}. A reactive contract is then defined as below.

\begin{definition}
  $\rc{P_1}{P_2}{P_3} \defs \healthy{R1} \circ \healthy{R2} \circ \healthy{R3}_{\!h}(ok \land P_1 \implies ok' \land (\conditional{P_2}{wait'}{P_3}))$ \isalink{https://github.com/isabelle-utp/utp-main/blob/f22f260d20c95337e30ee7bb4d6d16e6fda27af0/theories/rea_designs/utp_rdes_triples.thy}
\end{definition}

\noindent This first applies the reactive healthiness conditions, $\healthy{R1}$, $\healthy{R2}$, and
$\healthy{R3}_{\!h}$~\cite{Foster17c}. It then requires that if the predecessor has not diverged ($ok$), and the precondition holds
($P_1$), then the contract does not diverge ($ok'$). There are then two possibilities: either the contract is quiescent
($wait'$), in which case $P_2$ holds, otherwise it is terminating ($\neg wait'$), in which case $P_3$ holds.
$\healthy{NSRD}$ specialises the theory of reactive designs~\cite{Cavalcanti&06,Oliveira&09} to \emph{normal stateful
  reactive designs}~\cite{Foster17c}. This version of reactive designs imposes the requirement that $\state'$ cannot be
referenced in the pericondition, as we assume that quiescent observations do not reveal the state.

Reactive relations characterise the inner elements of a reactive contract, namely the pre-, peri-, and
postconditions. Using healthiness conditions called $\healthy{R1}$ and $\healthy{R2}$, $\healthy{RR}$ ensures that every
observation describes a well-formed trace ($\trace$), and furthermore does not depend on $ok$ or $wait$, as these are
only required by the reactive contract infrastructure. Technically, $\trace$ is not a relational variable, but a special
variable $\trace \defs tr' - tr$ where $tr, tr'$, as usual in UTP, encode the trace relationally~\cite{Hoare&98}, under
the assumption that $\healthy{RR}$ is satisfied. Nevertheless, due to our previous results~\cite{Foster16a,Foster17b},
$\trace$ can be treated as a variable, and it is more intuitive to do so. We treat $tr$ and $tr'$ as semantic machinery
that is concealed in $\trace$, which represents the actual trace.

Preconditions of a reactive contract are elements $\theoryset{\healthy{RC}}$, which specialises $\healthy{RR}$ by
requiring that only the initial state ($\state$) is referenced, and that the trace is prefix closed. The intuition here
is that when a trace violates the precondition of a contract, then any extension of this trace must also violate it,
similar to how the set of divergences in CSP is extension closed~\cite{Brookes1984}. By duality, if a trace satisfies
the precondition, then any prefix of the trace must also satisfy the precondition, and hence the precondition is prefix
closed with respect to the trace. The basic reactive relational operators are defined below.

\begin{definition}[Reactive Relational Operations] \isalink{https://github.com/isabelle-utp/utp-main/blob/9bbb6e407a4224591b2370605186ca0a45b718f9/theories/reactive/utp_rea_rel.thy\#L99}
  $$\truer \defs \healthy{R1}(\true) \qquad (\negr P) \defs \healthy{R1}(\neg P) \qquad P \rimplies Q \defs (\negr P \lor Q)$$
\end{definition}

\noindent The theory of reactive relations forms a Boolean algebra, but we have to redefine $\true$, $\neg$, and
$\implies$ as these are not reactive relations. The relational $\true$ is not $\healthy{RR}$ healthy, since it permits
any combination of $tr$ and $tr'$, and so we define $\truer$ to be the least reactive relation. We also need a bespoke
complement, $(\negr P)$, because $\theoryset{\healthy{RR}}$ is similarly not closed under $\neg$. So, after taking the
negation, we need to apply $\healthy{R1}$ to obtain a healthy relation. We also redefine implication for the same
reasons ($P \rimplies Q$). We do not need to redefine $\false$ because, unlike $\true$, it is already
$\healthy{RR}$-healthy, and the same follows for the other logical connectives. We then have proved the following
theorem~\cite{Foster17c}.

\begin{theorem} $(\theoryset{\healthy{RR}}, \lor, \false, \land, \truer, \negr)$ forms a Boolean algebra \isalink{https://github.com/isabelle-utp/utp-main/blob/9bbb6e407a4224591b2370605186ca0a45b718f9/theories/reactive/utp_rea_rel.thy\#L574} \end{theorem}

\noindent Both $\theoryset{\healthy{RR}}$ and $\theoryset{\healthy{NSRD}}$ are closed under sequential composition, and
have units $\IIr$ and $\IIsrd$, respectively, which are defined below

\begin{definition}[Reactive Relation and Reactive Contract Identities] \label{def:rrel-id}
  \isalink{https://github.com/isabelle-utp/utp-main/blob/f22f260d20c95337e30ee7bb4d6d16e6fda27af0/theories/reactive/utp_rea_rel.thy\#L119}
  $$\IIr \defs (tr' = tr \land \state' = \state \land r' = r) \qquad
    \IIsrd \defs \conditional{(\conditional{(\exists st @ \II)}{wait}{\II})}{ok}{(tr \le tr')}$$
\end{definition}
  
\noindent $\IIr$ is more complex that the basic identity $\II$. It requires that $tr$, $\state$, and $r$ are unchanged,
but leaves the other variables $ok$ and $wait$ unconstrained. We note that $\IIsrd$ and $\Skip$ are different
operators, as the latter does not restrict $\refu$ in the pericondition~\cite{Foster17c}. Both UTP theories also form
complete lattices under $\refinedby$, with top elements $\false$ and
$\Miracle = \rc{\!\truer\!}{\!\false\!}{\!\false\!}$, respectively. $\Miracle$ is not a reactive program, but
  denotes a miraculous or infeasible specification. $\Chaos = \rc{\!\false\!}{\!\false\!}{\!\false\!}$, the least
determinisitic contract, is the bottom of the reactive contract lattice. Any action refines $\Chaos$, and it therefore
allows us to denote unspecified or unpredictable behaviour, with the possibility of both termination and
non-termination. We define the reactive conditional operator $\rconditional{P}{b}{Q}$, which specialises the
relational conditional operator ($\conditional{P}{b}{Q}$), such that $P$ and $Q$ are reactive relations or contracts,
and $b : [\Sigma]\upred$ is a condition on state variables only.

\begin{table}
 \arraycolsep=1.4pt
 \def\arraystretch{1.3}
 $$\begin{array}{rlc|@{\hspace{1ex}}l}
      \multicolumn{2}{c}{\textbf{Healthiness Condition}} && \multicolumn{1}{c}{\textbf{Description}} \\ \hline
      \healthy{R1}(P)   &\defs P \land tr \le tr' && \text{The trace monotonically increases} \\
      \healthy{R2}(P)   &\defs \conditional{P[\langle\rangle, tr'-tr/tr, tr']}{tr \le tr'}{P} && \text{The trace extension is independent of the history} \\
      \healthy{R3}_{\!h}(P) &\defs \conditional{\IIsrd}{wait}{P} && \text{When a predecessor is quiescent behave as $\IIsrd$} \\
      \healthy{RR}(P)   &\defs (\exists (ok, ok', wait, wait') @ \healthy{R1}(\healthy{R2}(P)))&& \text{Reactive Relations: no references to $ok$ and $wait$} \\
      \healthy{RC}(P)   &\defs \healthy{R1}(\healthy{RR}(P) \relsemi tr' \le tr) && \hspace{-1.3ex}\text{\begin{tabular}{l} Reactive Conditions: Reactive Relations where trace \\[-1ex] is prefix closed and there are no references to $\state'$ \end{tabular} \hspace{-1ex}}
      \\
      \healthy{SRD1}(P) &\defs (ok \rimplies P) && \text{Observations are only possible without divergence} \\
      \healthy{SRD3}(P) &\defs (P \relsemi \IIsrd) && \text{Reactive skip is a right unit for $\relsemi$} \\
      \healthy{NSRD}    &\defs \healthy{SRD3} \circ \healthy{SRD1} \circ \healthy{R3}_{\!h} \circ \healthy{R2} \circ \healthy{R1} && \text{Normal Stateful Reactive Designs}
   \end{array}$$

  \vspace{-1ex}
  \caption{Overview of Reactive Design Healthiness Conditions}

  \label{tab:reahcond}

\end{table}

Verification can be facilitated through refinement $\rc{\!P_1\!}{\!P_2\!}{\!P_3\!} \refinedby Q$, where the required
property is specified as an explicit contract triple, and the program $Q$ is an $\healthy{NSRD}$ relation. Contract
refinement allows the precondition to be weakened, and the peri- and postcondition both to be
strengthened~\cite{Foster17c}.
\begin{theorem}[Reactive Design Refinement] \label{thm:rdesrefine} \isalink{https://github.com/isabelle-utp/utp-main/blob/07cb0c256a90bc347289b5f5d202781b536fc640/theories/rea_designs/utp_rdes_triples.thy\#L828}
  \vspace{.5ex}

  \noindent $\rc{P_1}{P_2}{P_3} \refinedby \rc{Q_1}{Q_2}{Q_3}$ if, and only if, $Q_1 \refinedby P_1$,
  $P_2 \refinedby (Q_2 \land P_1)$, and $P_3 \refinedby (Q_3 \land P_1)$.
\end{theorem}
\noindent Thus, if the contract of the reactive program $Q$ can be calculated to be $\rc{\!Q_1\!}{\!Q_2\!}{\!Q_3\!}$,
then refinement follows by three proof obligations: (1) $Q_1 \refinedby P_1$; (2) $P_2 \refinedby (Q_2 \land P_1)$; and
(3) $P_3 \refinedby (Q_3 \land P_1)$. In words, the precondition may be weakened, and both the peri- and postcondition
may be strengthened, assuming the precondition $P_1$ holds. As usual, refinement can remove choices, making a contract
more deterministic. A consequence is that a non-terminating contract, with postcondition $\false$, can refine a
terminating contract. Indeed we have that for any $P$, $P \refinedby \Miracle$. We can avoid refinement by miraculous
behaviour by adding feasibility healthiness conditions~\cite{Hoare&98,Cavalcanti&06}.

In addition to feasibility, refinement does not guarantee to preserve other properties, such as prefix closure
of the trace, which is often needed for languages such as CSP~\cite{Roscoe2005}. In this case, it is necessary to
check these properties of the refined process, or ensure that the process is only constructed of operators that
preserve prefix closure, as is the case for CSP. Either way, this check can be conducted separately to the refinement,
possibly using a type system, though this is not a concern for this paper.

Contracts can be composed using relational calculus. The following identities~\cite{Foster17c,Foster-RDES-UTP} show how
this entails composition of the underlying pre-, peri-, and postconditions for $\bigsqcap$ and $\relsemi$, and also
demonstrate closure of reactive contracts under these operators.

\begin{theorem}[Reactive Contract Composition] \label{thm:rc-comp} \isalink{https://github.com/isabelle-utp/utp-main/blob/07cb0c256a90bc347289b5f5d202781b536fc640/theories/rea_designs/utp_rdes_normal.thy\#L416}
\begin{align}
  \textstyle \bigsqcap_{i\in I} \, \rc{\!P_1(i)\!}{\!P_2(i)\!}{\!P_3(i)\!} &= \textstyle \rc{\bigwedge_{i \in I} P_1(i)}{\bigvee_{i \in I} P_2(i)}{\bigvee_{i \in I} P_3(i)} \label{thm:rc-choice} \\[.1ex]
 \rconditional{\rc{P_1}{P_2}{P_3}}{b}{\rc{Q_1}{Q_2}{Q_3}} &= \rc{\rconditional{P_1}{b}{Q_1}}{\rconditional{P_2}{b}{Q_2}}{\rconditional{P_3}{b}{Q_3}} \\[.1ex]
  \rc{\!P_1\!}{\!P_2\!}{\!\!P_3\!} \relsemi \rc{\!Q_1\!}{\!Q_2\!}{\!\!Q_3\!}&= \rc{\!P_1\!\land\!(P_3\!\wpR\!Q_1)\!}{P_2\!\lor\!(P_3\!\relsemi\!Q_2)\!}{\!P_3\!\relsemi\!Q_3} \label{thm:rc-seq} \\[.1ex]
  \rcs{P_2}{P_3} \relsemi \rcs{Q_2}{Q_3} &= \rcs{P_2\!\lor\!(P_3\!\relsemi\!Q_2)\!}{\!P_3\!\relsemi\!Q_3} \label{thm:rc-seqs}
\end{align}
\end{theorem}
\noindent Nondeterministic choice requires all preconditions, and asserts that one of the peri- and postcondition pairs
hold. Conditional ($\rconditional{P}{b}{Q}$) distributes through a reactive contract. For sequential composition, the
precondition assumes that $P_1$ holds, and that $P_3$ does not violate $Q_1$. The latter is formulated using a reactive
weakest liberal precondition.

\begin{definition}$P \wpR Q ~\defs~  \neg_r\, (P \relsemi \neg_r\, Q)$ where $P \is \healthy{RR}$ and $Q \is \healthy{RC}$ \label{def:wpR} \isalink{https://github.com/isabelle-utp/utp-main/blob/90ec1d65d63e91a69fbfeeafe69bd7d67f753a47/theories/reactive/utp_rea_wp.thy\#L10}
\end{definition}

\noindent Intuitively, $P \wpR Q$ is the weakest reactive condition such that when $P$ terminates, it satisfies $Q$. It obeys
standard predicate transformer laws~\cite{Dijkstra75,Foster17c} such as:
$$\textstyle(\bigvee_{i \in I}\,P(i))\!\wpR\!R = \bigwedge_{i \in I} \, (P(i)\!\wpR\!R) \qquad (P\!\relsemi\!Q)\!\wpR\!R = P\!\wpR\!(Q\!\wpR\!R) \qquad P \wpR \truer = \truer$$
In the pericondition of Theorem~\ref{thm:rc-comp}-\eqref{thm:rc-seq}, it is specified that an intermediate observation
is either of the first contract ($P_2$), or else it terminated ($P_3$) and then following we have an intermediate
observation of the second contract ($Q_2$). In the postcondition, the observation specified is for when the contracts
have both terminated ($P_3 \relsemi Q_3$). The final law, Theorem~\ref{thm:rc-comp}-\eqref{thm:rc-seqs}, is a simpler
case of the previous law. If both preconditions are true, then since $P_2 \wpR \truer$ reduces to $\truer$, the overall
precondition is also $\truer$.

With these and related theorems~\cite{Foster17c}, we can calculate contracts of reactive programs. Verification, then,
can be performed by proving refinement between two reactive contracts, a strategy we have mechanised in the Isabelle/UTP
tactics \textsf{rdes-refine} and \textsf{rdes-eq}~\cite{Foster17c}. The question remains, though, of how to reason about
the underlying compositions of reactive relations for the \text{pre-}, peri-, and postconditions. As an example, we
consider the action $(a\!\then\!\Skip) \relsemi x\!:=\!v$. To reason about its postcondition, we must simplify
$(\trace = \langle a \rangle \land \state' = \state) \relsemi (\state' = \state(x \mapsto v) \land \trace =
\langle\rangle)$.
To simplify its precondition, we also need to consider reactive weakest preconditions. Without such simplifications,
reactive relations can grow very quickly and hamper proof. Of particular importance is the handling of iterative and
parallel reactive relations. We address these issues in this paper.

\section{Linking UTP and Kleene Algebra}
\label{sec:utp-kleene}
In this section, we characterise properties of a UTP theory sufficient to identify a Kleene Algebra~\cite{Kozen90}, and
use this to obtain theorems for iterative contracts. The results in this section apply, not only to stateful-failure
reactive designs, but the larger class of reactive designs (\healthy{NSRD}) as well. Consequently, the theorems can be
applied in the context of other trace models~\cite{Foster17b}.

Kleene Algebras (KA) characterise sequential and iterative behaviour in nondeterministic programs using a signature
$(K, +, 0, \cdot, 1, {}\star)$, where $+$ is a choice operator with unit $0$, and $\cdot$ a composition operator, with
unit $1$. Kleene closure $P{\star}$ denotes finite iteration of $P$ using $\cdot$ zero or more times.

We consider the class of weak Kleene algebras~\cite{Guttman2010}, which build on weak dioids, as these are the most
useful class of Kleene algebra to characterise reactive programs.

\begin{definition}
  A weak dioid is an algebraic structure $(K, +, 0, \cdot, 1)$ such that $(K, +, 0)$ is an idempotent and commutative
  monoid; $(K, \cdot, 1)$ is a monoid; the composition operator $\cdot$ left- and right-distributes over $+$; and $0$ is
  a left annihilator for~$\cdot$.
\end{definition}
\noindent The $0$ operator represents miraculous behaviour. It is a left annihilator of composition, but not a right annihilator
as this often does not hold for programs. $K$ is partially ordered by $x \le y \defs (x + y = y)$, which is defined in
terms of $+$, and has least element $0$. A weak KA extends this with the behaviour of the star.
\begin{definition} \label{def:wka}
  A weak Kleene algebra is a structure $(K\!, +, 0, \cdot, 1, \star)$ such that

  \vspace{-3.4ex}
  \begin{center}
  \begin{tabular}{ll}
  \begin{minipage}{0.5\textwidth}
  \begin{enumerate}
    \item $(K, +, 0, \cdot, 1)$ is a weak dioid
    \item $1+x \cdot x{\star} \le x{\star}$
  \end{enumerate}
  \end{minipage} &
  \begin{minipage}{0.5\textwidth}
  \begin{enumerate}
    \setcounter{enumi}{2}
    \item $z + x \cdot y \le y \implies x{\star} \cdot z \le y$ \label{thm:starinductr}
    \item $z + y \cdot x \le y \implies z \cdot x{\star} \le y$
  \end{enumerate}
  \end{minipage}  
  \end{tabular}
  \end{center}
\end{definition}
Various enrichments and specialisations of these axioms exist; for a complete survey see~\cite{Kozen90}. For our
purposes, these axioms alone suffice. From this base, a number of useful identities can be derived:

\begin{theorem} \label{thm:kalaws} $1\star = 0\star = 1$ \quad $x{{\star}{\star}} = x{\star}$ \quad $x{\star} = 1 + x\cdot x\star$ \quad $(x+y){\star} = (x{\star} \cdot y{\star}){\star}$ \quad $x \cdot x{\star} = x{\star} \cdot x$ \end{theorem}
\noindent Kleene Algebra with Tests~\cite{Kozen97} (KAT) extends the algebra with conditions, and has been successfully
applied in program verification~\cite{Armstrong2015,Gomes2016}. A test is a kind of assumption that entails miraculous
behaviour if a condition is violated, and is otherwise ineffectual. The set of tests $T$ are those elements $a, b \in K$
below the identity: $a \le 1$, over which a Boolean algebra is defined.  Tests enjoy a number of additional properties.
\begin{theorem} $a \cdot 0 = 0 \quad a \cdot b = b \cdot a \quad a\star = 1$
\end{theorem}

\noindent UTP relations form a KA $(Rel, \intchoice, \false, \relsemi, \II, \star)$, where
$P\star \defs (\nu X @ \II \intchoice P \relsemi X)$. This definition is equivalent to
$P\star = (\Intchoice_{i \in \nat} ~ P^i)$~\cite{Foster19a-IsabelleUTP} where $P^n$ iterates sequential composition $n$
times. The proof proceeds by application of antisymmetry, the star induction law of Definition~\ref{def:wka}, and the
complete lattice theorems.

Typically, UTP theories, like $\theoryset{\healthy{NSRD}}$, share the operators for choice ($\intchoice$) and
composition ($\relsemi$), only redefining them when absolutely necessary. Formally, given a UTP theory defined by a
healthiness condition $\healthy{H}$, the set of healthy relations $\theoryset{\healthy{H}}$ is closed under $\intchoice$
and $\relsemi$. This has the major advantage that a large body of laws is directly applicable from the relational
calculus. The ubiquity of $\intchoice$, in particular, can be characterised through the subset of continuous UTP
theories, where $\healthy{H}$ distributes through arbitrary non-empty infima. We formally define this class of
healthiness condition below.
\begin{definition}[Continuous Healthiness Condition] \label{def:cont-hcond} \isalink{https://github.com/isabelle-utp/utp-main/blob/f22f260d20c95337e30ee7bb4d6d16e6fda27af0/utp/utp_healthy.thy\#L112}
$$\healthy{H}\left(\Intchoice_{i\in I}\,P(i)\right) = \Intchoice_{i\in I}\,\healthy{H}(P(i)) \text{ provided } I \neq \emptyset$$
\end{definition}
\noindent An infinite nondeterministic choice is necessary to support Kleene star iteration. Monotonicity of
$\healthy{H}$ follows from continuity, and so such theories induce a complete lattice. Moreover, if $\healthy{H}$ is
defined by composition $\healthy{H}_1 \circ \healthy{H}_2 \cdots \circ \healthy{H}_n$, as in
Theorem~\ref{thm:healthy-comp}, then we can show it is continuous by showing each $\healthy{H}_i$ is
continuous. Continuous UTP theories include designs~\cite{Hoare&98,Guttman2010}, CSP, and \Circus~\cite{Oliveira&09}. A
consequence of continuity is that the relational weakest fixed-point operator $\mu X @ F(X)$ constructs healthy
relations when $F : Rel \to \theoryset{\healthy{H}}$.

Though these theories share infima and weakest fixed points, they do not, in general, share $\top$ and $\bot$ elements,
which is why the infima are non-empty Definition~\ref{def:cont-hcond}. Rather, we have a top element
$\thtop{H} \defs \healthy{H}(\false)$ and a bottom element $\thbot{H} \defs \healthy{H}(\true)$~\cite{Foster17c}. The
theories also do not share the relational identity $\II$, but typically define a bespoke identity $\IIT{H}$, which means
that $\theoryset{\healthy{H}}$ is not closed under the relational Kleene star. However, $\theoryset{\healthy{H}}$ is
closed under Kleene plus, $P^{+} \defs P \relsemi P\star$, since it is equivalent to
$(\Intchoice_{i \in \nat} ~ P^{i+1})$, which iterates $P$ one or more times. Thus, we can obtain a theory Kleene star
with $P \bm{\star} \defs \IIT{H} \sqcap P^{+}$, under which $\healthy{H}$ is indeed closed. We, therefore, define the
following criteria for a UTP theory.
\begin{definition}
  A Kleene UTP theory $(\healthy{H}, \IIT{H})$ satisfies the following conditions: (1) $\healthy{H}$ is idempotent and
  continuous; (2) $\healthy{H}$ is closed under sequential composition; (3) identity $\IIT{H}$ is $\healthy{H}$-healthy;
  (4) it is a left- and right-unit, $\IIT{H} \relsemi P = P \relsemi \IIT{H} = P$, when $P$ is $\healthy{H}$-healthy;
  and (5) $\thtop{H} \relsemi P = \thtop{H}$, when $P$ is $\healthy{H}$-healthy. \isalink{https://github.com/isabelle-utp/utp-main/blob/f22f260d20c95337e30ee7bb4d6d16e6fda27af0/theories/kleene/utp_kleene.thy\#L126}
\end{definition}
\noindent From these properties, we can prove the following theorem.
\begin{theorem} \label{thm:kautp} If $(\healthy{H}, \IIT{H})$ is a Kleene UTP theory, then
  $(\theoryset{\healthy{H}}, \intchoice, \thtop{H}, \relsemi, \IIT{H}, \bm{\star})$ forms a weak Kleene algebra.
\end{theorem}
\begin{proof}
  We prove this in Isabelle/UTP by lifting of laws from the Isabelle/HOL KA
  hierarchy~\cite{Armstrong2015,Gomes2016}. For details see~\cite{Foster-KA-UTP}.
\end{proof}
\noindent The identities of Theorem~\ref{thm:kalaws} hold in a Kleene UTP theory, which allow us to reason about
iterative programs. In particular, we can show that
$(\theoryset{\healthy{NSRD}}, \intchoice, \Miracle, \relsemi, \IIsrd, \bm{\star})$ and
$(\theoryset{\healthy{RR}}, \intchoice, \false, \relsemi, \IIr, \bm{\star})$ both form weak KAs. Moreover, we can now
also show how to calculate iterative contracts~\cite{Foster-RDES-UTP}.

\begin{theorem}[Reactive Contract Iteration] \label{thm:rc-iter} \isalink{https://github.com/isabelle-utp/utp-main/blob/07cb0c256a90bc347289b5f5d202781b536fc640/theories/rea_designs/utp_rdes_normal.thy\#L768}
\begin{align*}
  \rc{P_1}{P_2}{P_3}\bm{\star} &= \rc{P_3\bm{\star} \wpR P_1}{P_3\bm{\star} \relsemi P_2}{P_3\bm{\star}} \\
  \rc{P_1}{P_2}{P_3}^+ &= \rc{P_3\bm{\star} \wpR P_1}{P_3\bm{\star} \relsemi P_2}{P_3^+}
\end{align*}
\end{theorem}
\noindent We note that the outer and inner star are different operators. The outer star is formed from the identity
$\IIsrd$, and the inner star from $\IIr$. The precondition states that $P_3$ must not violate $P_1$ after any number of
iterations. The pericondition has $P_3$ iterated followed by $P_2$ holding, since the final observation is
intermediate. The postcondition simply iterates $P_3$. We also provide a similar law for the Kleene plus operator. It
distributes in the same way, except that both the precondition and the pericondition use the star because they must hold
before the first iteration; only the postcondition uses the plus.

In this section we have established the basis for calculating and reasoning about iterative reactive contracts. In the
next section we specialise our UTP theory to stateful-failure reactive designs, and develop the underlying equational
theory. We return to the subject of iteration in Section~\ref{sec:iter}.

\section{Reactive Relations of Stateful Failures-Divergences}
\label{sec:circus-rc}
In this section, we specialise our contract theory to incorporate failure traces, which are used in CSP, \Circus, and
related languages~\cite{Zhan2008}. We define atomic operators to describe the underlying reactive relations, and the
associated equational theory to expand and simplify compositions arising from Theorems~\ref{thm:rc-comp} and
\ref{thm:rc-iter}, and thus support automated reasoning. We consider external choice separately
(\S\ref{sec:ext-choice}).

The failures-divergences model~\cite{Roscoe2005} was defined to give a denotational semantics to CSP. It models a
process with a pair of sets: $F \subseteq \power (\seq\,\textit{Event} \, \times \, \power\,\textit{Event})$ and
$D \subseteq \power (\seq\,\textit{Event})$, which are, respectively, the set of failures and divergences. A failure is
a trace of events plus a set of events can be refused at the end of the interaction. A divergence is a trace of events
that leads to divergent behaviour, that is, unpredictable behaviour like that of $\Chaos$. A distinguished event
$\tick \in Event$ is used as the final element of a trace to indicate that this is a terminating observation. The UTP
gives a relational account of the failures-divergences model~\cite{Hoare&98}, which was expanded upon by Woodcock and
Cavalcanti~\cite{Cavalcanti&06}, and by Oliveira~\cite{Oliveira2005-PHD} to account for state variables in
\Circus~\cite{Oliveira&09}. It is this latter model that we here call the stateful failures-divergences model.

Healthiness condition $\healthy{NCSP} \defs \healthy{NSRD} \circ \healthy{CSP3} \circ \healthy{CSP4}$ characterises the
stateful failures-divergences model~\cite{Cavalcanti&06,Oliveira&09}. Healthiness conditions $\healthy{CSP3}$ and
$\healthy{CSP4}$ are defined below.

\begin{definition}[Stateful-Failure Healthiness Conditions] \isalink{https://github.com/isabelle-utp/utp-main/blob/f22f260d20c95337e30ee7bb4d6d16e6fda27af0/theories/sf_rdes/utp_sfrd_healths.thy\#L32}
  \begin{align*}
  \healthy{CSP3}(P) &\defs (\Skip \relsemi P) && \text{There are no references to $\refu$}. \\
  \healthy{CSP4}(P) &\defs (P \relsemi \Skip) && \text{The postcondition may not refer to $\refu'$}.
  \end{align*}
\end{definition}  
\noindent $\healthy{CSP3}$ and $\healthy{CSP4}$ ensure the refusal sets are well-formed~\cite{Cavalcanti&06,Hoare&98}:
$\refu'$ can be mentioned only in the pericondition, since refusals are only observed in quiescent
observations. \healthy{NCSP}, like \healthy{NSRD}, is continuous and has $\Skip$ as a left and right unit. Thus, it
fulfils the criteria of a Kleene UTP theory (Definition~\ref{thm:kautp}), and consequently
$(\theoryset{\healthy{NCSP}}, \intchoice, \Miracle, \relsemi, \Skip, \bm{\star})$ forms a Kleene algebra. Every
\healthy{NCSP}-healthy relation corresponds to a reactive contract with the following specialised form~\cite{Foster-SFRD-UTP}.
$$\rc{P_1(\trace, \state)}{P_2(\trace, \state, \refu')}{P_3(\trace, \state, \state')}$$
The underlying reactive relations capture a portion of the stateful failures-divergences. $P_1$ is the precondition, which
captures the initial states and traces that do not induce divergence. It corresponds to the the complement of $D$, the
set of divergences~\cite{Roscoe2005,Cavalcanti&06}. $P_2$ is the pericondition, which captures the stateful failures of a
program: the set of events that may be refused ($\refu'$) having performed trace $\trace$, starting in state $\state$. It
corresponds to the the failure traces in $F$ that are not terminating. $P_3$ captures the terminated behaviours, where a
final state is observed but no refusals. It, of course, corresponds to the traces in $F$ that have $\tick$ as the final
element. We now characterise these reactive relations using healthiness conditions.

\begin{definition} Stateful-failure Reactive Relations, Finalisers, and Conditions are characterised as fixed-points of
  the healthiness conditions $\healthy{CRR}$, $\healthy{CRF}$, and $\healthy{CRC}$ defined
  below. \isalink{https://github.com/isabelle-utp/utp-main/blob/07cb0c256a90bc347289b5f5d202781b536fc640/theories/sf_rdes/utp_sfrd_rel.thy\#L7}
  $$\healthy{CRR}(P) \defs \exists \refu @ \healthy{RR}(P) \qquad \healthy{CRF}(P) \defs \exists (\refu, \refu') @ \healthy{RR}(P) \qquad \healthy{CRC}(P) \defs \exists \refu @ \healthy{RC}(P)$$
\end{definition}

\begin{figure}
  \centering
  \includegraphics[width=.6\linewidth]{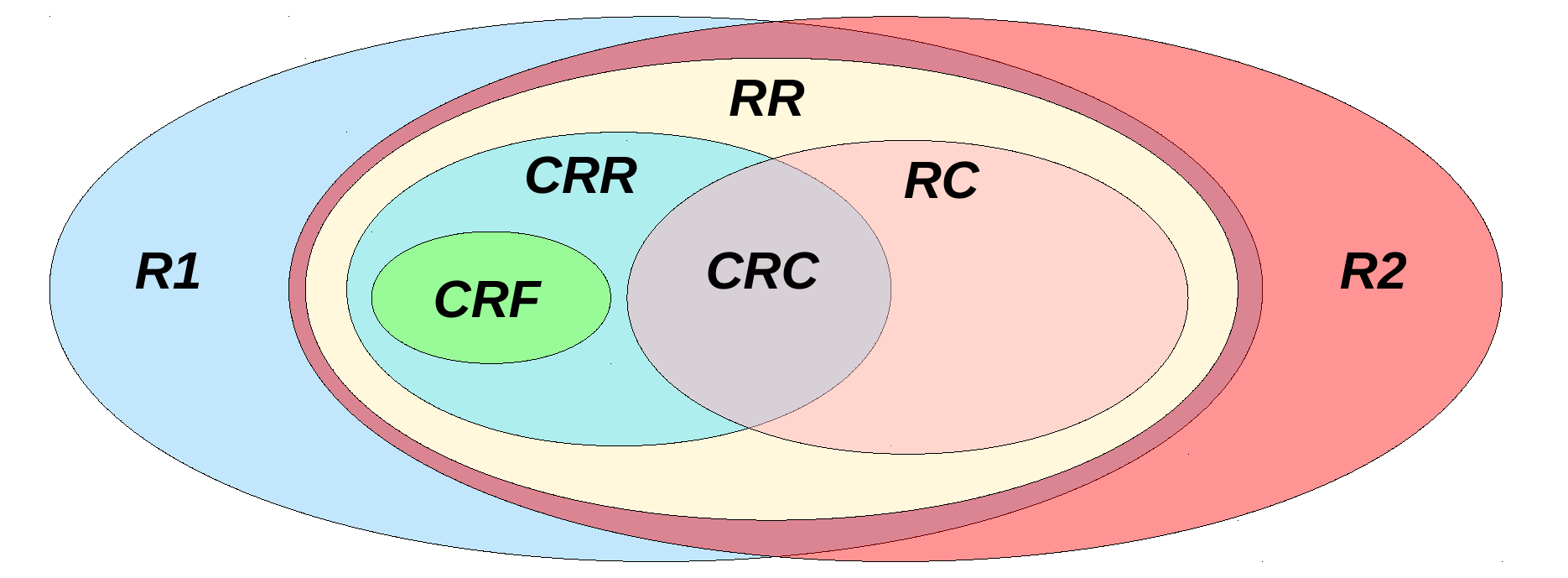}
  \caption{Reactive Relational Healthiness Conditions Venn diagram}
  \label{fig:rrel-impls}
\end{figure}

\noindent These are straightforward extensions of the healthiness conditions for reactive relations ($\healthy{RR}$) and
conditions ($\healthy{RC}$) that we previously defined~\cite{Foster17c} and are presented in
Table~\ref{tab:reahcond}. In addition to requiring that the relations describe a well-formed trace, $\healthy{CRR}$,
$\healthy{CRF}$, and $\healthy{CRC}$ require that there is no reference to $\refu$, because there is never a dependence
on the refusal set of a predecessor. Reactive finalisers ($\healthy{CRF}$) additionally forbid reference to $\refu'$:
they are used to characterise postconditions in a stateful-failure reactive contracts. Every $\healthy{CRF}$-healthy
relation is also $\healthy{CRR}$-healthy; further containments are shown in Figure~\ref{fig:rrel-impls}. We can formally
characterise \healthy{NCSP} contracts with the following theorem.

\begin{theorem} \label{thm:ncsp-intro} $\rc{P_1}{P_2}{P_3}$ is \healthy{NCSP}-healthy if the following conditions are satisfied: \isalink{https://github.com/isabelle-utp/utp-main/blob/07cb0c256a90bc347289b5f5d202781b536fc640/theories/sf_rdes/utp_sfrd_healths.thy\#L533}

  \begin{enumerate}
    \item $P_1$ is \healthy{CRC}-healthy;
    \item $P_2$ is \healthy{CRR}-healthy;
    \item $P_2$ does not refer to $\state'$; and
    \item $P_3$ is \healthy{CRF}-healthy.
  \end{enumerate}
\end{theorem}

\noindent Due to the restrictions on $\refu\!$, $\IIr$ is not $\healthy{CRR}$-healthy, and so we define the identity
relation below.

\begin{definition} $\IIc \defs (\state' = \state \land tr' = tr)$ \isalink{https://github.com/isabelle-utp/utp-main/blob/07cb0c256a90bc347289b5f5d202781b536fc640/theories/sf_rdes/utp_sfrd_rel.thy\#L65} \end{definition}

\noindent This identity specifies only that the trace and state remain unchanged, whilst $\refu$ is
  unspecified. This is different to $\IIr$ which has $\refu' = \refu$ because $r$ is specialised to $\refu$ (see
  Definition~\ref{def:rrel-id}). Consequently, $\IIc$ is indeed $\healthy{CRF}$-healthy. $\IIc$ is a left unit for
$\healthy{CRR}$, $\healthy{CRF}$, and $\healthy{CRC}$-healthy relations, but it is a right unit only for
$\healthy{CRF}$. This is because in $P \relsemi \IIc$, if $P$ refers to $\refu'$ then this information is lost, whilst
$\healthy{CRF}$ relations do not refer to $\refu'$. We can construct a Kleene algebra for $\healthy{CRF}$-healthy
relations.

\begin{theorem} $(\theoryset{\healthy{CRF}}, \intchoice, \false, \relsemi, \IIc, \bm{\star})$ forms a weak Kleene algebra. \isalink{https://github.com/isabelle-utp/utp-main/blob/07cb0c256a90bc347289b5f5d202781b536fc640/theories/sf_rdes/utp_sfrd_rel.thy\#L512}
\end{theorem}

\noindent Using the Kleene star operators for $\healthy{NCSP}$ and $\healthy{CRF}$, we can also revalidate
Theorem~\ref{thm:rc-iter} in this context. This is necessary because the star in Theorem~\ref{thm:rc-iter} is defined in
terms of $\IIr$ and not $\IIc$. Nevertheless the two star operators are strongly related, and so we need only to
slightly adapt the proof.

Having defined our theories of stateful-failure reactive relations, we now proceed to define operators for constructing
pre-, peri-, and postconditions. These operators allow us to describe a distinct pattern for the form of reactive
programs. We describe this pattern using the following constructs.

\begin{definition}[Reactive Relational Operators] \label{def:rrelop} \isalink{https://github.com/isabelle-utp/utp-main/blob/07cb0c256a90bc347289b5f5d202781b536fc640/theories/sf_rdes/utp_sfrd_rel.thy\#L757}
 \arraycolsep=1.5pt

$$\begin{array}{lrcl}
  \text{Assumption:}   &\cspin{s(\state)}{t(\state)} &\defs& \healthy{CRC}(s(\state) \rimplies \negr (t(\state) \le \trace)) \\[.2ex]
  \text{Quiescence:} \qquad \qquad  &\cspen{s(\state)}{t(\state)}{E(\state)} &\defs& \healthy{CRR}(s(\state) \land \trace = t(\state) \land (\forall e\!\in\!E(\state) @ e \notin \refu')) \\[.2ex]
  \text{Finalisation:} \qquad &\csppf{s(\state)}{\sigma}{t(\state)} &\defs& \healthy{CRR}(s(\state) \land \state' = \sigma(\state) \land \trace = t(\state))
\end{array}$$
\end{definition}
\noindent We utilise expressions $s$, $t$, and $E$ that refer only to the variables indicated. Namely, $s\!:\!\Bool$ is
a condition on $\state$, $t\!:\!\seq \textit{Event}$ is a trace expression that describes an event sequence in terms of
$\state$, and $E : \power \textit{Event}$ describes a set of events. The use of trace expressions allows us to handle
symbolic traces that contain free state variables, and characterise a potentially infinite number of traces with a
finite presentation. With this, we can calculate semantics for reactive programs with an infinite number of states.

$\cspin{s(\state)}{t(\state)}$ is a $\healthy{CRC}$-healthy reactive condition that is used to specify assumptions on
the state and trace in preconditions. It states that, if the state initially satisfies condition $s$, then $t$ is not a
prefix of the overall trace. For example, the assumption $\cspin{x > 2}{\langle a, b \rangle}$ means that if the state
variable $x$ is initially greater than 2, then we disallow the trace $\langle a, b \rangle$, and any extension
thereof. The intuition here is that $t$ is a trace that introduces divergence, and so any extension of $t$ violates the
precondition. Put another way, any trace that is a either strict prefix or orthogonal to $t$ satisfies the precondition
when $s$ holds. Effectively, $t$ sets a strict upper bound on the traces permitted by the precondition.

$\cspen{s(\state)}{t(\state)}{E(\state)}$ is used in periconditions to specify quiescent observations, and corresponds
to a set of symbolic failure traces. It specifies that the state variables initially satisfy $s$, the interaction
described by $t$ has occurred, and finally we reach a quiescent phase where none of the events in $E$ are being
refused. It has the form of an acceptance trace~\cite{Roscoe2005}, as this provides a more comprehensible presentation,
but the semantics is encoded as a collection of failure traces. Specifically,
$(\forall e\!\in\!E(\state) @ e \notin \refu')$ defines all possible refusals that satisfy symbolic acceptance set
$E$. We sometimes write $\cspens{t(\state)}{E(\state)}$ when $s$ is $\ptrue$.

$\csppf{s(\state)}{\sigma}{t(\state)}$ is used in postconditions to specify final terminated observations. It specifies
that the initial state satisfies $s$, the state update $\sigma$ is applied to update $\state$, and the symbolic
interaction $t$ has occurred. Since $\csppf{s}{\sigma}{t}$ does not refer to $\refu'$, it is $\healthy{CRF}$-healthy. We
sometimes write $\csppfs{\sigma}{t(\state)}$ in the case that $s$ is $\ptrue$. Moreover, we also introduce the
abbreviation $\casm{s} \defs \csppf{s}{id}{\snil}$ that denotes a reactive relational test (or assumption) on the state
of the property $s$.

These operators are all deterministic, in the sense that they describe a single interaction and state-update
history. There is no need for explicit nondeterminism here, as this is achieved using $\bigvee$. These operators allow
us to concisely specify the basic operators of our theory as given below.

\begin{definition}[Basic Reactive Operators] \label{thm:bcircus-def} \isalink{https://github.com/isabelle-utp/utp-main/blob/07cb0c256a90bc347289b5f5d202781b536fc640/theories/sf_rdes/utp_sfrd_prog.thy\#L257}
\begin{align}
  \assignsC{\sigma} ~\defs~& \rc{\truer}{\false}{\csppf{\ptrue}{\sigma}{\langle\rangle}} \label{def:gen-asn} \\[.1ex]
  \Skip              ~=~& \assignsC{id} \\[.1ex]
  \ckey{Do}(a)      ~=~& \rc{\truer}{\cspen{\ptrue}{\langle\rangle}{\{a\}}}{\csppf{\ptrue}{id}{\langle a \rangle}} \\[.1ex]
  \Stop             ~\defs~& \rc{\truer}{\cspen{\ptrue}{\langle\rangle}{\emptyset}}{\false} %
\end{align}
\end{definition}

\noindent The definitions of $\Skip$ and $\ckey{Do}(a)$ are expressed as theorems that we have proved using
Definition~\ref{def:skipprefas}. However, for the remainder of this paper we treat these identities as
definitions. Generalised assignment $\assignsC{\sigma}$ is again inspired by \cite{Back1998}. It has a $\truer$
precondition and a $\false$ pericondition: it has no intermediate observations. The postcondition states that for any
initial state ($\ptrue$), the state is updated using $\sigma$, and no events are produced ($\langle\rangle$). A
singleton assignment $x := v$ can be expressed using $\assignsC{x \mapsto v}$. We can use it to show that
$\Skip = \assignsC{id}$, where $id : \Sigma \to \Sigma$ is the identity function that leaves all variables unchanged.

$\ckey{Do}(a)$ encodes an event action. Its pericondition states that no event has occurred, and $a$ is accepted. Its
postcondition extends the trace by $a$, leaving the state unchanged. We can denote \Circus event prefix $a \then P$ as
$\ckey{Do}(a) \relsemi P$. Finally, $\Stop$ represents a deadlock: its pericondition states the trace is unchanged and
no events are being accepted. The postcondition is false as there is no way to terminate. A \Circus guard $g \guard P$
can be denoted as $(\rconditional{P}{g}{\Stop})$, which behaves as $P$ when $g$ is true, and otherwise deadlocks.

To calculate contractual semantics, we need laws to reduce pre-, peri-, and postconditions. These need to cater for
compositions of quiescent and final observations using operators like internal choice ($\bigsqcap$), sequential
composition ($\relsemi$), and external choice ($\Extchoice$, see \S\ref{sec:ext-choice}). So, we
prove~\cite{Foster-SFRD-UTP} the following laws for $\mathcal{E}$ and $\Phi$.

\begin{theorem}[Reactive Relational Compositions] \label{thm:crel-comp} \isalink{https://github.com/isabelle-utp/utp-main/blob/07cb0c256a90bc347289b5f5d202781b536fc640/theories/sf_rdes/utp_sfrd_rel.thy\#L757}
  \begin{align}
    \csppf{\ptrue}{id}{\langle\rangle} &= \IIc \label{thm:crc0} \\[.2ex]
    \csppf{s_1}{\sigma_1}{t_1} \relsemi \csppf{s_2}{\sigma_2}{t_2} &= \csppf{s_1 \land \substapp{\sigma_1}{s_2}}{\sigma_2 \circ \sigma_1}{t_1 \cat \substapp{\sigma_1}{t_2}} \label{thm:crc2} \\[.2ex]
    \csppf{s_1}{\sigma_1}{t_1} \relsemi \cspen{s_2}{t_2}{E} &= \cspen{s_1 \land \substapp{\sigma_1}{s_2}}{t_1 \cat \substapp{\sigma_1}{t_2}}{\substapp{\sigma_1}{E}} \label{thm:crc3} \\[.2ex]
    \rconditional{\csppf{s_1}{\sigma_1}{t_1}}{c}{\csppf{s_2}{\sigma_2}{t_2}} &= \csppf{\rconditional{s_1}{c}{s_2}}{\rconditional{\sigma_1}{c}{\sigma_2}}{\rconditional{t_1}{c}{t_2}} \label{thm:crc4} \\[.2ex]
    \rconditional{\cspen{s_1}{t_1}{E_1}}{c}{\cspen{s_2}{t_2}{E_2}} &= \cspen{\rconditional{s_1}{c}{s_2}}{\rconditional{t_1}{c}{t_2}}{\rconditional{E_1}{c}{E_2}} \label{thm:crc5} \\[.2ex]
    \left(\bigwedge_{i \in I}\,\cspen{s(i)}{t}{E(i)}\right) &= \cspen{\bigwedge_{i \in I}\,s(i)}{t}{\bigcup_{i \in I}\,E(i)} \label{thm:crc6} \\[.2ex]
    \left(\bigvee_{i \in I}\,\cspen{s(i)}{t}{E(i)}\right) &= \cspen{\bigvee_{i \in I}\,s(i)}{t}{\bigcap_{i \in I}\,E(i)} \label{thm:crc7} \\[.2ex]
    \csppf{s}{\sigma}{t}^{\star} &= \bigsqcap_{n \in \nat} ~ \csppf{\bigwedge_{i \le n} (\substapp{\sigma^i}{s})}{\sigma^n}{\prod_{j < n} ~ (\substapp{\sigma^j}{t})} \label{thm:crc8}
  \end{align}
\end{theorem}

\noindent Law \eqref{thm:crc0} gives the meaning of $\Phi$ with a trivial precondition, state update, and empty trace:
it is simply the reactive identity. Law \eqref{thm:crc2} states that the composition of two terminated observations
results in the conjunction of the state conditions, composition of the state updates, and concatenation of the
traces. It is necessary to apply the initial state update $\sigma_1$ as a substitution to both the second state
condition ($s_2$) and the trace expression ($t_2$). Law \eqref{thm:crc3} is similar, but accounts for the enabled events
rather than state updates. Laws \eqref{thm:crc2} and \eqref{thm:crc3} are required because of
Theorem~\thmeqref{thm:rc-comp}{thm:rc-seq}, which sequentially composes a pericondition with a postcondition, and a
postcondition with a postcondition.

Laws \eqref{thm:crc4} and \eqref{thm:crc5} show how conditional distributes through the operators. Law \eqref{thm:crc6}
shows that a conjunction of intermediate observations with a common trace corresponds to the conjunction of the state
conditions, and the union of the enabled events. It is needed for external choice, which conjoins the periconditions
(see \S\ref{sec:ext-choice}). Law \eqref{thm:crc7} shows the dual case of \eqref{thm:crc6}: when taking a choice of
periconditions, we have the disjunction of all the state conditions, and intersection of all enabled events.

Finally, \eqref{thm:crc8} gives the meaning of an iterated final observation. The nondeterministic choice over
$n \in \nat$ denotes the number of iterations.  Inside, the $\Phi$ operator distributes iteration though the condition,
state updates, and trace. Here, $f^n$ is iterated function composition ($f \circ f \circ f \cdots$), and $\prod$ is
iterated concatenation: $$\prod_{i < n}\, xs(i) \defs xs(0) \cat xs(1) \cat \cdots \cat xs(n-1)$$ The condition of an
iterated final observation requires that $s$ holds whenever $\sigma$ is applied as a substitution $i$ times, where
$i \le n$. The state update applies $\sigma$ a total of $n$ times in sequence. The trace expression concatenates $t$ a
total of $n$ times, and each instance has the state update applied $j < n$ times.

We can now use these laws, along with Theorem \ref{thm:rc-comp}, to calculate the semantics of
processes, and to prove equality and refinement conjectures, as we illustrate below.

\begin{example} We show that
  $(x\!:=\!1 \relsemi \ckey{Do}(a.x) \relsemi x := x + 2) ~=~ (\ckey{Do}(a.1) \relsemi x\!:=\!3)$. By applying
  Definition~\ref{thm:bcircus-def} and Theorems~\thmeqref{thm:rc-comp}{thm:rc-seq}, \ref{thm:crel-comp},
  \ref{thm:evwp}, both sides reduce to
  $\,\rcs{\cspen{true}{\langle\rangle}{\{a.1\}}}{\csppf{true}{\{x \mapsto 3\}}{\langle a.1 \rangle}}$, which has
  a single quiescent state, waiting for event $a.1$, and a single final state, where $a.1$ has occurred and state
  variable $x$ has been updated to $3$. We calculate the left-hand side below.
  \begin{align*}
    &(x\!:=\!1 \relsemi \ckey{Do}(a.x) \relsemi x := x + 2) \\[.5ex]
    =& \left(\begin{array}{l}
         \rcs{\false}{\csppf{\ptrue}{\substmap{x \mapsto 1}}{\langle\rangle}} \relsemi \\[.2ex]
         \rcs{\cspen{\ptrue}{\langle\rangle}{\{a.x\}}}{\csppf{\ptrue}{id}{\langle a.x \rangle}} \relsemi \\[.2ex]
         \rcs{\false}{\csppf{\ptrue}{\substmap{x \mapsto x + 2}}{\langle\rangle}} 
       \end{array}\right) & [\text{Def.}~\ref{thm:bcircus-def}] \\[.5ex]
    =& \left(\begin{array}{l}
         \rcs{\begin{array}{l} \false \lor \\ \csppf{\ptrue}{\substmap{x \mapsto 1}}{\langle\rangle} \relsemi \cspen{\ptrue}{\langle\rangle}{\{a.x\}} \end{array}}
             {\begin{array}{l} \csppf{\ptrue}{\substmap{x \mapsto 1}}{\langle\rangle} \relsemi \\ \csppf{\ptrue}{id}{\langle a.x \rangle} \end{array}
               }  \relsemi \\[3ex]
         \rcs{\false}{\csppf{\ptrue}{\substmap{x \mapsto x + 2}}{\langle\rangle}} 
       \end{array}\right) & [\text{Thm.}~\ref{thm:rc-comp}] \\[.5ex]
    =& \left(\begin{array}{l}
         \rcs{\cspen{\ptrue[1/x]}{\langle\rangle[1/x]}{\{a.x\}[1/x]}}
             {\csppf{\ptrue[1/x]}{\substmap{x \mapsto 1}}{\langle a.x \rangle[1/x]}}  \relsemi \\[.5ex]
         \rcs{\false}{\csppf{\ptrue}{\substmap{x \mapsto x + 2}}{\langle\rangle}} 
       \end{array}\right) & [\text{Thm.}~\ref{thm:crel-comp}] \\[.5ex]
    =& \left(\begin{array}{l}
         \rcs{\cspen{\ptrue}{\langle\rangle}{\{a.1\}}}
             {\csppf{\ptrue}{\substmap{x \mapsto 1}}{\langle a.1 \rangle}}  \relsemi \\[.5ex]
         \rcs{\false}{\csppf{\ptrue}{\substmap{x \mapsto x + 2}}{\langle\rangle}} 
       \end{array}\right) & \\[.5ex]
    =& \begin{array}{l}
         \rcs{\begin{array}{l} \cspen{\ptrue}{\langle\rangle}{\{a.1\}} \lor \\ \csppf{\ptrue}{\substmap{x \mapsto 1}}{\langle a.1 \rangle} \relsemi \false \end{array}}
             {\begin{array}{l} \csppf{\ptrue}{\substmap{x \mapsto 1}}{\langle a.1 \rangle} \relsemi \\ \csppf{\ptrue}{\substmap{x \mapsto x + 2}}{\langle\rangle} \end{array}}
       \end{array} & [\text{Thm.}~\ref{thm:rc-comp}] \\[.5ex]
    =& \rcs{\cspen{\ptrue}{\langle\rangle}{\{a.1\}}}
             {\csppf{\ptrue[1/x]}{\substmap{x \mapsto x + 2} \circ \substmap{x \mapsto 1}}{\langle a.1 \rangle}}
       & [\text{Thm.}~\ref{thm:crel-comp}] \\[.5ex]
    =&\rcs{\cspen{\ptrue}{\langle\rangle}{\{a.1\}}}{\csppf{\ptrue}{\substmap{x \mapsto 3}}{\langle a.1 \rangle}}
  \end{align*}

  \noindent In the first step, we expand out the definitions of the three sequential actions using
  Definition~\ref{thm:bcircus-def}. In the second step, we employ Theorem~\ref{thm:rc-comp} to calculate the sequential
  composition of the first two contracts. In the third step, we use Theorem~\ref{thm:crel-comp} to calculate the
  resulting composite peri- and postconditions, which in particular pushes the initial substitution into both the
  quiescent and terminated observations of the second contract. In the fourth step, we apply the resulting substitutions
  to complete composition of the first two contracts. In the remaining steps, we apply the same theorems again to
  compose with the third contract. \qed
  
\end{example}
\noindent This proof can be automated using a single invocation of the \textsf{rdes-eq} tactic~\cite{Foster17c}
  in Isabelle/UTP, which implements our calcuational proof strategy\footnote{Several examples of this can be
  found in our respository, using the link to the
  right. \isalink{https://github.com/isabelle-utp/utp-main/blob/1287e8d68ca23fea81bc31064febe3d956db6ee2/tutorial/utp_csp_ex.thy}}. We
  can also use our calculation theorems, with the help of \textsf{rdes-eq}, to prove a number of general laws, which
  would otherwise require a complex manual proof~\cite{Foster-SFRD-UTP}.

\begin{theorem}[Stateful Failures-Divergences Laws] \label{thm:sfdl} \isalink{https://github.com/isabelle-utp/utp-main/blob/07cb0c256a90bc347289b5f5d202781b536fc640/theories/sf_rdes/utp_sfrd_prog.thy\#L994}

  \noindent
  \begin{minipage}{0.5\textwidth}
  \begin{align}
  \assignsC{\sigma}\relsemi\rc{P_1}{P_2}{P_3} &= \rc{\substapp{\sigma}{P_1}}{\substapp{\sigma}{P_2}}{\substapp{\sigma}{P_3}} \label{thm:asgdist}  \\
  \assignsC{\sigma} \relsemi \ckey{Do}(e) &= \ckey{Do}(\substapp{\sigma}{e}) \relsemi \assignsC{\sigma} \label{thm:asgev} \\
  x := v \relsemi e \then P &= e[v/x] \then x\!:=\!v \relsemi P \label{thm:sagmsev}
  \end{align}
  \end{minipage}
  \begin{minipage}{0.5\textwidth}
  \begin{align}
  a \then (P \intchoice Q) &= (a \then P) \intchoice (a \then Q) \label{thm:prfchoice} \\
  \assignsC{\sigma} \relsemi \assignsC{\rho} &= \assignsC{\rho \circ \sigma} \label{thm:asgcomp} \\ 
  \Stop \relsemi P &= \Stop \label{thm:dlockzero}
  \end{align}
  \end{minipage}
\end{theorem}

\noindent Law~\eqref{thm:asgdist} shows how a leading assignment distributes substitutions through a
contract. Laws~\eqref{thm:asgev} and \eqref{thm:asgcomp} are consequences of
Law~\eqref{thm:asgdist}. Law~\eqref{thm:asgev} shows that an assignment can be pushed through an event by applying the
substitution to the event expression. Law~\eqref{thm:sagmsev} is a further consequence of Law~\ref{thm:asgdist} that
shows the case for a singleton assignment and a prefixed action. Law~\eqref{thm:prfchoice} shows that a prefix event
distributes from the left through nondeterministic choice. Law~\eqref{thm:asgcomp} shows that composing two assignments
yields a single assignment where the two substitution functions are composed. Effectively, this law shows the
correspondence between functional and relational composition for deterministic relations represented by
assignments. Finally, Law~\eqref{thm:dlockzero} shows that the deadlock action, $\Stop$, is a left annihilator.

So far, the reactive contracts we have considered have all contained trivial preconditions. However, divergence is a
useful modelling technique that allows us to model unspecified or unpredictable behaviour, when certain assumptions are
violated. We consider, for example, the simple action $a \then \Skip \extchoice b \then \Chaos$. If event $a$ occurs,
then it terminates, and if $b$ occurs it diverges. The behaviour following the occurrence of $a$ is predictable
(termination), but the behaviour following the occurrence of $b$ is unpredictable.

In order to calculate contracts for actions of this form, we need to consider the weakest liberal precondition operator
$\wpR$. So far, we have only considered simple formulae of the form $P \wpR \truer = \truer$; we now supply theorems for
more sophisticated preconditions. Theorem~\ref{thm:rc-comp}-\eqref{thm:rc-seq} requires that, in a sequential
composition $P \relsemi Q$, we need to show that the postcondition of contract $P$ satisfies the precondition of
contract $Q$. We, consider for example the following partial calculation of the contract for $b \then \Chaos$.
\begin{align*}
      b \then \Chaos
  =~& \ckey{Do}(b) \relsemi \Chaos \\
  =~& \rcs{\cspen{\ptrue}{\langle\rangle}{\{b\}}}{\csppf{\ptrue}{id}{\langle b \rangle}} \relsemi \rc{\false}{\false}{\false} \\
  =~& \rc{\csppf{\ptrue}{id}{\langle b \rangle} \wpR \false}{\cspen{\ptrue}{\langle\rangle}{\{b\}} \lor \false}{\false} \\
  =~& \rc{\csppf{\ptrue}{id}{\langle b \rangle} \wpR \false}{\cspen{\ptrue}{\langle\rangle}{\{b\}}}{\false}
\end{align*}

\noindent The postcondition is $\false$, so this action has no final state. It can be quiescent, waiting for $b$ to
occur. We cannot, however, calculate the precondition yet; it states that the trace $\langle b \rangle$ should never
occur.

In general, the precondition of a reactive contract uses the weakest liberal precondition of a previously applied
postcondition. Theorem~\ref{thm:crel-comp} explains how to eliminate most composition operators in a contract's
postcondition, but not disjunction ($\lor$). Postconditions are, therefore, typically expressed as disjunctions of the
$\Phi$ operator. So, our weakest liberal precondition calculus needs to handle disjunctions of $\Phi$ terms.

\begin{theorem}[Reactive Preconditions] \label{thm:evwp} \isalink{https://github.com/isabelle-utp/utp-main/blob/07cb0c256a90bc347289b5f5d202781b536fc640/theories/sf_rdes/utp_sfrd_rel.thy\#L941}
  \begin{align}
    \csppf{s}{\sigma}{t} \wpR \false               &= \cspin{s}{t} \label{law:wpfalse} \\
    \csppf{s_1}{\sigma}{t_1} \wpR \cspin{s_2}{t_2}  &= \cspin{s_1 \land \substapp{\sigma}{s_2}}{t_1 \cat (\substapp{\sigma}{t_2})} \label{law:wpin} \\
    \cspin{\pfalse}{t}                             &= \truer \label{law:wpinfalse} \\
    \cspin{\ptrue}{\snil}                          &= \false \label{law:wpintrue} \\
    \cspin{s_1}{t} \land \cspin{s_2}{t}            &= \cspin{s_1 \lor s_2}{t} \\
    \cspin{s_1}{t} \lor \cspin{s_2}{t}             &= \cspin{s_1 \land s_2}{t}
  \end{align}
\end{theorem}

We recall that $\cspin{s}{t}$ means that, if $s$ is satisfied in the current state, then the action can only perform
traces that do not have $t$ as a prefix, or else divergence will result. Law~\eqref{law:wpfalse} calculates the weakest
liberal precondition under which $\csppf{s}{\sigma}{t}$ achieves $\false$, which is impossible. Consequently, we must
require that the trace $t$ never occurs, when the state initially satisfies $s$. With this law, we can complete the
contract calculation of $b \then \Chaos$ to obtain
$$\rc{\cspin{\ptrue}{\langle b \rangle}}{\cspen{\ptrue}{\langle\rangle}{\{b\}}}{\false}$$
whose precondition assumes that event $b$ does not occur initially. Law~\eqref{law:wpin} considers the final observation
specified $\cspin{s_2}{t_2}$. If we start in a state that satisifes $s_1$ and $s_2$ with state update $\sigma$ applied,
then an upper bound on the trace is $t_1 \cat (\sigma \dagger t_2)$, which also inserts the state update.

The remaining laws are for different compositions for $\mathcal{I}$. Law~\eqref{law:wpinfalse} shows that if an
assumption's condition is $\pfalse$, then it reduces to the reactive precondition $\truer$. Conversely,
law~\eqref{law:wpintrue} shows that if the condition is $\ptrue$, but the trace is $\snil$, then this is $\false$, since
all traces are disallowed. The remaining two laws show the effect of conjunction and disjunction on assumptions sharing
a trace expression.

This completes the calculational approach for the core sequential programming operators. In the next section, we extend
our proof approach to support external choice~\cite{Hoare85,Roscoe2005}.

\section{External Choice and Productivity}
\label{sec:ext-choice}
In this section we consider external choice~\cite{Hoare85,Roscoe2005}, and characterise the class of productive
contracts~\cite{Foster17c}, which are also essential in verifying recursive and iterative reactive programs.

An external choice $P \extchoice Q$ is resolved whenever either $P$ or $Q$ engages in an event or terminates. Thus, its
semantics requires that we filter observations with a non-empty trace. We introduce healthiness condition
$\healthy{R4}(P) \defs (P \land \trace > \langle\rangle)$, whose fixed points strictly increase the trace, and its dual
$\healthy{R5}(P) \defs (P \land \trace = \langle\rangle)$ where the trace is unchanged. We use these to define indexed
external choice.
\begin{definition}[Indexed External Choice] \label{def:ext-choice} \isalink{https://github.com/isabelle-utp/utp-main/blob/07cb0c256a90bc347289b5f5d202781b536fc640/theories/sf_rdes/utp_sfrd_extchoice.thy\#L146}
\begin{align*}
& \Extchoice i \in I @ \rc{P_1(i)}{P_2(i)}{P_3(i)} \defs \\[.5ex]
& \qquad \textstyle \rc{\bigwedge_{i \in I} ~ P_1(i)}{\left(\bigwedge_{i \in I} ~ \healthy{R5}(P_2(i))\right) \lor \left(\bigvee_{i \in I} ~ \healthy{R4}(P_2(i))\right)}{\bigvee_{i \in I}{P_3(i)}}
\end{align*}
\end{definition}
\noindent This generalises the binary definition~\cite{Hoare&98,Oliveira&09}, and recasts our definition
in~\cite{Foster17c} for calculation. As we note in \S\ref{sec:circus-rc}, every $\healthy{NCSP}$ relation
  corresponds to a reactive contract, and so this definition applies to any stateful-failure reactive design. Like
nondeterministic choice, the precondition of external choice requires that all constituent preconditions are
satisfied. In the pericondition, \healthy{R4} and \healthy{R5} filter all observations. We take the conjunction of all
\ckey{R5} behaviours: no event has occurred, and all branches are offering an event. We also take the disjunction of all
\ckey{R4} behaviours: an event has occurred, and the choice is resolved. In the postcondition the choice is resolved,
either by synchronisation or termination, and so we take the disjunction of all constituent postconditions. Since
unbounded choice is covered by Definition~\ref{def:ext-choice}, we can denote indexed input prefix for any size of input
domain $A$.
$$a?x\!:\!A \then P(x) ~~~\defs~~~ \Extchoice x \in A @ a.x \then P(x)$$
We can also define the binary operator as a special case: $P \extchoice Q \defs \Extchoice X \in \{P, Q\} @ X$. We next
show how \healthy{R4} and \healthy{R5} filter the various reactive relational operators.
\begin{theorem}[Trace Filtering] \label{thm:filtering} \isalink{https://github.com/isabelle-utp/utp-main/blob/07cb0c256a90bc347289b5f5d202781b536fc640/theories/reactive/utp_rea_healths.thy\#L729}

\vspace{1ex}

\centering
$\begin{aligned}
  \textstyle\healthy{R4}\left(\bigvee_{i \in I} P(i)\right) & = \textstyle\bigvee_{i \in I} \healthy{R4}(P(i)) \\[1ex]
  \healthy{R4}(\csppf{s}{\sigma}{\langle \rangle}) & = \false \\[1ex]
  \healthy{R4}(\csppf{s}{\sigma}{\langle a, ... \rangle}) & = \csppf{s}{\sigma}{\langle a, ... \rangle}
\end{aligned}
\quad
\begin{aligned}
  \textstyle\healthy{R5}\left(\bigvee_{i \in I} P(i)\right) & = \textstyle\bigvee_{i \in I} \healthy{R5}(P(i)) \\[1ex]
  \healthy{R5}(\cspen{s}{\langle \rangle}{E}) & = \cspen{s}{\langle \rangle}{E} \\[1ex]
  \healthy{R5}(\cspen{s}{\langle a, ... \rangle}{E}) & = \false
\end{aligned}$
\end{theorem}
\noindent Both operators distribute through $\bigvee$. Relations that produce an empty trace yield $\false$ under
$\healthy{R4}$ and are unchanged under $\healthy{R5}$. Relations that produce a non-empty trace yield $\false$ for
$\healthy{R5}$, and are unchanged under $\healthy{R4}$. We can now filter the behaviours that do and do not resolve the
choice, as exemplified below.

\begin{example} We consider the calculation of the contract for the action
  $a\!\then\!b\!\then\!\Skip \extchoice c\!\then\!\Skip$. The left branch has two quiescent observations, one waiting
  for $a$, and one waiting for $b$ having performed $a$: its pericondition is
  $\cspen{true}{\langle \rangle}{\{a\}} \lor \cspen{true}{\langle a \rangle}{\{b\}}$.  Application of \healthy{R5} to
  this yields the first disjunct, since the trace has not increased, and application of \healthy{R4} yields the second
  disjunct. For the right branch there is one quiescent observation, $\cspen{true}{\langle \rangle}{\{c\}}$, which
  contributes an empty trace and is $\healthy{R5}$ only. The overall pericondition is
  $$(\cspen{true}{\langle \rangle}{\{a\}} \land \cspen{true}{\langle \rangle}{\{c\}}) \lor \cspen{true}{\langle a
    \rangle}{\{b\}}$$
  which is simply $\cspen{true}{\langle \rangle}{\{a, c\}} \lor \cspen{true}{\langle a \rangle}{\{b\}}$. \qed
\end{example}

\noindent By calculation, we can now prove that $(\theoryset{\healthy{NCSP}}, \extchoice, \Stop)$ forms a commutative
and idempotent monoid, and $\Chaos$, the divergent program, is its annihilator ($P \extchoice \Chaos =
\Chaos$). Sequential composition also distributes from the left and right through external choice, but only when the
choice branches are productive~\cite{Foster17c}, a notion defined below.

\begin{definition}
  A contract $\rc{P_1}{P_2}{P_3}$ is productive when $P_3$ is \healthy{R4}-healthy. \isalink{https://github.com/isabelle-utp/utp-main/blob/07cb0c256a90bc347289b5f5d202781b536fc640/theories/rea_designs/utp_rdes_productive.thy\#L156}
\end{definition}

\noindent A productive contract is one that, whenever it terminates, strictly increases the trace. For example
$a \then \Skip$ is productive, but $\Skip$ is not. Constructs that do not terminate, like $\Chaos$, are also
productive. The imposition of \healthy{R4} ensures that only final observations that increase the trace, or are $\false$,
are admitted.

We next define healthiness condition $\healthy{PCSP}$, which extends $\healthy{NCSP}$ with productivity. We also define
$\healthy{ICSP}$, which formalises instantaneous contracts where the postcondition is \healthy{R5}-healthy and the
pericondition is $\false$.

\begin{definition}[Productive and Instantaneous Healthiness Conditions] \isalink{https://github.com/isabelle-utp/utp-main/blob/f22f260d20c95337e30ee7bb4d6d16e6fda27af0/theories/rea_designs/utp_rdes_instant.thy}
  \begin{align*}
    \healthy{Productive}(P) &\defs P \mathop{\otimes} \rc{\truer}{\true}{\trace > \snil} \\
    \healthy{ISRD1}(P) &\defs P \mathop{\otimes} \rc{\truer}{\false}{\trace = \snil} \\
    \healthy{PCSP} & \defs \healthy{Productive} \circ \healthy{NCSP} \\
    \healthy{ICSP} & \defs \healthy{ISRD1} \circ \healthy{NCSP}
  \end{align*}
\end{definition}
\noindent Here, the $\otimes$ operator combines two contracts by conjoining the pre-, peri-, and postconditions, that is:
$$\rc{P_1}{P_2}{P_3} \otimes \rc{Q_1}{Q_2}{Q_3} = \rc{P_1 \land Q_1}{P_2 \land Q_2}{P_3 \land Q_3}$$ 
Healthiness condition $\healthy{Productive}$ leaves the pre- and periconditions unchanged, but conjoins the
postcondition with $\trace > \snil$ -- the trace must strictly increase. $\healthy{ISRD1}$ similarly leaves the precondition
unchanged, but coerces the pericondition to $\false$ to remove quiescent observations. The postcondition is conjoined
with $\trace = \snil$ to disallow events from occurring. We then define $\healthy{PCSP}$ and $\healthy{ICSP}$ by composing the
former two functions with $\healthy{NCSP}$. These healthiness conditions obey the following equations for reactive
contracts.
\begin{theorem}[\healthy{PCSP} and \healthy{ICSP} contracts] \label{thm:pcsp-icsp-ctr} \isalink{https://github.com/isabelle-utp/utp-main/blob/07cb0c256a90bc347289b5f5d202781b536fc640/theories/sf_rdes/utp_sfrd_healths.thy\#L43}
  \begin{align*}
    \healthy{PCSP}(\rc{P_1}{P_2}{P_3}) &= \rc{P_1}{P_2}{\healthy{R4}(P_3)} \\
    \healthy{ICSP}(\rc{P_1}{P_2}{P_3}) &= \rc{P_1}{\false}{\healthy{R5}(P_3)})
  \end{align*}
\end{theorem}
\noindent Application of \healthy{PCSP} to a reactive contract is equivalent to applying \healthy{R4} to its
postcondition. Application of \healthy{ICSP} to a reactive contract makes the pericondition $\false$, and applies
\healthy{R5} to its postcondition, meaning it can contribute no events. Both $\Skip$ and $x := v$ are
$\healthy{ICSP}$-healthy as they do not contribute to the trace and have no intermediate observations. As shown
below, we can also prove that $\Chaos$ is a right annihilator for \healthy{ICSP} reactive contracts with a feasible
postcondition.

\begin{theorem}[\healthy{ICSP} Annihilator] \label{thm:icsp-anhil} If $\rc{P_1}{P_2}{P_3}$ is $\healthy{ICSP}$-healthy and
  $P_3 \wpR \false = \false$, that is, $P_3$ is feasible, then $\rc{P_1}{P_2}{P_3} \relsemi \Chaos = \Chaos$.
\end{theorem}

\begin{proof}
  \begin{align*}
    \rc{P_1}{P_2}{P_3} \relsemi \Chaos &= \rc{P_1}{\false}{P_3} \relsemi \rc{\false}{\false}{\false} & \text{[Thm.~\ref{thm:pcsp-icsp-ctr}, $\Chaos$ definition]} \\
    &= \rc{P_1 \land P_3 \wpR \false}{\false}{\false} & \text{[Thm.~\ref{thm:rc-comp}]}\\
    &= \rc{\false}{\false}{\false} & \text{[Feasibility]} \\
    &= \Chaos & \qedhere
  \end{align*}
\end{proof}

\noindent This follows essentially because $\Chaos$ annihilates any postcondition, and an $\healthy{ICSP}$-healthy
contract already has a $\false$ pericondition. So, if the postcondition $P_3$ is feasible, then the precondition of the
overall contract also reduces to $\false$. An example of a feasible postcondition is any assignment to variables, such
as $\csppf{\ptrue}{\sigma}{\snil}$ (see Theorem~\ref{thm:evwp}), and consequently it is straightforward to show that
$\assignsC{\sigma} \relsemi \Chaos = \Chaos$.

We can also use $\healthy{PCSP}$ and $\healthy{ICSP}$ to prove the following laws of external choice.

\begin{theorem}[External Choice Distributivity] \label{thm:ext-choice-distrib} \isalink{https://github.com/isabelle-utp/utp-main/blob/07cb0c256a90bc347289b5f5d202781b536fc640/theories/sf_rdes/utp_sfrd_extchoice.thy\#L669}
  \begin{align*}
    (\Extchoice i\!\in\!I @ P(i)) \relsemi Q &~=~ \Extchoice i\!\in\!I @ (P(i) \relsemi Q) & ~~[\text{if, } \forall i\!\in\!I, P(i) \text{ is } \healthy{PCSP} \text{ healthy}] \\[.3ex]
    P \relsemi (\Extchoice i\!\in\!I @ Q(i)) &~=~ \Extchoice i\!\in\!I @ (P \relsemi Q(i)) & ~~[\text{if } P \text{ is } \healthy{ICSP} \text{ healthy}]
  \end{align*}
\end{theorem}

\noindent The first law follows because every $P(i)$, being productive, must resolve the choice before terminating, and
thus it is not possible to reach $Q$ before this occurs. It generalises the standard guarded choice distribution law for
CSP~\cite[page 211]{Hoare&98}. The second law follows for the converse reason: since $P$ cannot resolve the choice with
any of its behaviour, it is safe to execute it first. Productivity also forms an important criterion for guarded
recursion that we use in \S\ref{sec:iter} to calculate fixed points.

$\healthy{PCSP}$ is closed under several operators.

\begin{theorem}[Productive Constructions] \isalink{https://github.com/isabelle-utp/utp-main/blob/07cb0c256a90bc347289b5f5d202781b536fc640/theories/sf_rdes/utp_sfrd_prog.thy\#L654}
   \begin{itemize}
     \item $\Miracle$, $\Chaos$, $\Stop$, and $\ckey{Do}(a)$ are all \healthy{PCSP} healthy;
     \item $b \guard P$ is \healthy{PCSP} if $P$ is \healthy{PCSP};
     \item $P \relsemi Q$ is \healthy{PCSP} if either $P$ or $Q$ is \healthy{PCSP};
     \item $\Intchoice i \in I @ P(i)$ is \healthy{PCSP} if, for all $i \in I$, $P(i)$ is \healthy{PCSP};
     \item $\Extchoice i \in I @ P(i)$ is \healthy{PCSP} if, for all $i \in I$, $P(i)$ is \healthy{PCSP}.
   \end{itemize}
\end{theorem}
\noindent With these results, calculation of contracts for external choice is supported, and a notion of
productivity, with relevant laws, is defined. In the next section we use the latter for calculation of contracts for
while-loops.

\section{While Loops and Reactive Invariants}
\label{sec:iter}
In this section, we introduce a useful pattern for reasoning about iterative reactive programs with potentially
non-terminating behaviour. As indicated in \S\ref{sec:circus-rc}, the healthiness condition $\healthy{NCSP}$ is both
idempotent and continuous, and consequently our theory of stateful-failure reactive designs forms a complete
lattice~\cite{Foster17c}. As for $\healthy{NSRD}$, the top of the lattice is $\Miracle$, and the bottom is
$\Chaos$. Through the Knaster-Tarski theorem~\cite{Tarski55}, we also obtain operators for constructing both weakest
($\mu X @ F(X)$) and strongest ($\nu X @ F(X)$) fixed-points. Iterative programs can be constructed using the reactive
while loop, defined below.
\begin{definition}[Reactive While Loop] \label{def:reawhile}
  $\while{b}{P} ~~\defs~~ (\mu X @ \rconditional{P \relsemi X}{b}{\Skip}).$
\end{definition}
\noindent We use the weakest fixed-point so that an infinite loop with no observable activity corresponds to the
divergent action $\Chaos$, rather than $\Miracle$, which represents infeasible behaviour. The weakest fixed-point is
governed principally by the following standard theorems~\cite{Hoare&98}.

\begin{theorem}[Weakest Fixed-Point] \label{thm:wfp}
  \begin{align*}
    F(P) \refinedby P \implies (\mu X @ F(X)) & \refinedby P \\
    (\forall X @ F(X) \refinedby X \implies P \refinedby X) \implies P & \refinedby (\mu X @ F(X))
  \end{align*}
\end{theorem}

\noindent These theorems demonstrate that $\mu X @ F(X)$ is the greatest lower bound of the prefixed points of $F$, that is
$\bigsqcap \{X ~|~ F(X) \refinedby X\}$. We can therefore deduce that $(\mu X @ X) = \Chaos$, whereas in contrast
$(\nu X @ X) = \Miracle$. Similarly, we can calculate that $(\while{\ptrue}{~x := x + 1~}) = \Chaos$.

\begin{proof}
  Since $\Chaos$ is the bottom of the lattice, it suffices to show that
  $(\while{\ptrue}{~x := x + 1~}) \refinedby \Chaos$. We first note that $x := x + 1 \relsemi \Chaos = \Chaos$ by
  Theorem~\ref{thm:icsp-anhil}. We conclude with the following calculation.
\begin{align*}
  \while{\ptrue}{~x := x + 1~}
  &= (\mu X @ \rconditional{x := x + 1 \relsemi X}{\ptrue}{\Skip}) & \text{[Def.~\ref{def:reawhile}]} \\
  &= (\mu X @ x := x + 1 \relsemi X) & \text{[Conditional]} \\
  &\refinedby \Chaos & \text{[Thm.~\ref{thm:wfp}]} & \qedhere
\end{align*}
\end{proof}

\noindent The issue here is not simply the $\ptrue$ loop condition. Unlike its imperative counterpart, the reactive
while loop pauses for interaction with its environment, and therefore infinite executions are observable and potentially
useful. The problem is that the body of this infinite loop also produces no events, and therefore all observations reduce to
$\Chaos$. In general, to prove a loop refinement we first need to find a suitable variant function that demonstrates
termination~\cite{Mor1996}. For purely stateful interactions, it is usually necessary that this variant is manually
created. However, for productive reactive behaviour, where the body of a loop is guarded by events, we can identify a
variant automatically using Hoare and He's fixed-point theorem~\cite{Hoare&98}, as explained next.

A fixed-point ($\mu X @ F(X)$) is guarded provided at least one event is contributed to the trace by $F$ prior to it
reaching $X$. For instance, $\mu X @ a \then X$ is guarded, but $\mu X @ y := 1 \relsemi X$ is not. An unguarded
recursive action has the potential to introduce divergence, where neither an event nor termination is observable. Hoare
and He's fixed-point theorem~\cite[theorem 8.1.13, page 206]{Hoare&98} states that if $F$ is guarded, then there is a
unique fixed-point of $F$ and hence $(\mu X @ F(X)) = (\nu X @ F(X))$. This effectively shows that the loop is
productive, and so even though it is infinite we can observe any finite sequence of interactions. The reduction of $\mu$
to $\nu$ means that, provided $F$ is also continuous, we can invoke Kleene's fixed-point theorem~\cite{Lassez82} to
calculate $\nu F$. Our previous result~\cite{Foster17c} shows that if $P$ is productive, then $\lambda X @ P \relsemi X$
is guarded, and so we can calculate its fixed-point. We now generalise this for the function given in
Definition~\ref{def:reawhile}.

\begin{theorem} If $P$ is productive, then $(\mu X @ \rconditional{P \relsemi X}{b}{\Skip})$ is guarded. \isalink{https://github.com/isabelle-utp/utp-main/blob/07cb0c256a90bc347289b5f5d202781b536fc640/theories/rea_designs/utp_rdes_guarded.thy\#L149}
\end{theorem}
\begin{proof} In addition to our previous theorem~\cite{Foster17c}, we use the following properties:
  \vspace{-1ex}
  \begin{itemize}
    \item If $X$ is not mentioned in $P$ then $\lambda X @ P$ is guarded;
    \item If $F$ and $G$ are both guarded, then $\lambda X @ \rconditional{F(X)}{b}{G(X)}$ is guarded. \qedhere
  \end{itemize} 
\end{proof}
\noindent This allows us to convert the fixed-point into an iterative form. In particular, we can prove the following
theorem that expresses it in terms of the Kleene star.
\begin{theorem} If $P$ is \healthy{PCSP} healthy then \label{thm:whilekleene}
  $\while{b}{P} = (\Casm{b} \relsemi P)\bm{\star} \relsemi \Casm{\neg b}$. \isalink{https://github.com/isabelle-utp/utp-main/blob/07cb0c256a90bc347289b5f5d202781b536fc640/theories/sf_rdes/utp_sfrd_prog.thy\#L87}
\end{theorem}
\noindent Here, $\Casm{b} \defs \rconditional{\Skip}{b}{\Miracle}$ denotes a reactive program state test (cf. $\casm{b}$,
which is for reactive relations). This theorem is similar to the usual imperative definition in Kleene Algebra with
Tests~\cite{Kozen97,Armstrong2015,Gomes2016}. It relies on productivity of $P$, though the condition $b$ can be used to
guard $P$ and therefore prune away any unproductive behaviours that violate $b$. In Theorem~\ref{thm:whilekleene}, $P$
is executed multiple times when $b$ is true initially, but each run concludes when $b$ is false. However, due to the
embedding of reactive behaviour, there is more going on than meets the eye; the next theorem shows how to calculate an
iterative contract.

\begin{theorem} If $P_3$ is \healthy{R4} healthy then \label{thm:whilecalc} \isalink{https://github.com/isabelle-utp/utp-main/blob/07cb0c256a90bc347289b5f5d202781b536fc640/theories/sf_rdes/utp_sfrd_prog.thy\#L93}
  \begin{align*}
  \while{b}{\rc{\!P_1\!}{\!P_2\!}{\!P_3\!}} =& \rc{(\casm{b}\!\relsemi\!P_3)\bm{\star}\!\relsemi\!\casm{b} \wpR P_1}{(\casm{b} \relsemi P_3)\bm{\star}\!\relsemi\!\casm{b}\!\relsemi\!P_2}{(\casm{b}\!\relsemi\!P_3)\bm{\star}\!\relsemi\!\casm{\neg b}}
  \end{align*}
\end{theorem}
\noindent The precondition requires that any number of $P_3$ iterations, where $b$ is initially true, satisfies $P_1$. This
ensures that the contract does not violate its own precondition from one iteration to the next. The pericondition states
that intermediate observations have $P_3$ executing several times, with $b$ true, and following this $b$ remains true and
the contract is quiescent ($P_2$). The postcondition is similar, but after several iterations, $b$ becomes false and the
loop terminates, which is the standard relational form of a while loop.

Theorem~\ref{thm:whilecalc} can be used to prove a refinement introduction law for the reactive while loop. This employs
``reactive invariant'' relations, which describe how both the trace and state variables are permitted to evolve.
\begin{theorem}$\rc{\!I_1\!}{\!I_2\!}{\!I_3\!} \refinedby \while{b}{\rc{\!Q_1\!}{\!Q_2\!}{\!Q_3\!}}$
  provided that: \label{thm:rea-inv} \isalink{https://github.com/isabelle-utp/utp-main/blob/07cb0c256a90bc347289b5f5d202781b536fc640/theories/rea_designs/utp_rdes_prog.thy\#L746}
  \vspace{-.5ex}
  \begin{enumerate} \itemsep4pt
    \item $Q_3$ is $\healthy{R4}$-healthy, so that the reactive contract is productive;
    \item the assumption is weakened $((\casm{b} \relsemi Q_3)\bm{\star} \wpR (b \implies Q_1) \refinedby I_1)$;
    \item when $b$ holds, $Q_2$ establishes the $I_2$ pericondition invariant $(I_2 \refinedby (\casm{b} \relsemi Q_2))$ and,
      $Q_3$ maintains it $(I_2 \refinedby \casm{b} \relsemi Q_3 \relsemi I_2)$;
    \item postcondition invariant $I_3$ is established when $b$ is false ($I_3 \refinedby \casm{\neg b}$) and $Q_3$
      establishes it when $b$ is true ($I_3 \refinedby \casm{b} \relsemi Q_3 \relsemi I_3$).
  \end{enumerate}
\end{theorem}
\begin{proof}
  By application of refinement introduction, with Theorems~\thmeqref{def:wka}{thm:starinductr} and \ref{thm:whilecalc}.
\end{proof}
\noindent Theorem~\ref{thm:rea-inv} shows the conditions under which an iterated contract satisfies an invariant
contract $\rc{\!I_1\!}{\!I_2\!}{\!I_3\!}$. Relations $I_2$ and $I_3$ are reactive invariants that must hold in quiescent
and final observations. Both can refer to $\state$ and $\trace$, $I_2$ can additionally refer to $\refu'$, and $I_3$ to
$\state'$. There is no need to supply a variant since productivity guarantees the existence of a descending
approximation chain for the iteration~\cite{Foster17c}. Combined with the results from \S\ref{sec:circus-rc} and
\S\ref{sec:ext-choice}, this result forms the basis for a proof strategy for iterative reactive programs.

\section{Parallel Composition}
\label{sec:par}
In this section we extend our calculational approach to one of the most challenging operators: parallel composition. We
build on the parallel-by-merge scheme~\cite{Hoare&98}, $P \parallel_M Q$, where the semantics is expressed in terms of a
merge predicate $M$ that describes how the observations of each parallel program, $P$ and $Q$, should be
merged~\cite{Foster16a}. We create a specialised law for parallel composition of stateful-failure reactive designs, that
merges the pre-, peri-, and postconditions, and show how the $\mathcal{I}$, $\Phi$, and $\mathcal{E}$ operators are
merged. We use the strategy on a number of examples, and prove characteristic algebraic theorems for parallel
composition.

The contents of the first two subsections, \S\ref{sec:pbm} and \S\ref{sec:prd}, contain some restatements of theorems we
have previously proved~\cite{Foster17c}, which are included for the purpose of self-containment and
explanation. \S\ref{sec:psfrd} onwards, which specialises to stateful-failure reactive designs, is entirely novel.

\subsection{Parallel-by-Merge}
\label{sec:pbm}

We recall the parallel-by-merge operator~\cite{Foster17c}. It employs the $\psep{P}{n}$
construct, which renames all dashed variables of $P$ by adding an index $n$, so that they can be distinguished
from other indexed variables\footnote{These are called ``separating simulations'' in~\cite[page 172]{Hoare&98}, and are
  denoted using special relations called $U0$ and $U1$.}.

\begin{definition}[Parallel-by-Merge] $P \parallel_M Q ~\defs~ (\psep{P}{0} \land \psep{Q}{1} \land \uv' = \uv) \relsemi M$ \isalink{https://github.com/isabelle-utp/utp-main/blob/07cb0c256a90bc347289b5f5d202781b536fc640/utp/utp_concurrency.thy\#L275}
\end{definition}

\noindent This operator effectively splits the observation space into three identical segments: one for $P$, one for
$Q$, and a third that is identical to the original input. Relation $M$ then takes the outputs from $P$, $Q$, and the
original input $\uv$, and merges them into a single output. Here, $\uv$ is a special variable that denotes the entirety
of the state space. The dataflow of this operator is depicted in Figure~\ref{fig:pbm}, which illustrates the definition
for an example with three variables, $x$, $y$, and $z$.

\begin{figure}
  \centering\includegraphics[width=.5\linewidth]{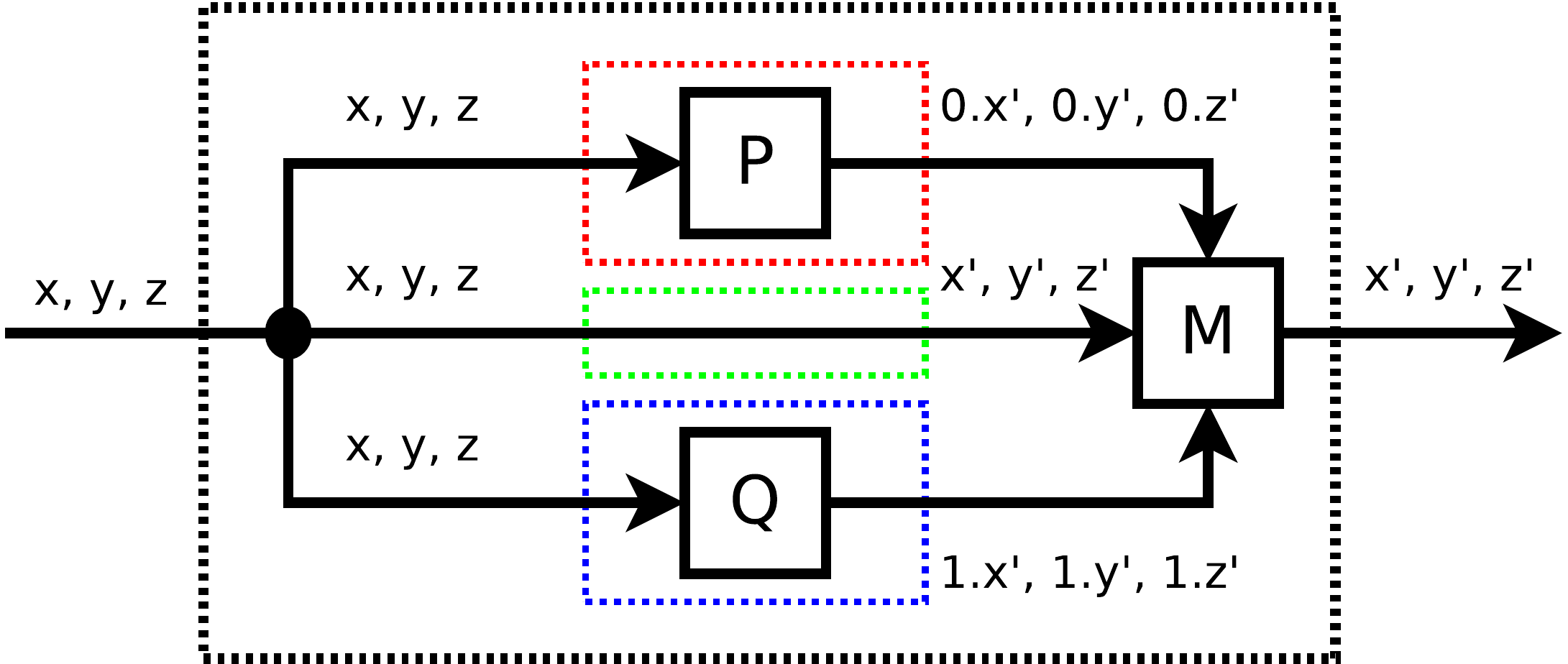}

  \caption{Parallel-by-merge dataflow}
  \label{fig:pbm}

\end{figure}

The outputs of $P$ and $Q$ are distinguished by a numeral prefix. If $P$ and $Q$ act on an observation space $\src$,
then $M$ is a heterogeneous relation of type $[\src \times \src \times \src, \src]\urel$, which relates three input
copies of $\src$ with a single output $\src$. The merge predicate therefore refers to variables from $P$, using the
$0.x$ notation, variables from $Q$, using $1.x$, initial variables, as usual written as $x$, and final variables, as
$x'$.

A substantial advantage of using parallel-by-merge is that several theorems can be proven for the generic
operator. Below, we highlight some of the most important theorems.

\begin{theorem}[Parallel-by-Merge Laws] \label{thm:pbm-laws} \isalink{https://github.com/isabelle-utp/utp-main/blob/07cb0c256a90bc347289b5f5d202781b536fc640/utp/utp_concurrency.thy\#L438}
  \begin{minipage}{0.5\textwidth}
  \begin{align*}
    \left(\bigsqcap_{i \in I} ~ P(i)\right) \parallel_M Q &= \bigsqcap_{i \in I} ~ \left(P(i) \parallel_M Q\right) \\[.5ex]
    \false \parallel_M P &= \false
  \end{align*}
  \end{minipage}\begin{minipage}{0.5\textwidth}
  \begin{align*}
    P \parallel_M \left(\bigsqcap_{i \in I} ~ Q(i)\right) &= \bigsqcap_{i \in I} ~ \left(P \parallel_M Q(i)\right) \\[.5ex]
    P \parallel_M \false &= \false
  \end{align*}
  \end{minipage}
  \vspace{1ex}

  \quad~$P_1 \refinedby P_2 \land Q_1 \refinedby Q_2 ~\implies~ (P_1 \parallel_M Q_1) \refinedby (P_1 \parallel_M Q_1)$
\end{theorem}

\noindent Parallel-by-merge distributes through nondeterministic choice ($\bigsqcap$) from both the left and right,
regardless of the merge predicate $M$. Since $\bigsqcap$ corresponds to $\exists$ and also $\vee$, we can similarly
distribute through an existential quantification and a disjunction. The miraculous relation $\false$ is both a left and
right annihilator for parallel composition, which is also monotonic with respect to refinement in both arguments.

Parallel-by-merge may or may not be commutative, depending on the merge predicate. A helpful scheme can be used for
proving that parallel composition is commutative, which reduces this to a property of the merge predicate. We adopt a
similar approach to \cite{Hoare&98}, but give an account that is more algebraic in nature. We first define the
following auxiliary operator.

\begin{definition}[Merge Swap] $\ckey{sw} \defs \uv' = \uv \land 0.\uv' = 1.\uv \land 1.\uv' = 0.\uv$ \isalink{https://github.com/isabelle-utp/utp-main/blob/07cb0c256a90bc347289b5f5d202781b536fc640/utp/utp_concurrency.thy\#L131}

\end{definition}

\noindent The relation $\ckey{sw}$ swaps the outputs from the left- and right-hand sides, whilst keeping the initial
values ($\uv$) the same. Using $\ckey{sw}$, we can prove the following property of parallel-by-merge.

\begin{theorem}[Parallel-by-Merge Swap] $P \parallel_{\ckey{\small sw}\, \relsemi M} Q = Q \parallel_M P$ \isalink{https://github.com/isabelle-utp/utp-main/blob/9bbb6e407a4224591b2370605186ca0a45b718f9/utp/utp_concurrency.thy\#L385}
\end{theorem}

This theorem shows that precomposing a merge predicate with $\ckey{sw}$ effectively commutes the arguments $P$ and
$Q$. A corollary of this~\cite{Foster16a}, given below, shows how this can be used to demonstrate commutativity.

\begin{theorem} $P \parallel_M Q = Q \parallel_M P$ provided that $\ckey{sw} \relsemi M = M$ 
\end{theorem}

\noindent This theorem shows how proof of commutativity can be reduced to a property of the merge predicate. Specifically, if
swapping the order of the inputs to the merge predicate has no effect then it is a symmetric merge, and consequently
parallel composition is commutative.

\subsection{Parallel Reactive Designs}
\label{sec:prd}

In previous work~\cite{Foster17c}, we have used parallel-by-merge to prove a general theorem for composing reactive
designs. As for the sequential operators, this develops operators that respectively merge the pre-, peri-, and
postconditions of the corresponding reactive contract. The theorem below, reproduced from \cite{Foster17c}, shows how we
may calculate a parallel reactive contract using these operators.

\begin{theorem}[Reactive Design Parallel Composition] \label{thm:rdespar} \isalink{https://github.com/isabelle-utp/utp-main/blob/07cb0c256a90bc347289b5f5d202781b536fc640/theories/rea_designs/utp_rdes_parallel.thy\#L769}
  \begin{align*}
    &\rc{P_1}{P_2}{P_3} \rcpar{M} \rc{Q_1}{Q_2}{Q_3} = \\[.5ex]
    & \qquad \rc{\begin{array}{l} 
          (P_1 \rimplies P_2) \wppR{M} Q_1 \land \\
          (P_1 \rimplies P_3) \wppR{M} Q_1 \land \\
          (Q_1 \rimplies Q_2) \wppR{M} P_1 \land \\
          (Q_1 \rimplies Q_3) \wppR{M} P_1
       \end{array}}{\begin{array}{l}
          P_2 \rcmergee{M} Q_2 ~\lor~ \\
          P_3 \rcmergee{M} Q_2 ~\lor~ \\
          P_2 \rcmergee{M} Q_3
       \end{array}
       }{P_3 \parallel_{\text{\tiny M}} Q_3}
  \end{align*}
\end{theorem}
\noindent 
This complex law describes how the pre-, peri-, and postconditions are merged by the parametric reactive design parallel
composition operator $\rcpar{M}$. Internally, this operator merges the observational variables $ok'$ and $wait'$ of $P$
and $Q$ using the parallel-by-merge operator and a bespoke merge predicate~\cite{Foster17c}. Again, as we noted in
\S\ref{sec:circus-rc}, every $\healthy{NCSP}$ relation corresponds to a reactive contract, and therefore this law
applies to any stateful-failure reactive design. Here, $M$ is an ``inner merge predicate''~\cite{Foster17c}, which needs
to deal only with observational variables like $\trace$, $\state$, and $\refu$; the variables $ok$ and $wait$ having
already been merged by $\rcpar{M}$, which constructs the ``outer merge predicate'', to which parallel-by-merge is
applied.

The precondition of the composite contract in Theorem~\ref{thm:rdespar} captures the possible divergent behaviours that
both $P$ and $Q$ permit. There are four conjuncts in the precondition, as we require that neither the peri-
nor the postcondition can permit divergent behaviour disallowed by its opposing precondition.

The predicate $A \wppR{M} B$, standing for ``weakest rely'', is a reactive condition that describes the weakest context
in which reactive relation $A$ does not violate the reactive condition $B$. It is analogous to the reactive weakest
precondition operator, $\wpR$ (outlined in \S\ref{sec:rea-prog}), but is defined with respect to parallel composition
rather than sequential composition. Specifically, whereas $\wpR$ gives the weakest condition in a sequential context,
$\wppR{M}$ gives the weakest condition in a parallel context. Its definition is given below.

\begin{definition}[Weakest Rely Condition] $P \wppR{M} Q ~~\defs~~ \negr ((\negr Q) \parallel_{\hbox{\small $M$} \relsemi \truer} P)$ \label{def:wppR} \isalink{https://github.com/isabelle-utp/utp-main/blob/90ec1d65d63e91a69fbfeeafe69bd7d67f753a47/theories/rea_designs/utp_rdes_parallel.thy\#L667}
\end{definition}

\noindent This merges the traces that $Q$ does not admit ($\negr Q$) with those of $P$ using the merge predicate $M$
composed with $\truer$, which makes $M$ extension closed. It determines all the behaviours permitted by merge predicate
$M$ that are enabled by $P$ and yet denied by $Q$. We then negate the overall relation to obtain the reactive
precondition. The operator obeys several related laws shown below.

\begin{theorem}[Weakest Rely Laws] \label{thm:wrlaws} \isalink{https://github.com/isabelle-utp/utp-main/blob/07cb0c256a90bc347289b5f5d202781b536fc640/theories/rea_designs/utp_rdes_parallel.thy\#L720}
  $$ \false \wppR{M} P = \truer \qquad
    P \wppR{M} \truer = \truer \qquad 
    \left(\bigvee_{i \in I} ~ P(i)\right) \wppR{M} Q = \left(\bigwedge_{i \in I} ~ (P(i) \wppR{M} Q)\right)
  $$
\end{theorem}

\noindent The laws show that (1) a miraculous relation satisfies any precondition, (2) any reactive relation satisfies a
true precondition, and (3) the weakest rely condition of a disjunction of relations is the conjunction of their weakest
rely conditions. These results are similar to those for $\wpR$

The pericondition in Theorem~\ref{thm:rdespar} is a disjunction of three terms that calculate possible quiescent merged
behaviours. Parallel composition is quiescent when at least one of $P$ and $Q$ is quiescent, and so the three conjuncts
characterise quiescence in both, in $Q$ only, and in $P$ only, respectively. The
$P \rcmergee{M} Q ~~\defs~~ P \parallel_{\exists \state' @ M} Q$ operator is an intermediate merge operator, which
restricts access to state (see~\cite[\S6.6]{Foster17c}). Finally, the overall contract can only terminate when both $P$
and $Q$ do, and so the postcondition simply merges their respective postconditions.

Using these laws, we can show that $\Miracle$ is always a annihilator for parallel composition, regardless of the inner
merge predicate~\cite{Foster17c}. The proof exemplifies the calculational approach for parallel composition

\begin{theorem} \label{thm:miracle-anhil} $\Miracle \rcpar{M} P = \Miracle$ \isalink{https://github.com/isabelle-utp/utp-main/blob/07cb0c256a90bc347289b5f5d202781b536fc640/theories/rea_designs/utp_rdes_parallel.thy\#L813} \end{theorem}

\begin{proof}
  \begin{align*}
    \Miracle \rcpar{M} P 
    &= \rc{\truer}{\false}{\false} \rcpar{M} \rc{P_1}{P_2}{P_3} \\
    &= \rc{\begin{array}{l} 
        (\truer \rimplies \false) \wppR{M} P_1 \land \\
        (\truer \rimplies \false) \wppR{M} P_1 \land \\
        (P_1 \rimplies P_2) \wppR{M} \truer \land \\
        (P_1 \rimplies P_3) \wppR{M} \truer
      \end{array}}{
      \begin{array}{l}
        \false \rcmergee{M} P_2 ~\lor~ \\
        \false \rcmergee{M} P_2 ~\lor~ \\ 
        \false \rcmergee{M} P_3
      \end{array}
    }{\false \parallel_{\text{\tiny M}} P_3} & [\ref{thm:rdespar}] \\
    &= \rc{\begin{array}{l} 
        \false \wppR{M} P_1 \land \false \wppR{M} P_1 \land \\
        \truer \land \truer
      \end{array}}{
      \begin{array}{l}
        \false ~\lor~  \false ~\lor~ \false
      \end{array}
    }{\false} & [\ref{thm:pbm-laws},\ref{thm:wrlaws}] \\
    &= \rc{\begin{array}{l} 
        \truer \land \truer
      \end{array}}{
      \begin{array}{l}
        \false
      \end{array} 
    }{\false} & [\ref{thm:wrlaws}] \\
    &= \rc{\truer}{\false}{\false} \\
    &= \Miracle & \qedhere 
  \end{align*} 
\end{proof}

\noindent In this case, the pericondition and the postcondition both reduce to $\false$, since by
Theorem~\ref{thm:pbm-laws} the merge of any relation with $\false$ reduces to $\false$. The four clauses in the
precondition all reduce to $\truer$ by the weakest rely laws of Theorem~\ref{thm:wrlaws}. Thus the entire relation
reduces to the $\Miracle$ contract. We will next specialise this calculational approach to stateful-failure reactive
designs.

\subsection{Parallel Stateful-Failure Reactive Designs}
\label{sec:psfrd}

Our parallel composition operator is adopted from \Circus~\cite{Oliveira&09} and has the general form
$$P \sfpar{ns_1}{cs}{ns_2} Q$$ for actions $P$ and $Q$, event set $cs \subseteq \textit{Event}$, and name-sets $ns_1$ and $ns_2$. $P$ and $Q$ both act on the same state space $\src$, and have the same event alphabet ($\textit{Event}$). 
Like for parallel composition in CSP, $P$ and $Q$ must synchronise on events contained in $cs$, but independently engage
in events outside $cs$. Since $P$ and $Q$ also have states, we must describe how to merge their final states. We do not
permit sharing, and so require partitioning of the state into two independent regions, characterised by two disjoint
variable name sets $ns_1$ and $ns_2$. The final state is then the composition of the two regions. We model these name
sets using independent lenses~\cite{Foster16a} from Isabelle/UTP, that is, $ns_1 : V_1 \lto \src$, and
$ns_2 : V_2 \lto \src$, for some $V_1$ and $V_2$ (see \S\ref{sec:isabelleutp}).

As usual~\cite{Oliveira&09}, we define a few abbreviations for the operator.

\begin{definition}[Parallel Composition Abbreviations] \label{def:parabbrev} \isalink{https://github.com/isabelle-utp/utp-main/blob/07cb0c256a90bc347289b5f5d202781b536fc640/theories/circus/utp_circus_parallel.thy\#L834}
  \begin{align*}
    P \sfpare{cs} Q &\defs P \sfpar{\lzero}{cs}{\lzero} Q \\
    P \interleave Q &\defs P \sfpare{\emptyset} Q
  \end{align*}
\end{definition}
\noindent The operator $P \sfpare{cs} Q$ synchronises on $cs$, but ignores the final state of both $P$ and $Q$. It is
therefore broadly equivalent to CSP parallel composition when applied to stateless actions. It uses the special $\lzero$
lens for the name sets, which characterises an empty region of the state space. The interleaving operator
$P \interleave Q$ synchronises on none of the events, and requires independent activity for $P$ and $Q$. We denote the
general operator using the reactive design parallel-by-merge operator, as shown below.

\begin{definition}[Parallel Composition] \label{def:parcomp} Let $ns_1 : \view_1 \lto \src$ and
  $ns_2 : \view_2 \lto \src$ be lenses that characterise disjoint regions, $\view_1$ and $\view_2$, of the state space
  $\src$ (that is, $ns_1 \lindep ns_2$), and let $cs$ be a set of events. Parallel composition is then defined as follows: \isalink{https://github.com/isabelle-utp/utp-main/blob/07cb0c256a90bc347289b5f5d202781b536fc640/theories/circus/utp_circus_parallel.thy\#L9} %
  \begin{align*}
  P \sfpar{ns_1}{cs}{ns_2} Q &\defs P \rcpar{\cmrgn} Q \\[2ex]
  \text{where} ~ \cmrg{ns_1}{cs}{ns_2}
                             &\defs
  \left(
  \begin{array}{l}
    \trace \in 0.\trace \parallel_{cs} 1.\trace \\ 
    \land 0.\trace \project cs = 1.\trace \project cs \\
    \land  \refu' \subseteq ((0.\refu \cup 1.\refu) \cap cs) \cup ((0.\refu \cap 1.\refu) \setminus cs) \\
    \land \state' = \lovrd{\lovrd{\state}{0.\state}{ns_1}}{1.\state}{ns_2}
  \end{array}\right)
  \end{align*}
\end{definition}
\noindent Here, $\cmrgn$ is an inner merge predicate~\cite{Foster17c} with arguments $ns_1$, $cs$, and $ns_2$, which we
omit when they can be determined from the context. It defines how the traces, states, and refusal sets from $P$ and $Q$
are merged. It is adapted from the original \Circus merge predicate~\cite{Oliveira2005-PHD,Oliveira&09}, which also
defines the function $t_1 \parallel_{cs} t_2$ that specifies the set of traces obtained by merging traces $t_1$ and
$t_2$, synchronising on the events in $cs$. For completeness, we define this recursive function below, adapting slightly
the original definition\footnote{Specifically, the definition of $\parallel_{cs}$ may be found in Oliveira's
  thesis~\cite{Oliveira2005-PHD}, Appendix B on page 183. It is based on the trace merge operator defined by Roscoe in
  Section~2.4, page 70, of~\cite{Roscoe2005}.}~\cite{Oliveira2005-PHD}.

\begin{definition}[Trace Merge Function] We define $\parallel_{cs} : \seq E \to \seq E \to \power(\seq~E)$ to be the
  least function that satisfies the following equations: \isalink{https://github.com/isabelle-utp/utp-main/blob/07cb0c256a90bc347289b5f5d202781b536fc640/theories/circus/utp_circus_traces.thy\#L9}
  \begin{align*}
    \snil \parallel_{cs} \snil                &= \{ \snil \} \\
    (e \cons t) \parallel_{cs} \snil          &= \left(\conditional{\{ \snil \}}{e \in cs}{\left(\{\langle e \rangle\} \frown (t \parallel_{cs} \snil)\right)}\right) \\
    \snil \parallel_{cs} (e \cons t)          &= \left(\conditional{\{ \snil \}}{e \in cs}{\left(\{\langle e \rangle\} \frown (\snil \parallel_{cs} t)\right)}\right) \\
    (e \cons t_1) \parallel_{cs} (e \cons t_2) &= \conditional{(\langle e \rangle \frown (t_1 \parallel_{cs} t_2))}{e \in cs}{\left(\{\langle e \rangle\} \frown (t_1 \parallel_{cs} (e \cons t_2) \cup (e \cons t_1) \parallel_{cs} t_2)\right)} \\
    (e_1 \cons t_1) \parallel_{cs} (e_2 \cons t_2) &= 
      \left(
      \begin{array}{l}
        \left(\conditional{\{ \snil \}}{e_2 \in cs}{\left(\{\langle e_2 \rangle\} \frown ((e_1 \cons t_1) \parallel_{cs} t_2)\right)}\right)\\
        \quad\conditional{}{e_1 \in cs}{} \\
        \left(
        \begin{array}{l}
          \left(\{\langle e_1 \rangle\} \frown (t_1 \parallel_{cs} (e_2 \cons t_2))\right) \\
          \quad\conditional{}{e_2 \in cs}{} \\
          \left((\{\langle e_1 \rangle\} \frown (t_1 \parallel_{cs} (e_2 \cons t_2))) \cup (\{\langle e_2 \rangle\} \frown ((e_1 \cons t_1) \parallel_{cs} t_2))\right)
        \end{array}\right)
      \end{array}
      \right) & e_1 \neq e_2
  \end{align*}

  where $ts_1 \frown ts_2 \defs \left\{t_1 \cat t_2 | t_1 \in ts_1 \land t_2 \in ts_2\right\}$

\end{definition}
\noindent Here, the operator $x \cons xs$ constructs a sequence by composing an element $x$ with an existing sequence $xs$. The
trace merge function $t_1 \parallel_{cs} t_2$ produces the set of maximal possible merges from every pair of traces;
that is the traces that include the maximum possible number of events from both $t_1$ and $t_2$, ordered to reflect
synchronisation on th events in $cs$~\cite{Oliveira2005-PHD}. If an event is encountered in $cs$, then both traces must
agree to allow this event simultaneously for behaviour to progress. For any other events, all possible interleavings are
recorded.

The merge predicate $\cmrgn$ in Definition~\ref{def:parcomp} has four conjuncts. The first conjunct states that any
permissible trace $\trace$ arises from merging the constituent traces $0.\trace$ and $1.\trace$. The second ensures that
the same synchronisations on events from $cs$ occur in both $0.\trace$ and $1.\trace$ in the same order. Filter function
$t \project cs$ returns the elements of sequence $t$ that are contained in $cs$, whilst retaining the order and number
of occurrences. A consequence of the second conjunct is that the resulting trace contains, in a suitable order, all the
events from both constituents. The third conjunct requires that the overall refusal is either a subset of the set of
synchronised events independently refused ($(0.\refu \cup 1.\refu) \cap cs$), or the non-synchronised events refused by
both ($(0.\refu \cap 1.\refu) \setminus cs$). The fourth conjunct constructs the final state by merging the $ns_1$
region of $P$'s state, the $ns_2$ region of $Q$'s state, and the remaining region from the initial state. It uses the
lens override operator $\lovrd{s_1}{s_2}{ns}$ from Definition~\ref{def:lovrd}.

The parallel operator of Definition~\ref{def:parcomp} is not, in general, commutative due to its asymmetric partitioning
of the state space. However, we can prove a useful theorem of the inner merge predicate.

\begin{theorem}[Swap Inner Merge] \label{thm:swapim} If $ns_1 \lindep ns_2$ then $\ckey{sw} \relsemi \cmrg{ns_1}{cs}{ns_2} = \cmrg{ns_2}{cs}{ns_1}$ \isalink{https://github.com/isabelle-utp/utp-main/blob/07cb0c256a90bc347289b5f5d202781b536fc640/theories/circus/utp_circus_parallel.thy\#L170}
\end{theorem}
\noindent If we precompose $\cmrgn$ with the swap relation ($\ckey{sw}$), then this amounts to switching the name
sets. The proof of this depends on the commutativity of $\parallel_{cs}$, a property that is proved
in~\cite{Oliveira&09}, and on Theorem~\ref{thm:ovrd} to switch the name set lenses. A corollary of
Theorem~\ref{thm:swapim} is a quasi-commutativity theorem for our parallel composition operator.

\begin{theorem} \label{thm:nparcomm} If $ns_1 \lindep ns_2$ then $P \sfpar{ns_1}{cs}{ns_2} Q = Q \sfpar{ns_2}{cs}{ns_1} P$ \isalink{https://github.com/isabelle-utp/utp-main/blob/07cb0c256a90bc347289b5f5d202781b536fc640/theories/circus/utp_circus_parallel.thy\#L958}
\end{theorem}

\noindent Thus, we can commute a parallel composition by also commuting the respective name sets.

\subsection{Composing Reactive Relations}

In order to calculate a stateful-failure reactive design for parallel composition, we specialise
Theorem~\ref{thm:rdespar}. This requires that we have specialised versions of the merge operators for peri- and
postconditions, and also a specialised weakest rely condition operator. These are defined below.

\begin{definition}[Intermediate Merge, Final Merge, and Weakest Rely Condition] \isalink{https://github.com/isabelle-utp/utp-main/blob/07cb0c256a90bc347289b5f5d202781b536fc640/theories/circus/utp_circus_parallel.thy\#L30}
  \vspace{-2ex}

  \begin{center}
  \begin{minipage}{.5\linewidth}
  \begin{align*}
    \textstyle P \imerge{cs} Q &\defs P \parallel_{(\exists \state' @ N_{\textsf{\tiny I}})} Q \\
    \textstyle P \fmerge{ns_1}{cs}{ns_2} Q &\defs P \parallel_{(\exists \refu' @ \cmrgn)} Q \\
    P \wppC{cs} Q &\defs P \wppR{N_{\textsf{\tiny I}}} Q
  \end{align*} 
  \end{minipage}
  where $N_{\textsf{\tiny I}} \defs \cmrg{\lzero}{cs}{\lzero}$
\end{center}
\end{definition}

\noindent The intermediate merge, $\imerge{cs}$ defines how two quiescent observations are merged. It is parametrised
only in $cs$ and not $ns_1$ or $ns_2$, as state is concealed in quiescent observations. Its definition applies the merge
predicate $(\exists \state' @ N_{\textsf{\tiny I}})$, which abstracts from the final state and is defined in terms of
$N_{\textsf{\tiny I}}$. The latter applies $\cmrgn$, but uses the $\lzero$ lens for $ns_1$ and $ns_2$, and therefore
ignores the final state of $P$ and $Q$.

The final state merge $\fmerge{ns_1}{cs}{ns_2}$ defines how terminated observations are merged. It is defined similarly,
but abstracts from $\refu'$, since there is no refusal information in a final observation, and uses $\cmrgn$ directly to
merge the states. Finally, the weakest rely condition $\wppC{cs}$ is simply the general reactive design version, using
$N_{\textsf{\tiny I}}$ as the merge predicate as final states are also not relevant in preconditions. 

We now demonstrate the healthiness of these new operators.

\begin{theorem}[Merge Closure Properties] \label{thm:mergeclos} \isalink{https://github.com/isabelle-utp/utp-main/blob/07cb0c256a90bc347289b5f5d202781b536fc640/theories/circus/utp_circus_parallel.thy\#L86}
  \begin{enumerate}
    \item If $P$ and $Q$ are $\healthy{CRR}$-healthy then $P \imerge{cs} Q$ is $\healthy{CRR}$-healthy.
    \item $P \imerge{cs} Q$ does not refer to $\state'$.
    \item If $P$ and $Q$ are $\healthy{CRF}$-healthy then $P \fmerge{ns_1}{cs}{ns_2} Q$ is $\healthy{CRF}$-healthy.
    \item If $P$ is $\healthy{CRR}$-healthy and $Q$ is $\healthy{CRC}$-healthy, then $P \wppC{cs} Q$ is $\healthy{CRC}$-healthy.
  \end{enumerate}
\end{theorem}
\noindent The intermediate merge constructs a reactive relation that does not refer to the final state. The final merge
constructs a reactive finaliser, since it does not refer to $\refu'$. Weakest rely constructs a reactive
condition. Following a similar approach to Theorem~\ref{thm:nparcomm}, we can also demonstrate commutativity properties.

\begin{theorem}[Inner Merge Commutativity] \isalink{https://github.com/isabelle-utp/utp-main/blob/07cb0c256a90bc347289b5f5d202781b536fc640/theories/circus/utp_circus_parallel.thy\#L181}
  \begin{align*}
    P \imerge{cs} Q &= Q \imerge{cs} P \\
    P \fmerge{ns_1}{cs}{ns_2} Q &= Q \fmerge{ns_2}{cs}{ns_1} P  & ns_1 \lindep ns_2
  \end{align*}
\end{theorem}
\noindent Using these new operators, we can finally prove the specialised calculation law for parallel composition.

\begin{theorem}[Parallel Calculation] \label{thm:nrdespar} \isalink{https://github.com/isabelle-utp/utp-main/blob/07cb0c256a90bc347289b5f5d202781b536fc640/theories/circus/utp_circus_parallel.thy\#L874}
  \begin{align*}
    &\rc{P_1}{P_2}{P_3} \sfpar{ns_1}{cs}{ns_2} \rc{Q_1}{Q_2}{Q_3} = \\[.5ex]
    & \qquad \rc{
      \begin{array}{l} 
        (P_1 \rimplies P_2) \wppC{cs} Q_1 \land \\
        (P_1 \rimplies P_3) \wppC{cs} Q_1 \land \\
        (Q_1 \rimplies Q_2) \wppC{cs} P_1 \land \\
        (Q_1 \rimplies Q_3) \wppC{cs} P_1
      \end{array}}{\begin{array}{l}
        P_2 \imerge{cs} Q_2 ~~\lor~~ \\
        P_3 \imerge{cs} Q_2 ~~\lor~~ \\
        P_2 \imerge{cs} Q_3
      \end{array}}{\textstyle P_3 \fmerge{ns_1}{cs}{ns_2} Q_3}
  \end{align*}
\end{theorem}
\noindent This is similar to Theorem~\ref{thm:rdespar}, but uses the specialised merge and weakest rely operators. This
theorem shows that calculation of reactive contracts can be reduced to merging the peri- and postconditions. A
corollary, for the simpler case when the preconditions are both $\truer$, is given below.

\begin{theorem}[Simplified Parallel Calculation] \label{thm:snrdespar}
  $$\rcs{P_2}{P_3} \sfpar{ns_1}{cs}{ns_2} \rcs{Q_2}{Q_3} = \rcs{P_2 \imerge{cs} Q_2 \lor P_3 \imerge{cs} Q_2 \lor P_2 \imerge{cs} Q_3}{P_3 \fmerge{ns_1}{cs}{ns_2} Q_3}$$
\end{theorem}

\noindent This follows by application of Theorem~\ref{thm:wrlaws} because $P \wppC{cs} \truer = \truer$. We can also
show, with the help of Theorems~\ref{thm:ncsp-intro} and \ref{thm:mergeclos}, that $\healthy{NCSP}$ is closed under
parallel composition.

\begin{theorem}
  If $P$ and $Q$ are $\healthy{NCSP}$-healthy, and $ns_1 \lindep ns_2$, then $P \sfpar{ns_1}{cs}{ns_2} Q$ is
  $\healthy{NCSP}$-healthy. \isalink{https://github.com/isabelle-utp/utp-main/blob/07cb0c256a90bc347289b5f5d202781b536fc640/theories/circus/utp_circus_parallel.thy\#L947}
\end{theorem}
\noindent For sequential processes, we have already shown that the peri- and postconditions of reactive programs can be
specified using disjunctions of the $\mathcal{E}$ and $\Phi$ operators. Consequently, to extend our calculational method
to parallel composition, we need to prove how these operators should be merged. The following theorems show how reactive
relations describing final and intermediate observations are merged.

\begin{theorem}[Merging Finalisers] If $ns_1 \lindep ns_2$ then \label{thm:rrmerge} \isalink{https://github.com/isabelle-utp/utp-main/blob/07cb0c256a90bc347289b5f5d202781b536fc640/theories/circus/utp_circus_parallel.thy\#L556}
  \begin{align*}
    \csppf{s_1}{\sigma_1}{t_1} ~\fmerge{ns_1}{cs}{ns_2}~ \csppf{s_2}{\sigma_2}{t_2} &= \left(\exists t @ \csppf{s_1 \land s_2 \land t \in t_1\!\parallel_{cs}\!t_2 \land t_1 \project cs = t_2 \project cs}{\sigma_1 \usubpar{ns_1}{ns_2} \sigma_2}{t}\right) \\[.5ex]
    \csppf{s_1}{\sigma_1}{\snil} ~\fmerge{ns_1}{cs}{ns_2}~ \csppf{s_2}{\sigma_2}{\snil} &= \csppf{s_1 \land s_2}{\sigma_1 \usubpar{ns_1}{ns_2} \sigma_2}{\snil} 
  \end{align*}

\end{theorem}
\noindent The first equation shows how to merge two finalisers. We require the existence of the trace $t$, which is one
of the possible merges of $t_1$ and $t_2$, and require that both preconditions $s_1$ and $s_2$ hold initially. The
overall trace of the finaliser is then $t$. The second equation is a corollary for when both traces are empty, and the
event merge is trivial. In either case, the final state update is constructed using the operator $\usubpar{ns_1}{ns_2}$,
which uses lens override to merge the two disjoint state updates. It obeys the following laws.

\begin{theorem} \label{thm:subpar} Given independent lenses, $ns_1 \lindep ns_2$, the following identities hold: \isalink{https://github.com/isabelle-utp/utp-main/blob/07cb0c256a90bc347289b5f5d202781b536fc640/utp/utp_subst.thy\#L486}
  \begin{align}
  id \usubpar{ns_1}{ns_2} id &= id \label{law:subpar1} \\
  \sigma \usubpar{ns_1}{ns_2} \rho &= \rho \usubpar{ns_2}{ns_1} \sigma \label{law:subpar2} \\
  (\sigma(x \mapsto v)) \usubpar{ns_1}{ns_2} \rho &= (\sigma \usubpar{ns_1}{ns_2} \rho)(x \mapsto v) & x \lsubseteq ns_1 \label{law:subpar3} \\
  (\sigma(x \mapsto v)) \usubpar{ns_1}{ns_2} \rho &= \sigma \usubpar{ns_1}{ns_2} \rho & x \lindep ns_1 \label{law:subpar4}
  \end{align}
\end{theorem}
\noindent Merging two identity (vacuous) assignments yields an identity assignment \eqref{law:subpar1}. The operator is
quasi-commutative, when the name sets are also swapped \eqref{law:subpar2}. When one of the assignments is constructed
with a state update, if the variable being assigned is part of the corresponding name set ($x \lsubseteq ns_1$), then
the update is applied to the top-level assignment \eqref{law:subpar3}. Effectively, this means that the assignment is
retained when the parallel composition terminates. Conversely, if the assignment is to a variable outside of the name
set ($x \lindep ns_1$), then its effect is lost \eqref{law:subpar4}. Using these laws we calculate the contracts for
some examples.

\begin{example} We assume the existence of lenses $x$ and $y$, with $x \lsubseteq ns_1$, $y \lsubseteq ns_2$, and $ns_1 \lindep ns_2$, and
  calculate the meaning of parallel assignment to these variables. \isalink{https://github.com/isabelle-utp/utp-main/blob/dd8dd58e22ec371bc43b9c9f42be574060310d54/tutorial/utp_csp_ex.thy\#L84}
  \begin{align*}
      &(x := u) \sfpar{ns_1}{cs}{ns_2} (y := v) \\ 
    =~& \rcs{\false}{\csppfs{\substmap{x \mapsto u}}{\snil}} \sfpar{ns_1}{cs}{ns_2} \rcs{\false}{\csppfs{\substmap{y \mapsto v}}{\snil}} & [\ref{thm:bcircus-def}] \\
    =~& \rcs{
          \begin{array}{l}
             \false ~\imerge{cs}~ \false  \\
             \lor \csppfs{\substmap{x \mapsto u}}{\snil} ~\imerge{cs}~ \false \\
             \lor \false ~\imerge{cs} \csppfs{\substmap{y \mapsto v}}{\snil} 
          \end{array}
    }{\csppfs{\substmap{x \mapsto u}}{\snil} \fmerge{ns_1}{cs}{ns_2} \csppfs{\substmap{y \mapsto v}}{\snil}} & [\ref{thm:snrdespar}] \\
    =~& \rcs{\false}{\csppfs{\substmap{x \mapsto u} ~\usubpar{ns_1}{ns_2}~ \substmap{y \mapsto v}}{\snil}} & [\ref{thm:pbm-laws}, \ref{thm:rrmerge}] \\
    =~& \rcs{\false}{\csppfs{\substmap{x \mapsto u, y \mapsto v}}{\snil}} & [\ref{thm:subpar}] \\
    =~& \assignsC{x \mapsto u, y \mapsto v} & [\ref{thm:bcircus-def}]
  \end{align*}
  This does not rule out having $u \defs y$ and $v \defs x$, since both processes have access to the entirety of the
  initial state. We first calculate the contract for the two assignments. Since the preconditions are trivial, we apply
  Theorem~\ref{thm:snrdespar} to compute the composition contract. Since both periconditions are $\false$, by
  application of Theorem~\ref{thm:pbm-laws} and relational calculus, the overall pericondition is also $\false$. Thus,
  we can simply apply Theorem~\ref{thm:rrmerge} to compute the merge of the two finalisers, and then
  Theorem~\ref{thm:subpar} to merge the two assignments. The final form is, by Definition~\ref{thm:bcircus-def},
  equivalent to a single assignment. In Isabelle/UTP, this proof of this equality is fully automated by the
  \textsf{rdes-eq} tactic~\cite{Foster17c}.

  Using a similar calculation, and using Definition~\ref{def:parabbrev}, we can also show that
  $$(x := u) \interleave (y := v) = \Skip$$ The name sets are both $\lzero$ and consequently, since
  $x \lindep \lzero$ and $y \lindep \lzero$, both assignments are lost. \qed
\end{example}

The independence constraints on the process state spaces and loss of assignments may, at first sight, seem
unsatisfactory as this prevents shared variables. However, variables here are only for the sequential case. Shared
variables in languages like CSP and \Circus should be modelled using channel communication, for separation of
concerns. This approach has been demonstrated in several previous works~\cite{Roscoe2007}, including
JCSP~\cite{Vinter2002JCSP} and the RoboChart state-machine language~\cite{Miyazawa2019-RoboChart}.

We next show how quiescent observations are merged using the $\imerge{cs}$.

\begin{theorem}[Merging Quiescent Observations] \label{thm:qmerge} \isalink{https://github.com/isabelle-utp/utp-main/blob/07cb0c256a90bc347289b5f5d202781b536fc640/theories/circus/utp_circus_parallel.thy\#L601}
  \begin{align*}
    \cspen{s_1}{t_1}{E_1} ~\imerge{cs}~ \cspen{s_2}{t_2}{E_2} &= \left(\exists t  @ \cspen{\left(\begin{array}{l} s_1 \land s_2 \land t \in t_1\!\parallel_{cs}\!t_2 \land \\ t_1 \project cs = t_2 \project cs \end{array}\right)}{t}{\left(\!\begin{array}{l} (E_1 \cap E_2 \cap cs) \, \cup \\ ((E_1 \cup E_2) \setminus cs) \end{array}\!\right)}\right) \\[.5ex]
    \cspen{s_1}{t_1}{E_1} ~\imerge{cs}~ \csppf{s_2}{\sigma_2}{t_2} &= \left(\exists t @ \cspen{\left(\begin{array}{l} s_1 \land s_2 \land t \in t_1\!\parallel_{cs}\!t_2 \land \\ t_1 \project cs = t_2 \project cs \end{array}\right)}{t}{E_1 \setminus cs}\right)
  \end{align*}
\end{theorem}
\noindent The equations in Theorem~\ref{thm:qmerge} are similar to those in Theorem~\ref{thm:rrmerge}, but involve at
least one quiescent observation. We omit the symmetric case, since we know that $\imerge{cs}$ is commutative. As for the
finaliser, we need to merge the traces $t_1$ and $t_2$. There is no state update merge, as this is a quiescent
observation. We also need to consider merging the refusal sets together, which here are given by a set of accepted
events. If merging two quiescent observations, we accept (1) the events in the synchronisation set $cs$ enabled by both
$P$ and $Q$ ($E_1 \cap E_2 \cap cs$); and (2) the events not in $cs$ that enabled available in either $P$ or $Q$
($(E_1 \cup E_2) \setminus cs$). If a quiescent observation is merged with a final observation, then the accepted events
are simply those not in $cs$, that is $E_1 \setminus cs$.

In order to illustrate the use of these theorems, we provide the following example.

\setcounter{equation}{0}
\begin{example} \label{ex:commsem}
  We calculate the meaning of $a \then b \then \Skip \sfpare{\{ b \}} b \then c \then \Skip$. \isalink{https://github.com/isabelle-utp/utp-main/blob/dd8dd58e22ec371bc43b9c9f42be574060310d54/tutorial/utp_csp_ex.thy\#L72}
  \begin{align}
    =~& \left. \rc{}{
       \begin{array}{l}
         \cspens{\snil}{\{a\}} \lor \\
         \cspens{\langle a \rangle}{\{b\}}
       \end{array}
    }{\csppfs{id}{\langle a, b \rangle}} \msfpare{\{b\}}
    \rc{}{
       \begin{array}{l}
         \cspens{\snil}{\{b\}} \lor \\
         \cspens{\langle b \rangle}{\{c\}}
       \end{array}
    }{\csppfs{id}{\langle b, c \rangle}}
    \right. \label{eq:pcalc1} \\[.5ex]
    =~& \rc{}
       {
       \begin{array}{l}
         \cspens{\snil}{\{a\}} ~\imerge{\{b\}}~ \cspens{\snil}{\{b\}} \lor \\
         \cspens{\snil}{\{a\}} ~\imerge{\{b\}}~ \cspens{\langle b \rangle}{\{c\}} \lor \\
         \cspens{\snil}{\{a\}} ~\imerge{\{b\}}~ \csppfs{id}{\langle b, c \rangle} \lor \\
         \cspens{\langle a \rangle}{\{b\}} ~\imerge{\{b\}}~ \cspens{\snil}{\{b\}} \lor \\
         \cspens{\langle a \rangle}{\{b\}} ~\imerge{\{b\}}~ \cspens{\langle b \rangle}{\{c\}} \lor \\
         \cdots
       \end{array}
       }{
       \csppfs{id}{\langle a, b \rangle} ~\fmerge{\lzero}{\{b\}}{\lzero}~ \csppfs{id}{\langle b, c \rangle}
       } \label{eq:pcalc2} \\[.5ex]
    =~& \rc{}{
       \!\begin{array}{l}
       \left(\!\exists t @ \cspen{
       \left(\begin{array}{l}
         t \in \snil\!\parallel_{\{\!b\}}\!\snil \land \\
         \snil\!\project\!cs = \snil\!\project\!cs
       \end{array}\!\right)
       }{t}{\!\!
       \begin{array}{l}
         \{a\}\!\cap\!\{b\}\!\cap\! \{b\}\cup \\ 
         (\{a\}\!\cup\!\{b\}) \setminus \{b\}
       \end{array} \!
       }\right) \\
       ~ \lor \cdots
       \end{array} \!\!
       }{\exists t @ \csppf{
       \begin{array}{l}
         t \in \langle a, b \rangle \parallel_{\{b\}} \langle b, c \rangle \land \\ 
         \langle a, b \rangle \project \{b\} = \langle b, c \rangle \project \{b\}
       \end{array}
       }{id}{t}
       } \label{eq:pcalc3} \\[.5ex]
    =~& \rc{}{\cspens{\snil}{\{a\}} \lor \cspens{\langle a \rangle}{\{b\}} \lor \cspens{\langle a, b \rangle}{\{c\}}}{\csppfs{id}{\langle a, b, c \rangle}} \label{eq:pcalc4} \\[.5ex]
    =~& a \then b \then c \then \Skip \label{eq:pcalc5}
  \end{align}
  Step~\eqref{eq:pcalc1} calculates the sequential contracts for the left- and right-hand sides of the parallel
  composition using the rules already outlined in \S\ref{sec:circus-rc}. Step~\eqref{eq:pcalc2} expands out all the possible merges for the peri- and postcondition, by application of
  Theorem~\ref{thm:nrdespar}, and also Theorem~\ref{thm:pbm-laws} to distribute through the various disjunctions. There
  are a total of nine observations (of which we show five) in the pericondition, because we need to merge every disjunct
  of the pericondition, plus the postcondition, with every corresponding disjunct. The majority of these are
  inadmissible and thus reduce to $\false$; for example
  $$\cspens{\snil}{\{a\}} ~\imerge{\{b\}}~ \cspens{\langle b \rangle}{\{c\}} = \false$$ since the $b$ event cannot occur
  independently. Step~\eqref{eq:pcalc3} uses Theorem~\ref{thm:rrmerge} to demonstrate explicitly how to merge the first
  of the nine periconditions, and also the postcondition. For both the peri- and postcondition, we need to find a $t$
  that merges to two traces ($\snil$), and respects the synchronisation order. For the pericondition, there is only one
  such trace, $\snil$, and so this is the one selected. Moreover, it is necessary to calculate the events being
  accepted by appropriately selecting synchronised and non-synchronised events. For the postcondition, there is again
  only one trace, $\langle a, b, c \rangle$. Step~\eqref{eq:pcalc4} calculates all the admissible periconditions, of
  which there are three, and the postcondition. The three possible quiescent observations are (1) nothing has happened,
  and $a$ is accepted; (2) $a$ has occurred, and $b$ is accepted; and (3) $a$ and $b$ have occurred, and $c$ is
  accepted. The postcondition performs no state updates ($id$), and have the total sequence of events. This
  reactive contract is equivalent to the action $a \then b \then c \then \Skip$, as shown in step~\eqref{eq:pcalc5}.

  Though this calculation seems very complicated, the benefit of our theorems and mechanisation is that it can be
  performed automatically in Isabelle/UTP. The \textsf{rdes-eq} tactic can also discover the contract form given in
  step~\eqref{eq:pcalc4} of the proof, though not the final form given in step~\eqref{eq:pcalc5}. In practice, reasoning
  about this kind of example is more easily conducted with the help of higher level algebraic laws, such as
  $$(a \then P) \sfpare{E} (a \then Q) = a \then (P \sfpare{E} Q)   \text{ if } a \in E$$
  the like of which our proof strategy can help to prove.  \qed
\end{example}

In the example given above, the preconditions are always trivial. For non-trivial preconditions, we need laws analogous
to those for the sequential case in Theorem~\ref{thm:evwp}, but for the weakest rely calculus.

\begin{theorem}[Parallel Preconditions] \label{thm:parpre} \isalink{https://github.com/isabelle-utp/utp-main/blob/07cb0c256a90bc347289b5f5d202781b536fc640/theories/circus/utp_circus_parallel.thy\#L737}
  \begin{align*}
    \csppf{s_1}{\sigma_1}{t_1} \wppC{cs} \cspin{s_2}{t_2} &= (\forall tt_0, tt_1 @ \cspin{s_1 \land s_2 \land tt_1 \in (t_2 \cat tt_0) \parallel_{cs} t_1 \land (t_2 \cat tt_0) \project cs = t_1 \project cs}{tt_1}) \\
    \cspen{s_1}{t_1}{E} \wppC{cs} \cspin{s_2}{t_2} &= (\forall tt_0, tt_1 @ \cspin{s_1 \land s_2 \land tt_1 \in (t_2 \cat tt_0) \parallel_{cs} t_1 \land (t_2 \cat tt_0) \project cs = t_1 \project cs}{tt_1})
  \end{align*}
\end{theorem}
\noindent As usual, we need to conjoin both conditions $s_1$ and $s_2$. However, determining permissible traces is
rather more involved. We recall that in $\cspin{s_2}{t_2}$, $t_2$ is a strict upper bound on the permissible
traces. These two laws give the conditions under which a concurrent quiescent or final observation does not allow
divergence; both have the same form. Divergence occurs when $t_1$, a trace contributed by one action, permits $t_2$, a
divergent trace, to be exhausted when the two are merged. This situation occurs for any trace $tt_1$ such that there is
an arbitrary extension $tt_0$, where (1) $tt_1$ is one of the traces obtained by merging $t_2$ extended with $tt_0$,
with $t_1$, and (2) the order of synchronisation of $t_2$ with its extension is the same as that of $t_1$. The trace
$tt_1$ must therefore be a trace that exhausts all the events in $t_2$, whilst respecting both the synchronisation set
and $t_1$. Consequently, it is a strict upper bound on the permissible behaviours.

We exemplify these laws with the calculation below.

\begin{example}
  \allowdisplaybreaks
  \begin{align*}
      & a \then \Chaos \sfpare{\{a\}} a \then \Skip \\
    =~& \rc{\cspin{\ptrue}{\langle a \rangle}}{\cspen{\ptrue}{\langle\rangle}{\{a\}}}{\false} \sfpare{\{a\}} \rcs{\cspens{\snil}{\{a\}}}{\csppfs{id}{\langle a \rangle}} & [\ref{thm:bcircus-def}] \\
    =~& \rc{
        \begin{array}{l}
          \cspens{\snil}{\{a\}} \wppC{\{a\}} \cspin{\ptrue}{\langle a \rangle} \\ 
          \land \csppfs{id}{\langle a \rangle} \wppC{\{a\}} \cspin{\ptrue}{\langle a \rangle}
        \end{array}
        }{\begin{array}{l}
            \cspens{\langle\rangle}{\{a\}} ~\imerge{\{a\}}~ \cspens{\snil}{\{a\}} \\
            \lor \cspen{\ptrue}{\langle\rangle}{\{a\}} ~\imerge{\{a\}}~ \csppfs{id}{\langle a \rangle}
          \end{array}
        }{\false} & [\ref{thm:snrdespar}] \\
    =~& \rc{
        \begin{array}{l}
          \cspens{\snil}{\{a\}} \wppC{\{a\}} \cspin{\ptrue}{\langle a \rangle} \\ 
          \land \csppfs{id}{\langle a \rangle} \wppC{\{a\}} \cspin{\ptrue}{\langle a \rangle}
        \end{array}
        }{\begin{array}{l}
            \cspens{\langle\rangle}{\{a\}} \\
            \lor \false
          \end{array}
        }{\false} & [\ref{thm:qmerge}] \\
    =~& \rc{ \!\!\!
        \begin{array}{l}
          \left(\forall (tt_0, tt_1) @ \cspin{tt_1 \in (\langle a \rangle \cat tt_0) \parallel_{\{a\}} \snil \land (\langle a \rangle \cat tt_0) \project \{a\} = \snil \project cs}{tt_1}\right) \land \\
          \left(\forall (tt_0, tt_1) @ \cspin{tt_1 \in (\langle a \rangle \cat tt_0) \parallel_{\{a\}} \langle a \rangle \land (\langle a \rangle \cat tt_0) \project \{a\} = \langle a \rangle \project \{a\}}{tt_1}\right)
        \end{array} \!\!\!\!
        }{\!\cspens{\langle\rangle}{\{a\}}\!}{\!\false\!} & [\ref{thm:parpre}] \\
    =~& \rc{
        \begin{array}{l}
          \truer \land \\
          (\forall tt_1 @ \cspin{tt_1 \in (\langle a \rangle \cat \snil) \parallel_{\{a\}} \langle a \rangle \land (\langle a \rangle \cat \snil) \project \{a\} = \langle a \rangle \project \{a\}}{tt_1})
        \end{array}
        }{\cspens{\langle\rangle}{\{a\}}}{\false} \\
    =~& \rc{
        \begin{array}{l}
          (\forall tt_1 @ \cspin{tt_1 \in \langle a \rangle \parallel_{\{a\}} \langle a \rangle}{tt_1})
        \end{array}
        }{\cspens{\langle\rangle}{\{a\}}}{\false} \\
    =~& \rc{
        \begin{array}{l}
          (\forall tt_1 @ \cspin{tt_1 = \langle a \rangle}{tt_1})
        \end{array}
        }{\cspens{\langle\rangle}{\{a\}}}{\false} \\
    =~& \rc{
        \begin{array}{l}
          \cspin{\ptrue}{\langle a \rangle}
        \end{array}
        }{\cspens{\langle\rangle}{\{a\}}}{\false} \\
    =~& a \then \Chaos & [\ref{thm:bcircus-def}]
  \end{align*}

  \noindent We calculate the contract for the parallel composition as usual, but in this case it is necessary to
  calculate two weakest rely formulae: (1) $\cspens{\snil}{\{a\}} \wppC{\{a\}} \cspin{\ptrue}{\langle a \rangle}$, and
  (2) $\csppfs{id}{\langle a \rangle} \wppC{\{a\}} \cspin{\ptrue}{\langle a \rangle}$. We expand them both out using
  Theorem~\ref{thm:parpre}. For (1), we observe that the resulting formula has the equation
  $(\langle a \rangle \cat tt_0) \project \{a\} = \snil \project \{a\}$. This is impossible to satisfy, since the
  left-hand trace contains $a$, but the right-hand side does not. Consequently, this term, and therefore the whole
  condition, reduces to $\pfalse$, and so the resulting formula is $(\forall (tt_0, tt_1) @ \cspin{\pfalse}{tt_1})$,
  which, by Theorem~\ref{thm:evwp}, is simply $\truer$. The intuition here is that the empty trace does not permit
  violation of the corresponding precondition. For (2), we notice that there is exactly one possible valuation of $tt_0$
  that satisfies the resulting formula, which is $\snil$. In this case, the precondition can be violated, and $tt_1$
  also has one possible value, which is $\langle a \rangle$. The resulting formula is simply
  $\cspin{\ptrue}{\langle a \rangle}$, and the overall behaviour is equivalent to $a \then \Chaos$. \qed
\end{example}

\subsection{Algebraic Properties}

We now explore the algebraic properties of parallel composition. We show that, under certain conditions characterised by
healthiness conditions, $\Skip$ is a unit for parallel composition, and $\Chaos$ is an annihilator.

In previous work~\cite{Oliveira2005-PHD}, support for this law is provided by an additional healthiness condition, which
imposes downward closure of the refusals. This property is imposed in the standard CSP failures-divergences
model~\cite[Section~8.3]{Roscoe2005} by healthiness condition $\healthy{F2}$. However, not all expressible healthy reactive
relations satisfy this property. For example the relation $\refu' = \{ a, b \}$, which is $\healthy{CRC}$-healthy,
identifies a single refusal and thus forbids the refusal sets $\{a\}$, $\{b\}$, and $\emptyset$, which we would normally
expect to be admissible due to subset closure. We therefore define the following additional healthiness condition for
quiescent observations.

\begin{definition}[Refusal Downward Closure] A reactive relation is downward closed with respect to refusals if it is a fixed-point of healthiness
  condition $\healthy{CDC}$, defined below. \isalink{https://github.com/isabelle-utp/utp-main/blob/07cb0c256a90bc347289b5f5d202781b536fc640/theories/sf_rdes/utp_sfrd_rel.thy\#L1111}
  $$\healthy{CDC}(P) \defs (\exists ref_0 @ P[ref_0/\refu'] \land \refu' \subseteq ref_0)$$
\end{definition}
\noindent A reactive relation $P$ has downward closed refusals if, when we replace $\refu'$ with an arbitrary subset, we
obtain an observation that is still within $P$. It is easy to prove that $\healthy{CDC}$ is idempotent, which follows
due to transitivity of $\subseteq$, and also monotonic. We can also show that $\healthy{CDC}$ is closed under existing
operators.

\begin{theorem}[$\healthy{CDC}$ Closure Properties] \isalink{https://github.com/isabelle-utp/utp-main/blob/07cb0c256a90bc347289b5f5d202781b536fc640/theories/sf_rdes/utp_sfrd_rel.thy\#L1149}
  \begin{itemize}
    \item $\healthy{CDC}$ is closed under the following constructs: $\truer$, $\false$, $\lor$, and $\land$;
    \item If $Q$ is $\healthy{CDC}$-healthy, then $(P \relsemi Q)$ is $\healthy{CDC}$-healthy;
    \item If $\forall i \in I @ P(i) \is \healthy{CDC}$ then $\bigwedge_{i \in I} \, P(i)$ is $\healthy{CDC}$ and $\bigvee_{i \in I} \, P(i)$ is $\healthy{CDC}$;
    \item For any $s$, $t$, and $E$, $\cspen{s}{t}{E}$ is $\healthy{CDC}$-healthy.
  \end{itemize}
\end{theorem}
\noindent $\cspen{s}{t}{E}$ is $\healthy{CDC}$-healthy because, as seen in Definition~\ref{def:rrelop}, we construct the
set of refusals which do not include any event in $E$, a formulation that is downward closed. Consequently, we know that
all the forms of pericondition considered so far, which are disjunctions of $\cspen{s}{t}{E}$ terms, are
$\healthy{CDC}$-healthy. Next, we recast Oliveira's healthiness condition for downward closure, called
$\healthy{C2}$~\cite{Oliveira&09}, to our setting.

\begin{definition} \label{def:C2} $\healthy{C2}(P) \defs P \sfpar{\lone}{\emptyset}{\lzero} \Skip$ \isalink{https://github.com/isabelle-utp/utp-main/blob/07cb0c256a90bc347289b5f5d202781b536fc640/theories/circus/utp_circus_parallel.thy\#L855}
\end{definition}

\noindent $\healthy{C2}$ states that $\Skip$, defined in Definition~\ref{thm:bcircus-def}, is a right unit for the
composition operator $\sfpar{\lone}{\emptyset}{\lzero}$, which takes the entirety of its final state from the left
action, and employs an empty synchronisation set. We now link $\healthy{C2}$ to $\healthy{CDC}$. The proof depends on
two properties of final state merge and weakest rely predicates.

\begin{theorem}[Merge and Weakest Rely of Identity] \label{thm:c2-props} \isalink{https://github.com/isabelle-utp/utp-main/blob/07cb0c256a90bc347289b5f5d202781b536fc640/theories/circus/utp_circus_parallel.thy\#L979}
  \begin{align*}
    P \fmerge{\lone}{\emptyset}{\lzero} \csppfs{id}{\snil} &= P & \text{if $P$ is $\healthy{CRF}$-healthy} \\
    \csppfs{id}{\snil} \wppC{\emptyset} P &= P & \text{if $P$ is $\healthy{CRC}$-healthy}
  \end{align*}
\end{theorem}

\noindent The first property states that merging an arbitrary finaliser $P$ with an identity finaliser, with $P$
contributing all the final state, and an empty synchronisation set, is simply $P$. The second property, similarly,
states that the weakest rely condition that an identity finaliser reaches reactive condition $P$, with an empty
synchronisation set, is simply $P$. We can now prove the following important theorem for $\healthy{C2}$, employing our
calculational strategy, which reveals its intuitive meaning.

\begin{theorem} If $\rc{P_1}{P_2}{P_3} \is \healthy{NCSP}$ then $\healthy{C2}(\rc{P_1}{P_2}{P_3}) = \rc{P_1}{\healthy{CDC}(P_2)}{P_3}$ \isalink{https://github.com/isabelle-utp/utp-main/blob/07cb0c256a90bc347289b5f5d202781b536fc640/theories/circus/utp_circus_parallel.thy\#L1009}
\end{theorem}

\begin{proof}
  \begin{align*}
    \healthy{C2}(\rc{P_1}{P_2}{P_3}) &= \rc{P_1}{P_2}{P_3} \sfpar{\lone}{cs}{\lzero} \rcs{\false}{\csppfs{id}{\snil}} & [\ref{def:C2}, \ref{thm:bcircus-def}] \\
                                    &= \rc{\csppfs{id}{\snil} \wppC{\emptyset} P_1}{P_2 \imerge{cs} \csppfs{id}{\snil}}{P_3 \fmerge{\lone}{\emptyset}{\lzero} \csppfs{id}{\snil}} & [\ref{thm:nrdespar}] \\
                                    &= \rc{P_1}{P_2 \imerge{cs} \csppfs{id}{\snil}}{P_3} & [\ref{thm:c2-props}]\\
                                    &= \rc{P_1}{(\exists ref_0 @ P_2[ref_0/\refu'] \land \refu' \subseteq ref_0)}{P_3} \\
                                    &= \rc{P_1}{\healthy{CDC}(P_2)}{P_3} & \qedhere
  \end{align*}
  
\end{proof}

\noindent This theorem tells us that an $\healthy{NCSP}$-healthy reactive contract (see Theorem~\ref{thm:ncsp-intro}) is
$\healthy{C2}$ when its pericondition is $\healthy{CDC}$. The proof calculates a contract for the parallel composition
with $\Skip$, and then shows that both the precondition and postcondition are unaltered, using
Theorem~\ref{thm:c2-props}. Finally, we show that the pericondition formula $P_2 \imerge{cs} \csppfs{id}{\snil}$ is
equivalent to $\healthy{CDC}(P_2)$, by application of relational calculus. From this theorem, and previous definitions, we
can now prove the following closure theorems for \healthy{C2}.

\begin{theorem}[\healthy{C2} closure properties] \label{thm:c2closure} \isalink{https://github.com/isabelle-utp/utp-main/blob/07cb0c256a90bc347289b5f5d202781b536fc640/theories/circus/utp_circus_parallel.thy\#L1156}
\begin{itemize}
\item $\Miracle$, $\Chaos$, $\Skip$, $\Stop$, $\ckey{Do}(a)$, and $\assignsC{\sigma}$ are all $\healthy{C2}$;
\item If $P$ and $Q$ are both $\healthy{NCSP}$ and $\healthy{C2}$, then $P \relsemi Q$, $\rconditional{P}{b}{Q}$, and $P \extchoice Q$  are all $\healthy{C2}$;
\item If $\forall i \in I @ P(i) \is \healthy{C2}$ then $\bigsqcap_{i \in I} \, P(i)$ is $\healthy{C2}$;
\item If $P$ is $\healthy{PCSP}$ and $\healthy{C2}$ then $\while{b}{P}$ is $\healthy{C2}$;
\item If $ns_1 \lindep ns_2$, and $P$ and $Q$ are both $\healthy{NCSP}$ and $\healthy{C2}$, then $P \sfpar{ns_1}{cs}{ns_2} Q$ is $\healthy{C2}$.
\end{itemize}
\end{theorem}

\noindent We can now prove two algebraic theorems for $\healthy{C2}$ reactive programs.

\begin{theorem} \label{thm:sfparunit}
  If $P$ is $\healthy{NCSP}$ and $\healthy{C2}$ then $P \sfpar{\lone}{cs}{\lzero} \Skip = P$
\end{theorem}

\begin{theorem} \label{thm:interunit}
  If $P$ is $\healthy{NCSP}$ and $\healthy{C2}$, and $\Sigma = \{\emptyset\}$, then $P \interleave \Skip = P$ \isalink{https://github.com/isabelle-utp/utp-main/blob/07cb0c256a90bc347289b5f5d202781b536fc640/theories/circus/utp_circus_parallel.thy\#L1466}
\end{theorem}

\noindent Theorem~\ref{thm:sfparunit} is essentially a restatement of $\healthy{C2}$, that is, $\Skip$ is a right-unit
when $P$ controls the entire state-space. However, we can now use Theorem~\ref{thm:c2closure} to satisfy its provisos,
and thus apply it to programs whose pericondition is $\healthy{CDC}$. Theorem~\ref{thm:interunit} is similar, but has
the additional proviso that the state space $\Sigma$ is unitary, and the state therefore contains no information. This
being the case, $P$ is a process~\cite{Oliveira&09}, to use \Circus terminology, rather than an action, since it has no
visible state updates.

Next, we consider annihilators for parallel composition. We calculate the meaning of $\Chaos \sfpar{ns_1}{cs}{ns_2} P$,
for $\healthy{NCSP}$ healthy reactive contract $P$, using our proof strategy:
\begin{example}{$\Chaos$ parallel composition} \label{ex:par-chaos}
\begin{align*}
  \Chaos \sfpar{ns_1}{cs}{ns_2} P 
  &= \rc{\false}{\false}{\false} \sfpar{ns_1}{cs}{ns_2} \rc{P_1}{P_2}{P_3} \\
  &= \rc{P_2 \wppC{cs} \false \land P_3 \wppC{cs} \false \land \truer \wppC{cs} P_1}{\false}{\false}
\end{align*}
\end{example}
\noindent Due to the definition of $\Chaos$, by Theorem~\ref{thm:nrdespar} the peri- and postcondition reduce to
$\false$. Consequently, to show that $\Chaos$ is an annihilator, it is necessary simply to show that the precondition
reduces to $\false$ in order to complete the reduction to $\Chaos$. We already know by Theorem~\ref{thm:miracle-anhil}
that at least $\Miracle$ does not satisfy this requirement, since it is itself an annihilator for any reactive design,
including $\Chaos$. Consequently, we need to consider constraints under which one of the precondition conjuncts reduce
to $\false$.

The third conjunct, $\truer \wppC{cs} P_1$, in general reduces to $\false$ only when $P_1$ is itself $\false$, and
therefore $P = \Chaos$, which is a trivial and therefore uninteresting case. The first two conjuncts are more interesting,
and correspond to the presence of feasible behaviour by either the peri- or the postcondition. Here, we investigate the
circumstances under which $P_2 \wppC{cs} \false = \false$. For this, we need an additional healthiness condition that
ensures that there is at least one observation in the pericondition.

\begin{definition}[Accepting Actions] \isalink{https://github.com/isabelle-utp/utp-main/blob/07cb0c256a90bc347289b5f5d202781b536fc640/theories/circus/utp_circus_parallel.thy\#L1482}
  \begin{align*}
    \healthy{Accept} &\defs \rc{\truer}{\cspens{\snil}{\textit{Event}}}{\false} \\
    \healthy{CACC}(P) &\defs (P \lor \ckey{Accept})
  \end{align*}
\end{definition}

\noindent $\healthy{Accept}$ is the action that does not terminate, but has a single quiescent observation where nothing
has occurred ($\snil$), and every event in $\textit{Event}$ is accepted. It has a similar form to $\Stop$, except that
the latter accepts no events. $\healthy{Accept}$ accepts every event, and yet no event can ever be added to the
trace. Like $\Miracle$, it is excluded by several of Roscoe's standard failures-divergences healthiness
conditions~\cite[Section~8.3]{Roscoe2005}; in particular it violates $\healthy{F3}$, which requires every enabled event must also appear
in the trace. However, like $\Miracle$, it also possesses interesting theoretical properties. We emphasise that
$\cspens{\snil}{\textit{Event}}$ is both $\healthy{CRR}$ and $\healthy{CDC}$ healthy.

The healthiness condition $\healthy{CACC}$ takes the disjunction of $P$ with $\ckey{Accept}$, which effectively states
that $P$ is refined by $\ckey{Accept}$, and thus sets an upper bound on $P$~\cite{Hoare&98}. Using
Theorem~\ref{thm:rc-comp}, we prove the following calculation for application of $\healthy{CACC}$ to a reactive
contract:
\begin{theorem} \label{thm:cacc-form} $\healthy{CACC}(\rc{P_1}{P_2}{P_3}) = \rc{P_1}{\cspens{\snil}{\textit{Event}} \lor
    P_2}{P_3}$ \isalink{https://github.com/isabelle-utp/utp-main/blob/07cb0c256a90bc347289b5f5d202781b536fc640/theories/circus/utp_circus_parallel.thy\#L1487}
\end{theorem} 
\noindent $\healthy{CACC}$ thus requires that the pericondition refines $\cspens{\snil}{\textit{Event}}$. This means
that the pericondition must admit an observation where nothing has occurred ($\snil$), and also that any subset of
$\textit{Event}$ is accepted (including $\emptyset$). This intuition is demonstrated by the following useful theorem.
\begin{theorem} $\cspen{s_1}{t}{E_1} \refinedby \cspen{s_2}{t}{E_2} \iff (s_1 \implies s_2 \land E_1 \subseteq E_2)$ \isalink{https://github.com/isabelle-utp/utp-main/blob/07cb0c256a90bc347289b5f5d202781b536fc640/theories/sf_rdes/utp_sfrd_rel.thy\#L828}
\end{theorem}
\noindent Refinement of one quiescent observation by another, sharing the same trace, requires that the state condition
is weakened, and that set of enabled events becomes larger. This may seem counter-intuitive, but it is because we
encode refusals in $\refu'$, and therefore the most constrained refusal observation is $\refu' = \emptyset$, which
corresponds to every event being enabled. The majority of productive operators presented so far satisfy this constraint,
and therefore we can prove the following closure properties for $\healthy{CACC}$.

\begin{theorem}[$\healthy{CACC}$ Closure] Let $P$ and $Q$ be $\healthy{NCSP}$-healthy relations, then: \isalink{https://github.com/isabelle-utp/utp-main/blob/07cb0c256a90bc347289b5f5d202781b536fc640/theories/circus/utp_circus_parallel.thy\#L1500}
  \begin{itemize}
    \item $\Chaos$, $\Stop$, and $\ckey{Do}(e)$ are $\healthy{CACC}$-healthy;
    \item If $P$ is $\healthy{CACC}$ then $P \relsemi Q$ is $\healthy{CACC}$;
    \item If $P$ and $Q$ are both $\healthy{CACC}$ then $P \extchoice Q$ is $\healthy{CACC}$;
    \item If $P$ and $Q$ are both $\healthy{CACC}$ then $P \sqcap Q$ is $\healthy{CACC}$.
  \end{itemize}
\end{theorem}
\noindent $\healthy{Miracle}$ is not $\healthy{CACC}$ because its pericondition is $\false$, and therefore does not have
an empty interaction. More importantly, however, $\assignsC{\sigma}$ and $\Skip$ are also not $\healthy{CACC}$, because
they too have a $\false$ pericondition. However, for most processes that include at least one interaction, and do not
invoke infeasible actions like $\Miracle$, $\healthy{CACC}$ is satisfied. In particular, we note that closure of
$\healthy{CACC}$ under sequential composition only requires that the first argument is $\healthy{CACC}$. Therefore,
since by Theorems \ref{thm:sfdl} and \ref{thm:ext-choice-distrib} we can usually push leading assignments forward, most
actions are $\healthy{CACC}$-healthy. Alternatively, we could define a pseudo unit,
$\ckey{NoOp} \defs \rc{\truer}{\cspens{\snil}{\textit{Event}}}{\csppfs{id}{\langle\rangle}}$, but this has the
undesirable characteristic of not refusing any event whilst also not engaging in any event, which also violates the
failures-divergences healthiness condition $\healthy{F3}$~\cite{Roscoe2005}.

With this healthiness condition, we can finally prove the following theorem:

\begin{theorem}
  If $ns_1 \lindep ns_2$, and $P$ is $\healthy{NCSP}$ and $\healthy{CACC}$, then
  $\Chaos \sfpar{ns_1}{cs}{ns_2} P = \Chaos$. \isalink{https://github.com/isabelle-utp/utp-main/blob/07cb0c256a90bc347289b5f5d202781b536fc640/theories/circus/utp_circus_parallel.thy\#L1547}
\end{theorem}

\begin{proof}
  Given that $P = \rc{P_1}{P_2}{P_3}$, and noting the calculation in Example~\ref{ex:par-chaos}, it suffices to show that
  $P_2 \wppC{cs} \false$ reduces to $\false$.
  \begin{align*}
    P_2 \wppC{cs} \false 
    &= (\cspens{\snil}{\textit{Event}} \lor P_2) \wppC{cs} \false & [\ref{thm:cacc-form}] \\
    &= (\cspens{\snil}{\textit{Event}} \wppC{cs} \false) \land (P_2 \wppC{cs} \false)  \\
    &= (\cspens{\snil}{\textit{Event}} \wppC{cs} \cspin{\ptrue}{\snil}) \land (P_2 \wppC{cs} \false) \\
    &= \left(\forall (tt_0, tt_1) @ \cspin{\left(\begin{array}{l} tt_1 \in (\snil \cat tt_0) \parallel_{cs} \snil \land \\ (\snil \cat tt_0) \project cs = \snil \project cs\end{array}\right)}{tt_1}\right) \land (P_2 \wppC{cs} \false) \\
    &= \left(\forall (tt_0, tt_1) @ \cspin{\left(\begin{array}{l} tt_1 \in tt_0 \parallel_{cs} \snil \\ \land tt_0 \project cs = \snil\end{array}\right)}{tt_1}\right) \land (P_2 \wppC{cs} \false) \\
    &= \begin{array}{l} \cspin{\snil \in \snil \parallel_{cs} \snil \land \snil \project cs = \snil}{\snil} \\[1ex] \land \left(\forall (tt_0, tt_1) @ \cspin{\left(\begin{array}{l} tt_1 \in tt_0 \parallel_{cs} \snil \land tt_0 \project cs = \snil\end{array}\right)}{tt_1}\right) \land (P_2 \wppC{cs} \false)
       \end{array} \\
    &= \cspin{\ptrue}{\snil} \land \left(\forall (tt_0, tt_1) @ \cspin{\cdots}{tt_1}\right) \land (P_2 \wppC{cs} \false) \\
    &= \false & \qedhere
  \end{align*}
\end{proof}

\noindent The crucial part of the proof is that the complex $\mathcal{I}$ formula must hold for any given $tt_0$ and
$tt_1$, and so we can pick $\snil$ for both of them, and add this as an extra conjunct. Since $\snil$ merged with
$\snil$ yields $\{\snil\}$, the resulting formula reduces to $\cspens{\ptrue}{\snil}$, which is simply $\false$. This
calculation would not be possible if we could not exhibit $\snil$ as a possible trace in the pericondition, which
is the purpose of $\healthy{CACC}$.

In this section, we have shown how the calculational strategy can be extended to handle parallel composition, and proved
proved some important theorems that follow. In the next section we demonstrate the proof strategy on a small example.

\section{Verification Strategy for Reactive Programs}
\label{sec:verify}
Our results give rise to an automated verification strategy for reactive programs, whereby we (1)~calculate the contract
of a reactive program, (2)~use our equational theory to simplify the underlying reactive relations, (3)~identify
invariants for reactive while loops, and (4) finally prove refinements using relational calculus. Although the relations
can be complex, our equational theory from \S\ref{sec:circus-rc} and \S\ref{sec:ext-choice}, aided by the Isabelle/HOL
simplifier, can be used to rapidly reduce them to more compact forms amenable to automated proof. In this section we
illustrate this strategy using the buffer in Example~\ref{ex:buffer}. We prove two properties: (1) deadlock freedom, and
(2) the order of values produced is the same as those consumed.

Deadlock freedom can be demonstrated with the help of the following
specification contract~\cite{Foster17c}.

\begin{definition}[Deadlock-freedom Contract] $\ckey{CDF} \defs \textstyle\rcs{\exists s, t, E, e @ \cspen{s}{t}{\{e\} \cup E}}{\truer}$ \isalink{https://github.com/isabelle-utp/utp-main/blob/90ec1d65d63e91a69fbfeeafe69bd7d67f753a47/theories/sf_rdes/utp_sfrd_fdsem.thy\#L425}
\end{definition}

\noindent Since only quiescent observations can deadlock, $\ckey{CDF}$ constrains only the pericondition, which
characterises observations where at least one event $e$ is being accepted: there is no deadlock. For example, we can
show that $a \then b \then \Skip \sfpare{\{ b \}} b \then c \then \Skip$ is deadlock-free with the help of
Example~\ref{ex:commsem}.

\begin{example}[Deadlock-Freedom Calculation]
  \begin{align*}
    &\ckey{CDF} \refinedby a \then b \then \Skip \sfpare{\{ b \}} b \then c \then \Skip \\
    \iff\,& \rcs{\exists s, t, E, e @ \cspen{s}{t}{\{e\} \cup E}}{\truer} \refinedby \rcs{\cspens{\snil}{\{a\}} \lor \cspens{\langle a \rangle}{\{b\}} \lor \cspens{\langle a, b \rangle}{\{c\}}}{\csppfs{id}{\langle a, b, c \rangle}} \\
    \iff\,& (\exists s, t, E, e @ \cspen{s}{t}{\{e\} \cup E}) \refinedby (\cspens{\snil}{\{a\}} \lor \cspens{\langle a \rangle}{\{b\}} \lor \cspens{\langle a, b \rangle}{\{c\}}) \land \truer \refinedby \csppfs{id}{\langle a, b, c \rangle} \\
    \iff\,& (\exists s, t, E, e @ \cspen{s}{t}{\{e\} \cup E}) \refinedby \cspens{\snil}{\{a\}} \land (\exists s, t, E, e @ \cspen{s}{t}{\{e\} \cup E}) \refinedby \cspens{\langle a \rangle}{\{b\}} \land \cdots \\
    \iff\,& \true
  \end{align*}

  \noindent The intuition is that every $\cspen{\cdot}{\cdot}{\cdot}$ term in the process's pericondition corresponds to
  a possible transition. Consequently, we need to show that no transition exists without an enabled event. This is the
  case for all three disjuncts --- they enable $\{a\}$, $\{b\}$ and $\{c\}$, respectively --- and so the process is
  deadlock-free. \qed
\end{example}

To prove that the buffer is deadlock-free, we first calculate the contract of the main loop in the $Buffer$ process in
Example~\ref{ex:buffer}, and then use this to calculate the overall contract for the iterative behaviour.

\begin{theorem}[Loop Body]
  The body of the loop is $\rc{\truer}{B_2}{B_3}$ where \isalink{https://github.com/isabelle-utp/utp-main/blob/07cb0c256a90bc347289b5f5d202781b536fc640/tutorial/utp_csp_buffer.thy\#L59}
  \begin{align*}
    B_2 =~~& \textstyle\cspen{true}{\langle\rangle}{\bigcup_{v \in \nat}\, \{inp.v\} \cup (\rconditional{\{out.head(bf)\}}{0 < \#bf}{\emptyset})} \\[1ex]
    B_3 =~~& \left(
      \begin{array}{l}
        \textstyle
        \left(\bigvee_{v \in \nat}\, \csppf{true}{\{bf \mapsto bf \cat \langle v \rangle\}}{\langle inp.v \rangle} \right) \lor \\[1ex]
        \csppf{0 < \#bf}{\{bf \mapsto tail(bf)\}}{\langle out.head(bf) \rangle}
      \end{array}
    \right)
  \end{align*}
  \noindent \textnormal{The $\truer$ precondition implies no divergence. The pericondition states that every input event
    is enabled, and the output event is enabled if the buffer is non-empty. The postcondition contains two
    possible final observations: (1) an input event occurred and the buffer variable was extended; or (2)
    provided the buffer was non-empty initially, then the buffer's head is output and $bf$ is contracted.}
\end{theorem}
\begin{proof} To exemplify, we calculate the left-hand side of the choice, employing Theorems~\ref{thm:rc-comp},
  \ref{thm:bcircus-def}, \ref{thm:crel-comp}, and \ref{thm:filtering}. The entire calculation is automated in Isabelle/UTP.
  \begin{align*}
    & inp?v \then bf := bf \cat \langle v \rangle \\
    & = \Extchoice v\!\in\!\nat @ \ckey{Do}(inp.v) \relsemi bf := bf \cat \langle v \rangle & [\textnormal{Defs}]\\ 
    & = \Extchoice v\!\in\!\nat @ \left(
      \begin{array}{l}
        \rc{\truer}{\cspen{\ptrue}{\langle\rangle}{\{inp.v\}}}{\csppf{\ptrue}{id}{\langle inp.v \rangle}} \relsemi \\ 
        \rc{\truer}{\false}{\csppf{\ptrue}{\llparenthesis bf \mapsto bf \cat \langle v \rangle \rrparenthesis}{\langle\rangle}}
      \end{array} \right) & [\ref{thm:bcircus-def}] \\[.5ex]
    & = \Extchoice v\!\in\!\nat @
        \rc{\!\truer\!}{\!
        \begin{array}{l}
          \cspen{\ptrue}{\langle\rangle}{\{inp.v\}} \\
          \lor \false
        \end{array} \!
      }{\!
        \begin{array}{l}
          \csppf{\ptrue}{id}{\langle inp.v \rangle} \relsemi \\
          \csppf{\ptrue}{\llparenthesis bf\!\mapsto\!bf\!\cat\!\langle v \rangle \rrparenthesis}{\langle\rangle}
        \end{array}
      } & [\ref{thm:rc-comp}, \ref{thm:evwp}] \\[1ex]
    & = \Extchoice\! v\!\in\!\nat @\! \rc{\!\truer\!}{\!\cspen{\ptrue}{\snil}{\{\!inp.v\!\}}\!}{\!\csppf{\ptrue}{\llparenthesis bf\!\mapsto\!bf\!\cat\!\langle v \rangle \rrparenthesis}{\langle inp.v \rangle}\!} & [\ref{thm:crel-comp}] \\[1ex]
    & = \rc{\!\truer\!}{\!\cspen{\ptrue}{\snil}{\!\bigcup_{v \in \nat} \{\!inp.v\!\}}\!\!}{\!\bigvee_{v \in \nat} \!\csppf{\ptrue}{\llparenthesis bf \!\mapsto\! bf\!\cat\!\langle v \rangle \rrparenthesis}{\langle inp.v \rangle}\!} & [\ref{def:ext-choice},\ref{thm:filtering}]
  \end{align*}

  \noindent Though this calculation seems complicated, in practice it is fully automated and a user need not be
  concerned with these minute calculational details, but can rather focus on finding suitable reactive invariants.
\end{proof}

Then, by Theorem~\ref{thm:whilecalc} we can calculate the overall behaviour of the buffer.
$$Buffer = \rc{\truer}{\csppf{true}{\{bf \mapsto \langle\rangle\}}{\langle\rangle} \relsemi B_3\bm{\star} \relsemi
  B_2}{\false}$$

\noindent This is a non-terminating contract where every quiescent behaviour begins with an empty buffer, performs some
sequence of buffer inputs and outputs accompanied by state updates ($B_3\bm{\star}$), and is finally offering the
relevant input and output events ($B_2$). We can now employ Theorem~\ref{thm:rea-inv} to verify the buffer. First, we
tackle deadlock freedom, which can be proved using the following refinement.

\begin{theorem}[Deadlock Freedom]  $\ckey{CDF} \refinedby Buffer$ \isalink{https://github.com/isabelle-utp/utp-main/blob/07cb0c256a90bc347289b5f5d202781b536fc640/tutorial/utp_csp_buffer.thy\#L118}
\end{theorem}

\noindent This theorem can be discharged automatically in 1.8s on an Intel i7-4790 desktop machine. This proof approach
has also been applied in demonstrating that formalised state machine models, in the RoboChart
language~\cite{Miyazawa2019-RoboChart}, are deadlock-free~\cite{Foster18b} with a similar level of automation. We next
tackle the second property.

\begin{theorem}[Buffer Order Property] The sequence of items output is a prefix of those that were previously
  input. This can be formally expressed as \isalink{https://github.com/isabelle-utp/utp-main/blob/07cb0c256a90bc347289b5f5d202781b536fc640/tutorial/utp_csp_buffer.thy\#L100}
  $$\rc{\truer}{outps(\trace) \le inps(\trace)}{\false} \refinedby Buffer$$ where $inps(t), outps(t) : \seq\,\nat$ extract the sequence of
  input and output elements from the trace $t$, respectively. The postcondition simply requires that $Buffer$ does not terminate.
\end{theorem}
  
\begin{proof}
  First, we identify the reactive invariant $I \defs outps(\trace) \le bf \cat inps(\trace)$, and show that
  $\rc{\!\truer\!}{\!I\!}{\!\false\!} \refinedby \while{\truer}{\rc{\!\truer\!}{\!B_2\!}{\!B_3\!}}$. By
  Theorem~\ref{thm:rea-inv} it suffices to show case (2), that is $I \refinedby B_2$ and $I \refinedby B_3 \relsemi I$,
  as the other two cases are vacuous. These two properties can be discharged by relational calculus. Second, we prove
  that
  $\rc{\!\truer\!}{outps(\trace) \le inps(\trace)}{\!\false\!} \refinedby bf := \langle\rangle \relsemi
  \rc{\!\truer\!}{\!I\!}{\!\false\!}$.
  This holds, by Theorem~\thmeqref{thm:sfdl}{thm:asgdist}, since
  $I[\langle\rangle/bf] = outps(\trace) \le inps(\trace)$.  Thus, the overall theorem holds by monotonicity of
  $\relsemi$ and transitivity of $\refinedby$. The proof is semi-automatic --- since we have to manually apply induction
  with Theorem~\ref{thm:rea-inv} --- with the individual proof steps taking 2.2s in total.

\end{proof}

\section{Conclusion}
\label{sec:concl}
We have demonstrated an effective verification strategy for concurrent and reactive programs employing reactive
relations and Kleene algebra. We have provided three novel operators for expressing pre-, peri-, and postconditions in
stateful-failure reactive contacts, and shown how they can be used to support automated verification through
calculation. We have defined a number of novel UTP healthiness conditions for both reactive relations and reactive
contracts, that capture important properties needed by the verification strategy and algebraic laws. Our theory supports
most of the operators of the \Circus language, including all the sequential operators from~\cite{Oliveira&09}, and also
parallel composition. Our theorems and verification tool can be found in our theory repository\footnote{Isabelle/UTP:
  \url{https://github.com/isabelle-utp/utp-main}}, together with companion proofs for the theorems presented here.

Related work includes the works of Struth \emph{et al.} on verification of imperative
programs~\cite{Armstrong2015,Gomes2016} using Kleene algebra for verification-condition generation, which our work
heavily draws upon to deal with iteration. Automated proof support for the failures-divergences model was previously
provided by the CSP-Prover tool~\cite{Isobe2008}, which can be used to verify infinite-state systems in CSP with
Isabelle. Our work is different both in its contractual semantics, and also in our explicit handling of state, which
allows us to express variable assignments. However, we believe that several of the proof tactics defined for
CSP-Prover~\cite{Isobe2008} could be applicable in our work for a restricted subset of reactive programs that model CSP
processes. More recently Taha, Wolff, and, Ye~\cite{Taha2020CSP-Isabelle} have mechanised the set-based
failures-divergences semantics~\cite{Roscoe2005} of CSP in Isabelle/HOL. With suitable semantic links to our
mechanisation, this work similarly has the potential to improve automation.

Our work lies within the ``design-by-contract'' field~\cite{Meyer92}, and is related to the
  assume-guarantee reasoning frameworks~\cite{Benveniste2007,Benvenuti2008,Vincentelli2012}; a detailed comparison can
  be found in~\cite{Foster17c}. The refinement calculus of reactive systems~\cite{Preoteasa2014} is a language
based on property transformers containing trace information. Like our work, they support verification of reactive
systems that are nondeterministic, non-input-receptive, and infinite state. The main differences are our handling of
state variables, the basis in relational calculus, and our failures-divergences semantics. Nevertheless, our contract
framework~\cite{Foster17c} can be linked to those results, and we plan to derive an assume-guarantee calculus to support
verification of multi-party concurrent systems.

In future work, we will further optimise proof support for parallel composition through mechanisation of Oliveira's
refinement and step laws~\cite{Oliveira2005-PHD}, which allow efficient proof for concurrency patterns like bulk
synchronous parallelism. Furthermore, we will investigate the algebraic properties of the weakest rely operator
($\wppR{M}$), and its relationship with rely-guarantee algebra~\cite{Hayes2016a}. We will also tackle the remaining
operators of the \Circus language~\cite{Woodcock2001-Circus}, including hiding and renaming. Moreover, we hope to
identify a normal form for stateful-failure reactive designs using our specialised operators, which we speculate may
have the following approximate form:
$$\rc{\bigwedge_{i \in I} \, \cspin{b_1(i)}{t_1(i)}}{\bigvee_{i \in J} \, \cspen{b_2(i)}{t_2(i)}{E(i)}}{\bigvee_{i \in K} \, \csppf{b_3(i)}{\sigma(i)}{t_3(i)}}$$
This contains a conjunction of trace assumptions, a disjunction of quiescent observation, and a disjunction of
finalisers. It may well be the case that additional healthiness conditions will be required for this. We therefore will
also explore additional properties that the healthiness conditions \healthy{C2} and \healthy{CACC} support. We will
endeavour to establish formal links, using Galois connections, to existing semantic models like the original
failure-divergences model of CSP and its healthiness conditions~\cite{Roscoe2005,Cavalcanti&06}. This could provide a
way of harnessing CSP-Prover proof tactics~\cite{Isobe2008}, and therefore expand our verification capabilities.

We also aim to apply our strategy to more substantial examples, and are currently using it to build a prototype tactic
for verifying robotic controllers using a statechart-style notion with a mechanised denotational
semantics~\cite{Miyazawa2019-RoboChart,Foster18b}. To support this, we will develop a \Circus-based intermediate
verification language with annotations, such as loop invariants, to provide greater automation. Further in this direction,
our semantics and techniques will be also be extended to cater for real-time, probabilistic, and hybrid computational
behaviours~\cite{Foster17b}, which is possible due to the parametric nature of our reactive contract theory.

\section*{Acknowledgments}

\noindent This research is funded by the CyPhyAssure project\footnote{CyPhyAssure Project:
  \url{https://www.cs.york.ac.uk/circus/CyPhyAssure/}}, EPSRC grant EP/S001190/1, the RoboCalc project\footnote{RoboCalc
  Project: \url{https://www.cs.york.ac.uk/circus/RoboCalc/}}, EPSRC grant EP/M025756/1, and the Royal Academy of
Engineering. We would like to thank the anonymous reviewers of our article for their diligent and helpful feedback, which has
greatly improved the presentation of our results.

\bibliographystyle{plain}
\bibliography{JLAMP}

\end{document}